\documentclass[11pt, oneside]{article}
\usepackage{amsthm}
\usepackage{amsmath}

\usepackage{letltxmacro}                         
\LetLtxMacro\amsproof\proof                      
\LetLtxMacro\amsendproof\endproof                
\usepackage{thmtools}                            %
\AtBeginDocument{
  \LetLtxMacro\proof\amsproof                    %
  \LetLtxMacro\endproof\amsendproof              %
}                                                %

\usepackage{thm-restate}
\usepackage[colorlinks=true, allcolors=blue]{hyperref}
\usepackage{
  hyperref,
  cleveref,
}
\usepackage[dvipsnames]{xcolor} 

\newcommand{\MyThmtoolsConstructor}[3]{
  \declaretheorem[
  numberlike=COUNTERHACK,
  numbered=yes,
  name=#2,
  Refname={#2}, 
  refname={\MakeLowercase{#2}},  
  ]{#1}
  \declaretheorem[
  numbered=no,
  name=Informal #2,
  Refname={#2}, 
  refname={\MakeLowercase{#2}},  
  ]{#1-informal}
  \declaretheorem[
  numbered=no,
  name=#2,
  Refname={#2}, 
  refname={\MakeLowercase{#2}},  
  ]{#1*}
}

\MyThmtoolsConstructor{definition}   {Definition}     {NavyBlue} 
\MyThmtoolsConstructor{theorem}      {Theorem}        {WildStrawberry}
\MyThmtoolsConstructor{corollary}    {Corollary}      {Thistle}
\MyThmtoolsConstructor{porism}       {Porism}         {Cyan}
\MyThmtoolsConstructor{lemma}        {Lemma}          {Cyan}
\MyThmtoolsConstructor{proposition}  {Proposition}    {teal}
\MyThmtoolsConstructor{example}      {Example}        {green}
\MyThmtoolsConstructor{comment}      {Comment}        {White}
\MyThmtoolsConstructor{exercise}     {Exercise}       {JungleGreen}
\MyThmtoolsConstructor{conjecture}   {Conjecture}     {Yellow}
\MyThmtoolsConstructor{reflection}   {Reflection}     {Purple}
\MyThmtoolsConstructor{question}     {Question}       {lime}
\MyThmtoolsConstructor{observation} {Observation}    {Gray}
\MyThmtoolsConstructor{claim}        {Claim}          {Gray}
\MyThmtoolsConstructor{fact}         {Fact}           {Gray}
\MyThmtoolsConstructor{note}         {Note}           {Gray}
\MyThmtoolsConstructor{remark}       {Remark}         {Gray}
\MyThmtoolsConstructor{notation}     {Notation}       {Red}
\MyThmtoolsConstructor{summary}      {Summary Result} {WildStrawberry}
\MyThmtoolsConstructor{problem}      {Problem}        {WildStrawberry}

\usepackage[dvipsnames]{xcolor}
\usepackage{amsmath}
\usepackage{amsthm}
\usepackage{amssymb}
\usepackage{changepage} 
\usepackage{enumitem}
\setlist{nolistsep} 
\usepackage{verbatim}
\usepackage{color}
\usepackage{geometry}
\usepackage{graphicx}
\usepackage{ifthen}
\usepackage{mathtools}
\usepackage{mathdots}
\usepackage{lipsum}
\usepackage[many]{tcolorbox}
\usepackage{complexity}
\usepackage{hyperref}
\usepackage{tikz-cd} 
\usepackage{setspace}
\usepackage{algorithm}
\usepackage[noend]{algpseudocode}
\usepackage{thmtools, thm-restate}
\usepackage{todonotes}
\usepackage{subcaption}
\usepackage{faktor}

\newcounter{mycomment}

\newcounter{proposed}

\newboolean{print} 
\setboolean{print}{true}

\newfunc{\exponential}{exp}

\newlang{\langINDEX}{INDEX}
\newlang{\langLENGTHGEQ}{LENGTHGEQ}
\newlang{\langtemp}{A}          
\newlang{\langREACH}{REACH}

\newcommand{\set}[1]{\left\{#1\right\}}

\newcommand{\bitst}{\{0,1\}^{*}} 
\newcommand{\nbits}{\{0,1\}^{n}} 
\newcommand{\closure}{\overline}
\newcommand{\clos}{\overline}

\newcommand{\N}{\mathbb{N}}     
\renewcommand{\R}{\mathbb{R}}   
\newcommand{\Z}{\mathbb{Z}}     
\newcommand{\B}{\mathbb{B}}     

\newcommand{\pr}{\mathbb{P}}   
\newcommand{\ex}{\mathbb{E}}   

\renewcommand{\S}{\mathcal{S}}
\renewcommand{\P}{\mathcal{P}}

\DeclareMathOperator{\range}{range}
\DeclareMathOperator{\defeq}{\overset{def}{=}}
\DeclareMathOperator{\adj}{\overset{adj}{\sim}}
\DeclareMathOperator{\rep}{rep}
\DeclareMathOperator{\sign}{sign}

\DeclareMathOperator{\mycolor}{color}
\DeclareMathOperator{\diam}{diam}
\DeclareMathOperator{\sperner}{Sperner}
\DeclareMathOperator{\member}{member}

\DeclarePairedDelimiter{\ceil}{\lceil}{\rceil}
\DeclarePairedDelimiter\floor{\lfloor}{\rfloor}

\DeclarePairedDelimiter\abs{\lvert}{\rvert}

\DeclarePairedDelimiter\norm{\lVert}{\rVert}

\newcommand{\mytcbtheoremconstructor}[4]{

  \ifthenelse{\boolean{print}}{
    \newtcbtheorem[use counter=boxcounter, number within=section]{#1}{#2}{
      boxrule=0.7mm,
      skin=enhanced,
      middle=3mm,
      breakable, 
      colback=white,
      colframe=#4!50!black,
      fonttitle=\bfseries
    }{#3}
  }{
    \newtcbtheorem[use counter=boxcounter, number within=section]{#1}{#2}{
      boxrule=0.7mm,
      skin=enhanced,
      middle=3mm,
      breakable, 
      bicolor,
      colback=#4!20,
      colbacklower=white,
      colframe=#4!50!black,
      fonttitle=\bfseries,
    }{#3}
  }

}

\usepackage{authblk}

\makeatletter
\newcommand\email[2][]%
  {\newaffiltrue\let\AB@blk@and\AB@pand
      \if\relax#1\relax\def\AB@note{\AB@thenote}\else\def\AB@note{\relax}%
        \setcounter{Maxaffil}{0}\fi
      \begingroup
        \let\protect\@unexpandable@protect
        \def\thanks{\protect\thanks}\def\footnote{\protect\footnote}%
        \@temptokena=\expandafter{\AB@authors}%
        {\def\\{\protect\\\protect\Affilfont}\xdef\AB@temp{\url{#2}}}%
         \xdef\AB@authors{\the\@temptokena\AB@las\AB@au@str
         \protect\\[\affilsep]\protect\Affilfont\AB@temp}%
         \gdef\AB@las{}\gdef\AB@au@str{}%
        {\def\\{, \ignorespaces}\xdef\AB@temp{\url{#2}}}%
        \@temptokena=\expandafter{\AB@affillist}%
        \xdef\AB@affillist{\the\@temptokena \AB@affilsep
          \AB@affilnote{}\protect\Affilfont\AB@temp}%
      \endgroup
      \let\AB@affilsep\AB@affilsepx
}
\makeatother

\renewcommand\Affilfont{\itshape\small}
\author{Jason Vander Woude}
\affil{School of Computing and Department of Mathematics, University of Nebraska-Lincoln}
\email{jasonvw@huskers.unl.edu}

\author{Peter Dixon\thanks{Part of the work done while the author was at Iowa State University}}
\affil{Ben-Gurion University of the Negev}
\email{tooplark@gmail.com}

\author{A. Pavan}
\affil{Department of Computer Science, Iowa State University}
\email{pavan@cs.iastate.edu}

\author{Jamie Radcliffe}
\affil{Department of Mathematics, University of Nebraska-Lincoln}
\email{aradcliffe1@unl.edu}

\author{ N. V. Vinodchandran}
\affil{School of Computing, University of Nebraska-Lincoln}
\email{vinod@cse.unl.edu}

\title{The Geometry of Rounding
\thanks{Research supported in part by NSF grant 1934884, 2130536, 2130608}
}
\date{\today}

\begin{document}
\maketitle

\renewcommand{\P}{\mathcal{P}}

\begin{abstract}
Rounding has proven to be a fundamental tool in theoretical computer science. By observing that rounding and partitioning of $\R^d$ are equivalent, we introduce the following natural partition problem which we call the {\em secluded hypercube partition problem}: Given $k\in\N$ (ideally small) and $\epsilon>0$ (ideally large), is there a partition of $\R^d$ with unit hypercubes such that for every point $\vec{p} \in \R^d$, its closed $\epsilon$-neighborhood (in the $\ell_{\infty}$  norm) intersects at most $k$ hypercubes? 

We undertake a comprehensive study of this partition problem. We prove that for every $d\in\N$, there is an explicit (and efficiently computable) hypercube partition of $\R^d$ with $k = d+1$ and $\epsilon = \frac{1}{2d}$. We complement this construction by proving that the value of $k=d+1$ is the best possible (for any $\epsilon$) for a broad class of ``reasonable'' partitions including hypercube partitions. We also investigate the optimality of the parameter $\epsilon$ and prove that any partition in this broad class that has $k=d+1$, must have $\epsilon\leq\frac{1}{2\sqrt{d}}$. These bounds imply limitations of certain deterministic rounding schemes existing in the literature. Furthermore, this general bound is based on the currently known lower bounds for the dissection number of the cube, and improvements to this bound will yield improvements to our bounds. 

While our work is motivated by the desire to understand rounding algorithms, one of our main conceptual contributions is the introduction of the {\em secluded hypercube partition problem}, which fits well with a long history of investigations by mathematicians on various hypercube partitions/tilings of Euclidean  space.

\end{abstract}


\pagenumbering{roman}
\newpage
\tableofcontents
\newpage
\pagenumbering{arabic}

\section{Introduction}
\label{sec:introduction}

Rounding has proven to be  a fundamental tool in theoretical computer science. Generically, if $\mathcal{X}\subseteq\mathcal{Y}$ are metric spaces with the same metric, then a rounding algorithm $\mathcal{R}$ maps points in $\mathcal{X}$ to points in $\mathcal{Y}$ (typically so that $x$ is close to $\mathcal{R}(x)$). Intuitively, the purpose of rounding is to limit the number of possible outcomes of the algorithm, which in turn can help reduce the complexity (e.g. space complexity, time complexity, or the number of random bits needed). Often $\mathcal{X}=\R^d$ with the metric  induced by the $l^\infty$ norm (we call this metric $d_{max}$).




In complexity theory, one of the earliest applications of rounding is in the seminal work of Saks and Zhou~\cite{SaksZhou99} in the context of space-bounded derandomization. They devised a randomized algorithm $\mathcal{R}$ rounding from $[0, 1]^d$ to $[0, 1]^d$  with the following property:
for every $\vec{x} \in [0,1]^d$,
\[
  \Pr_{r \in \Sigma^{\ell}}\big[\forall \vec{y}\in B_{\epsilon}(\vec{x}): \mathcal{R}(r, \vec{x}) = \mathcal{R}(r,\vec{y})\big] \geq 1- 1/\poly(d)
\]
and for any $r\in\Sigma^\ell$, $d_{max}\left(\vec{x},\mathcal{R}(r,\vec{x})\right)<\epsilon$.

Here, $B_{\epsilon}(\vec{x})$ is the open $\epsilon$-ball around $\vec{x}$ with respect to the $l^\infty$ norm and $\Sigma^\ell=\set{0,1}^\ell$ is the  sample space of random bits. The property above expresses that for any $\epsilon$-ball, with high probability over the choice of randomness, all points in the ball round to the same value. The Saks and Zhou rounding scheme is a randomized rounding scheme with $\epsilon = 1/\poly(d)$ using $\ell = O(\log d)$ bits of randomness. This randomized rounding is a critical step in the well known derandomization of probabilistic space-bounded algorithms, namely $\BPSPACE(S)\subseteq \DSPACE(S^{3/2})$. 

A more recent example in which randomized rounding is employed is in the context of pseudodeterminism and multi-pseudodeterminism~
\cite{gat_probabilistic_2011,Goldreich19,GrossmanLiu19}. A pseudodetermistic algorithm is a randomized algorithm that on any particular input, returns a canonical solution with probability at least $2/3$ (which can be boosted to $1-\delta$). A $k$-pseudodeterministic algorithm, a generalization of a pseudodeterministic algorithm, is a probabilistic algorithm such that for each input $x$, there is a set $S_x$ of cardinality at most $k$, and on input $x$, the algorithm returns a value in $S_x$ with probability at least $\frac{k+1}{k+2}$ (which again can be boosted to $1-\delta$). In the context of designing multi-pseudodeterministic algorithms, Goldreich~\cite{Goldreich19} designed a randomized algorithm $\mathcal{R}$ with the following property:
for every $\vec{x}\in\R^d$, there exists $S_{\vec{x}}$ of cardinality at most $k = d+1$ such that
\[
\Pr_{r \in \Sigma^\ell}\left[
    \forall\vec{y}\in B_\epsilon(\vec{x}) :\mathcal{R}(r,\vec{y})\in S_{\vec{x}}
\right] \geq \frac{k+1}{k+2}
\]
and for any $r\in\Sigma^\ell$, $d_{max}\left(\vec{x},\mathcal{R}(r,\vec{x})\right)<\epsilon$.

Goldreich's rounding  has a different requirement than Saks and Zhou's: for any $\epsilon$-ball, Goldreich does not require that with high probability all points in the ball round to the same value, but he does require that there is a small set such that with high probability all points in the ball round to points in that set. Also, in Saks and Zhou's scheme, the single value that the ball is rounded to can depend on random bits $r$, but Goldreich's scheme requires the set $S_{\vec{x}}$ to be independent of the randomness. Using this rounding scheme, Goldreich showed the existence of $(d+1)$-pseudodeterministic approximation algorithms for a class of functions whose range is $\R^d$. Grossman and Liu used a similar rounding scheme to design low {\em influential-bit} algorithms, a notion that generalizes pseudodeterministic algorithms~\cite{GrossmanLiu19}. Similar rounding schemes have found applications in a very recent work of Impagliazzo, Lei, Pitassi, and Sorrell~\cite{impagliazzo_reproducibility_2022} in the context of making statistical learning algorithms reproducible. 



In~\cite{hoza_preserving_2018}, Hoza and Klivans designed a certain {\em deterministic} rounding scheme. Their work is motivated by the problem of reducing randomness for adaptive algorithms. Consider the problem of simulating  $m$ adaptive (adversarially chosen) invocations of a randomized estimation algorithm that returns a vector in $\R^d$.  If the algorithm requires $n$ bits of randomness, then the trivial way of simulating $m$ invocations requires $nm$ random bits. To reduce the randomness, they designed a deterministic rounding algorithm $\mathcal{R}$ that has the same properties as that of Goldreich. In other words, for every $\vec{x}\in\R^d$, the size of the set $\set{\mathcal{R}(\vec{y}):\vec{y}\in B_{\epsilon}(\vec{x})}$ has cardinality at most $d+1$. Combining this rounding algorithm with the INW pseudorandom generator~\cite{INW}, Hoza and Klivans reduced the amount of randomness required for the above adaptive invocation problem to $n+O(m\log(d+1))$. 
We give more elaborate details on all these rounding schemes and how they fit into our work in \Autoref{sec:rounding-schemes-in-prior-work}.



The present work investigates {\em the geometry of rounding}.  We equate the notion of  deterministic rounding schemes to partitions of Euclidean space $\R^{d}$. This connection led us to the introduction of a very natural partition problem that we call the {\em secluded hypercube partition problem}. Using this geometric approach, we  establish upper and lower bounds on the parameters of certain deterministic rounding schemes. Perhaps more importantly, the partition problem we investigate is very natural and should be of independent interest to a broader community. The introduction of the secluded hypercube partition problem is one of our main conceptual contributions. 

\begin{definition}[$(k(d),\epsilon(d))$-Deterministic Rounding]
A deterministic rounding scheme is a family of functions ${\mathcal F} = \{f_d\}_{d\in\N}$ where $f_d: \R^d \to \R^d$. We call $\mathcal{F}$ a  $(k(d),\varepsilon(d))$-deterministic rounding scheme if (1)  $\forall \vec{x}\in\R^d$, $d_{max}(\vec{x},f_d(\vec{x}))\leq 1
\footnote{
  The bound of 1 is not critical. We can use any constant and scale the parameters appropriately.
}$
(2) $\forall \vec{x}\in \R^d$, $\abs{\{f_d(\vec{z})\colon \vec{z}\in B_{\varepsilon(d)}(\vec{x})\}} \leq k(d)$. 
\end{definition}


\begin{definition}[$(k,\epsilon)$-Secluded Partition] 
A partition $\P$ of $\R^d$ is called a $(k,\epsilon)$-secluded partition if for every $\vec{p} \in \R^d$, $\abs{\set{X\in\P\colon X\cap \closure{B}_{\epsilon}(\vec{p})\not=\emptyset}}\leq k$.
\end{definition}

\begin{remark}
The use of the word ``secluded'' is meant as a synonym for ``remote'' or ``isolated'' and is meant to indicate that every point in space is only nearby (within $\epsilon$) a few (at most $k$) members of the partition.
\end{remark}

We refer to $k$ as the {\em degree} of the partition and $\epsilon$ as its {\em tolerance}. The values of $k$ and $\epsilon$ will often be functions of the dimension $d$.

There is a natural equivalence between rounding schemes and partitions of Euclidean space in a very general sense. A rounding function $f: \R^d \rightarrow \R^d$ induces a partition $\P_f$ as follows: $\P_f=\set{f^{-1}(y):y\in\range(f)}$. Conversely, a partition $\P$ induces a deterministic rounding function $f_{\P}$ as follows: for each member $X \in \P$ let $\vec{p}_{X} \in X$ some fixed representative. Then the rounding function $f_{\P}$ maps any point $\vec{p} \in X$ to $\vec{p}_{X}$. This connection leads to the following observation. 

\begin{observation}[Equivalence of Rounding Schemes and Partitions]\label{:rounding-schemes-and-partitions}
A $(k(d), \epsilon(d))$-deterministic rounding scheme induces, for each $d\in\N$, a $(k(d), \epsilon(d))$-secluded partition of $\R^d$ in which each member has diameter\footnote{The diameter is at most $2$, not at most $1$. For example, the rounding scheme which rounds a point by sending each coordinate to the nearest even integer (breaking ties by rounding up) is a $(2^d, 1)$-deterministic rounding scheme, and the partition it induces in each dimension consists of hypercubes of side length $2$.}  at most $2$. Conversely, a sequence $\langle \P_d\rangle_{d=1}^\infty$ of partitions where $\P_d$ is $(k(d), \epsilon(d))$-secluded and contains only members of diameter at most $1$ induces many\footnote{For each member, any representative of that member can be chosen.} $(k(d), \epsilon(d))$-deterministic rounding schemes.
\end{observation}

The connection between rounding schemes and partitions of Euclidean space has been observed in earlier works. In particular, the works of Fiege et. al., Kindler et. al., and Braverman and Minzer  investigated the foams problem in the context of parallel repetition~\cite{FKO07,kindler_spherical_2008,kindler_spherical_2012,BM21}. The foams problem asks to find a body that tiles $\R^d$ with $\Z^d$ with the smallest surface area. These works established  connections between the foams problem and the parallel repetition theorem and gave a new foams construction. This new construction led to an optimal {\em noise-resistant} randomized rounding scheme (more precisely a distribution $\mathcal{D}$ over deterministic rounding schemes) with the following property for all $d\in\N$:  $\forall \vec{x},\vec{y}\in\R^d$ such that $d_{max}(\vec{x},\vec{y})\leq\epsilon$, it holds that $\Pr_{f\sim\mathcal{D}}\left[f(\vec{x})=f(\vec{y})\right]\geq1-O(\varepsilon)$. In \cite[Definition~1.2]{kindler_spherical_2008}, a randomized rounding scheme is defined as a distribution of functions for each dimension rather than a single function in each dimension. Not all rounding algorithms necessarily abide by this definition of a randomized rounding scheme, but all of those discussed in this paper do including \cite{SaksZhou99,Goldreich19,kindler_spherical_2008,kindler_spherical_2012,hoza_preserving_2018,impagliazzo_reproducibility_2022}. Thus, while our results will be about deterministic rounding schemes, there is a connection between deterministic and randomized rounding schemes in this sense.



\subsection{Our Contributions}

In this work we conduct an extensive investigation of secluded partitions.
Contributions of this work are two-fold. The first contribution is the formulation of a very natural partition problem known as  the {\em secluded hypercube partition problem} and an explicit construction of such partition with degree $d+1$ and tolerance $1/2d$. The second contribution is   
establishing impossibility results on the degree and the tolerance parameters. In particular,  we establish that any ``reasonable" secluded partition must have degree at least $d+1$ and every reasonable partition with degree $d+1$ must have tolerance at most $1/2\sqrt{d}$.

We start our investigation by considering {\em unit hypercube partitions} of $\R^d$, which are very natural and extensively studied partitions of Euclidean space. A partition $\P$ of $\R^d$ is a unit hypercube partition if every $X\in\P$ is a $d$-dimensional unit hypercube. Note that the diameter of a unit hypercube is $1$ in the $d_{max}$ metric.   

\begin{question*}[Secluded Hypercube Partition Problem]\label{:secluded-hypercube-partition-problem}
Let $d\in\N$. For what values of $k\in\N$ and $\epsilon\in(0,\infty)$ does there exist a $(k, \epsilon)$-secluded unit hypercube partition $\P$ of $\R^{d}$?
\end{question*}

One of our main conceptual contributions is the formulation of the above question which is a very natural geometric question and should be of broad interest. The question asks: given $d\in\N$, $\epsilon\in(0,\infty)$, and $k\in\N$, is there a partition of $\R^d$ with unit hypercubes so that for any point $\vec{p}\in\R^{d}$, its $\epsilon$-neighborhood, in the $d_{max}$ metric, intersects with at most $k$ hypercubes?  It is easy to see that the natural grid partition of $\R^d$ with unit hypercubes is a $(2^d, 1/2)$-secluded partition. Our first technical contribution is the design of a parameterized class of explicit $(d+1,\frac{1}{2d})$-secluded unit hypercube partitions for each $d\in\N$. (See \Autoref{:hypercube-partition-thm} for a more elaborate statement of the following result and its proof.)

\begin{theorem*}[Hypercube Partition Theorem]\label{:intro-hypercube-partition-theorem}
For all $d\in\N$, there exists $(d+1,\frac{1}{2d})$-secluded unit hypercube partitions of $\R^{d}$.
\end{theorem*}

Our next result shows that the above construction is optimal with respect to the degree parameter. In particular, we show that for any unit hypercube partition of $\R^d$ and any $\epsilon\in(0,\infty)$, the degree has to be at least $d+1$. (See \Autoref{:unit-hypercube-optimality} for a slightly stronger version of the following result and its proof.)

\begin{theorem*}[Degree Optimality for Unit Hypercube Partitions]
For every $(k, \epsilon)$-secluded unit hypercube partition $\P$ of $\R^d$, it must be that $k \geq d+1$.
\end{theorem*}

This theorem raises the question of whether there exist non-hypercube partitions with degree smaller than $d+1$. There are trivial examples of such partitions since any partition with fewer than $d+1$ members trivially has degree smaller than $d+1$ no matter what value of tolerance is used. Another trivial example is concentric ``shells'' in the $d_{max}$ metric: $\P=\set{\closure{B}_1,\; \closure{B}_2\setminus \closure{B}_1,\; \closure{B}_3\setminus \closure{B}_2,\; \closure{B}_4\setminus \closure{B}_3,\; \ldots}$ (where $\closure{B}_i$ denotes the ball of radius $i$ centered at $\vec{0}$). For tolerance $\epsilon\leq\frac12$, this is a $(2,\epsilon)$-secluded partition. 

In the first of these examples, since there are a finite number of members, some member has infinite diameter and infinite measure. In the second example, though every member has finite measure and finite diameter, there is no bound on either of these quantities (i.e. for any $D$ and any $\mu$, the partition contains a member with diameter greater than $D$ and measure greater than $\mu$). If we consider partitions with either a bound on the diameter of the members or a bound on the measure of the members, then we can rule out the possibility of degree less than $d+1$. We say a partition has bounded diameter (resp. bounded measure) if there exists some bound $b\in(0,\infty)$ such that every member of the partition has diameter (resp. measure) at most $b$. Such $b$ will be called a diameter bound (resp. measure bound) for the partition.
Our next contribution is a degree optimality theorem for bounded measure partitions. (See \Autoref{:first-optimality-thm} for an equivalent statement of the following result and a proof.)

\begin{theorem*}[Degree Optimality for Bounded Measure Partitions]
   For every $(k, \epsilon)$-secluded partition $\P$ of $\R^d$ with bounded measure, it must be that $k \geq d+1$. 
\end{theorem*}

Since a partition with bounded diameter must also be a partition of bounded measure\footnote{Suppose $\P$ is a partition of $\R^d$ of bounded diameter. Then there is some $D$ such every member $X\in\P$ has $\diam(X)\leq D$. Let $\mu=(2D)^d$. For any $X\in\P$, fix some $\vec{x}\in X$ so $\closure{B}_{D}(\vec{x})\supset X$ and note that $\closure{B}_{D}(\vec{x})$ (in the $d_{max}$ metric) has measure $(2D)^d=\mu$ which shows that each member of $\P$ has measure at most $\mu$, so $\P$ has bounded measure.}, we get the following corollary. (See \Autoref{:second-optimality-thm} for a stronger version of this result and its proof.)

\begin{corollary*}[Degree Optimality for Bounded Diameter Partitions]
   For every $(k, \epsilon)$-secluded partition $\P$ of $\R^d$ with bounded diameter, it must be that $k \geq d+1$. 
\end{corollary*}

Now we turn from the degree parameter to the tolerance parameter. The \nameref{:intro-hypercube-partition-theorem} above states that there is a $(k,\epsilon)$-secluded partition with degree $k=d+1$ (the minimum possible for bounded partitions) and tolerance $\epsilon=\frac{1}{2d}$. Is it possible to design $(k,\epsilon)$-secluded partitions maintaining optimal $k=d+1$, but with larger $\epsilon$? We could just scale up the sizes of the hypercubes in the \nameref{:intro-hypercube-partition-theorem} to obtain any $\epsilon$ we want, so this question only makes sense if we compare $\epsilon$ to the size of the largest member of the partition. (See \Autoref{:sperner-upper-bound} for a more detailed version of the following result and its proof.)


\begin{theorem*}[Tolerance Upper Bound]
  For every $(d+1, \epsilon)$-secluded partition $\P$ of $\R^{d}$ with diameter bound $D$, it must be that $\epsilon \leq \frac{D}{2\sqrt{d}}$.
\end{theorem*}

Our goal was to show an upper bound of $1/\Omega(d)$ on the tolerance parameter as that would establish that our construction is optimal in both degree and tolerance. However, we could only establish this for $d=1$ and $d=2$. Closing the gap between our upper and lower bounds on the tolerance parameter is the main open question that emerged from our investigation.

The above results and \Autoref{:rounding-schemes-and-partitions} establish a limitation of deterministic rounding schemes as follows. 
\begin{corollary}[Limitations of Deterministic Rounding Schemes]
For any $(k(d),\epsilon(d))$-deterministic rounding scheme, $k(d)$ must be at least $d+1$, and for $(d+1, \epsilon(d))$-deterministic rounding schemes, $\epsilon(d)$ must be at most\footnote{\Autoref{:rounding-schemes-and-partitions} gives a diameter bound of $D=2$, so $\frac{D}{2\sqrt{d}}=\frac{1}{\sqrt{d}}$} $\frac{1}{\sqrt{d}}$.
\end{corollary}

Finally, we show an application of the hypercube partition theorem in the context of sample complexity of pseudodeterministic algorithms. Let $f_1, \ldots, f_d$ be functions from $\{0, 1\}^n$ to $[0, 1]$. Given blackbox access to the functions, the goal is to output a vector $\vec{\alpha} = \langle \alpha_1, \ldots, \alpha_d\rangle$ that is an  $(\epsilon, \delta)$-approximation to the vector $\langle\ex(f_{i})\rangle_{i=1}^{d}$ 
with respect to $d_{max}$ metric. This means that with probability at least $1-\delta$, the vector $\vec{\alpha}\in\R^{d}$ returned by the algorithm is such that $d_{max}(\vec{\alpha}, \langle\ex(f_{i})\rangle_{i=1}^{d})\leq\epsilon$ which is equivalent to saying that with probability at least $1-\delta$, it holds for all $i$ that $\abs{\alpha_{i}-\ex(f_{i})}\leq\epsilon$. Goldreich proved that there is $(d+1)$-pseudodeterministic algorithm for this task, and this algorithm has a sample complexity of $\widetilde{O}(d^{4})$~\cite{Goldreich19}. Using our hypercube partitions, we show that the sample complexity can be improved to  $\widetilde{O}(d^{2})$ samples.

\begin{restatable*}[Application to Pseudodeterministic Algorithms]{proposition}{RestatableApplicationToPseudodeterministicAlgorithms}
  Let $f_{1},\ldots,f_{d}:\nbits\to[0,1]$ be functions and $\epsilon\in(0,\infty)$ and $\delta\in(0,\frac{1}{2})$.
  There is an algorithm that given sample access to each $f_{i}$ can $(d+1)$-pseudodeterministically $(\epsilon,\delta)$-approximate $\langle\ex(f_{i})\rangle_{i=1}^{d}$ relative to the $d_{max}$ metric using $O\left(\frac{(d+1)^{2}}{\epsilon^{2}}\cdot\log\left(\frac{d}{\delta}\right)\right)$ samples.
\end{restatable*}

\section{Main Results and Proof Outlines}
\label{sec:results}

In this section, we give formal statements and proof outlines for all our results.  We start with some required notation. A complete set of notations that we use in this paper is given in \Autoref{sec:notation}.

For a point $\vec{p} \in \R^d$ and $\varepsilon \in (0,\infty)$, $\closure{B}_{\epsilon}(\vec{p})$ denotes the closed ball of radius $\epsilon$ around $\vec{p}$, with respect to the $d_{max}$ metric (we will use this metric throughout this paper). In other words,
$
\closure{B}_{\epsilon}(\vec{p})\defeq\set{\vec{x}\in\R^d:\norm{\vec{x}-{\vec{p}}}_{\infty}\leq\epsilon}=\set{\vec{x}\in\R^d:\forall i\;\abs{x_{i}-y_{i}}\leq\epsilon}.
$
 
When it is understood
that we are considering a particular partition $\P$, then for any point
$\vec{p}\in\R^{d}$, and any $\epsilon\in(0,\infty)$, we let
$\mathcal{N}_{\epsilon}(\vec{p})$ denote the members of $\P$ intersecting
the closed $\epsilon$-neighborhood of $\vec{p}$:
\[
	\mathcal{N}_{\epsilon}(\vec{p})\defeq 
	\set{X\in\P:X\cap \closure{B}_{\epsilon}(\vec{p})\not=\emptyset}.
\]
In addition, we let $\mathcal{N}_{\closure{0}}(\vec{p})$
denote the members of $\P$ whose closure contains $\vec{p}$.
\[
	\mathcal{N}_{\closure{0}}(\vec{p})\defeq
	\set{X\in\P:\closure{X}\ni\vec{p}}.
\]

These two sets will be ubiquitous throughout the paper. We will restate the \nameref{:secluded-hypercube-partition-problem} using this notation as follows: 

\begin{question}[Secluded Hypercube Partition Problem]\label{ques:motivating}
  Let $d\in\N$. Given $k\in\N$ (degree) and $\epsilon\in(0,\infty)$ (tolerance) does there exists a partition $\P$ of $\R^{d}$ consisting of unit hypercubes such that for every point $\vec{p}\in\R^{d}$, $\abs{\mathcal{N}_{\epsilon}(\vec{p})} \leq k$ ?
\end{question}
This work is a  comprehensive study of \Autoref{ques:motivating}.
It is easy to verify that  for $k=2^d$ and any $\epsilon\leq\frac{1}{2}$ such a partition exists. In particular, the standard grid partition demonstrated in \Autoref{fig:simple-partition}
has the desired property.

Can we do better? For $d=2$ by shifting layers of the standard grid partition, as demonstrated in \Autoref{fig:reclusive-partition}, we can reduce the degree to $3$ and have tolerance of $\frac14$. While it is not obvious, for $d = 3$, this shifting method yields a partition with degree $4$ and tolerance $\frac16$ (again, see \Autoref{fig:reclusive-partition}
.

\begin{figure}[H]
    \centering
      \includegraphics[width=\textwidth]{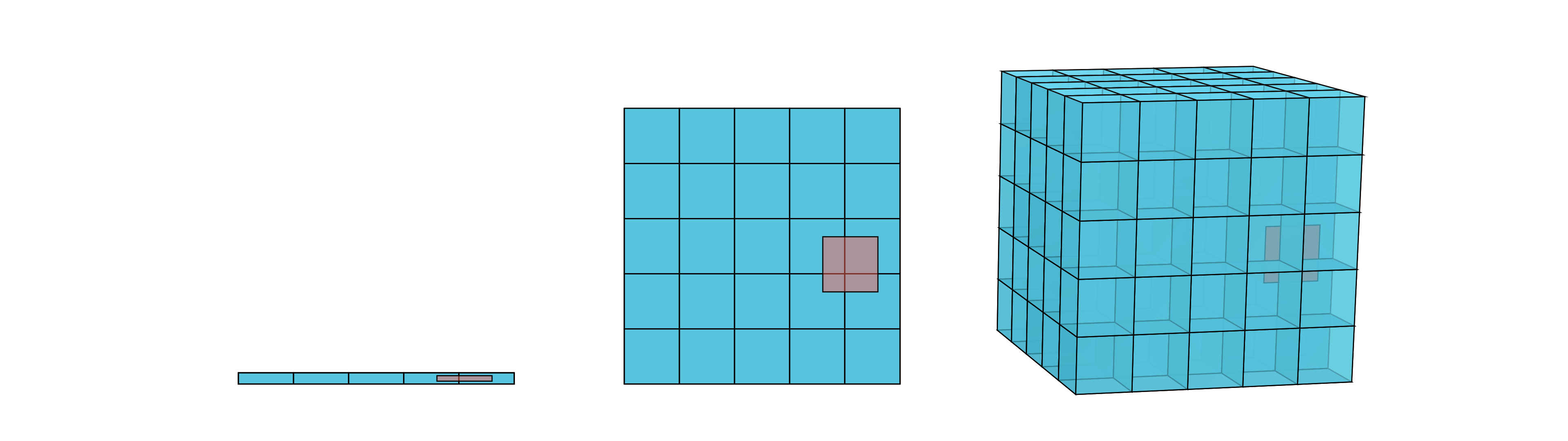}
    \caption{Simple unit hypercube partitions of $\R^{d}$ for $d=1,2,3$. The shaded red regions are $\epsilon=\frac{1}{2}$ radius balls (in the $d_{max}$ metric/$\norm{\cdot}_{\infty}$ norm) showing that these neighborhoods don't intersect more than $k=2^d$ members of the partition (it might look from the picture that centering the ball in one of the hypercubes would intersect more, but it does not because the hypercubes are half open). Furthermore, there are hypercubes that do intersect $2^d$ members as shown. (This figure was produced using Manim \cite{The_Manim_Community_Developers_Manim_Mathematical_2021}.)}
    \label{fig:simple-partition}
\end{figure}

\begin{figure}[H]
    \centering
      \includegraphics[width=\textwidth]{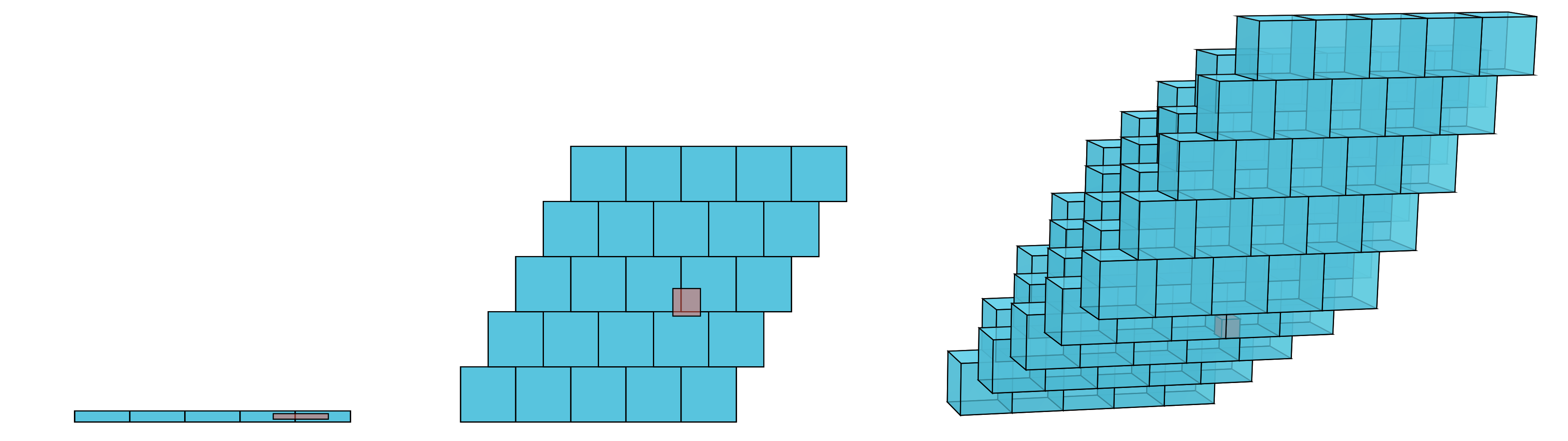}
    \caption{This is a particular reclusive partition of $\R^{d}$ for $d=1,2,3$. The shaded red regions are $\epsilon=\frac{1}{2d}$ radius balls (in the $d_{max}$ metric/$\norm{\cdot}_{\infty}$ norm) showing that these neighborhoods don't intersect more than $k=d+1$ members of the partition. Observe that in the $d=2$ case, the $x$-direction shift (horizontal) is $\frac{1}{2}$ unit; in the $d=3$ case, there are shifts of $\frac{2}{3}$ as well as $\frac{1}{3}$ in various directions. (This figure was produced using Manim~\cite{The_Manim_Community_Developers_Manim_Mathematical_2021}.)}
    \label{fig:reclusive-partition}
\end{figure}


Our first result shows that this intuition generalizes to all dimensions. We construct hypercube partitions with degree $d+1$ and tolerance $\frac{1}{2d}$. We note that, while the intuition about the degree bound generalizes from lower dimensions, we do not believe it was obvious that $\epsilon=\frac{1}{2d}$ was
achievable. We show at the end of \Autoref{sec:reclusive-lattice-partitions} that a slight change to our partitions, while seemingly innocuous, results in significantly worse degree. 

\begin{restatable*}[Hypercube Partition Theorem]{theorem}{RestatableHypercubePartitionThm}\label{:hypercube-partition-thm}
  Let $d\in\N$. Then there exists a $(d+1,\frac{1}{2d})$-secluded unit hypercube partition $\P$ of $\R^{d}$. I.e., for any point $\vec{p}\in\R^{d}$, we have
  \[
     \abs{\mathcal{N}_{\frac{1}{2d}}(\vec{p})}\leq d+1.
  \]
\end{restatable*}

\begin{proof}[Proof Outline]
  The proof of this theorem is constructive. We specify a particular basis for the vector space $\R^d$  and  consider all linear combinations of these vectors in which all coefficients are integers (sometimes called integer linear combinations). These integer linear combinations are such that a unit hypercube can be placed at each position and it will form a partition of $\R^{d}$. We view the basis as the columns of a matrix in an appropriate order.  We define a class of matrices which we call {\em reclusive matrices}---these are upper triangular matrices with $1$'s on the diagonal, and in each row the values that appear after the diagonal entry (inclusive) are decreasing. Given such a reclusive matrix $A$, we associate a lattice group $L_A$ which is the set of vectors of the form $A\vec{v}$ where $\vec{v} \in \Z^d$. Now the members of the partition are the hypercube shifted by elements of the lattice group $L_A$. We prove that every reclusive matrix induces a $(d+1, \epsilon)$-secluded unit hypercube partition where the value of $\epsilon$ depends on a property of the matrix called the {\em reclusive distance} $\Delta$.
 
 For two hypercubes $X,Y$ in this partition, we consider the vectors $\vec{n}$ and $\vec{m}$ (viewed as integer sequences) which define the positions of $X$ and $Y$ relative to the chosen basis. We show that $X$ and $Y$ are adjacent if and only if the vector/sequence $\vec{n}-\vec{m}$ is an alternating sequence of $-1$'s and $1$'s padded by $0$'s (what we call {\em weak-alt-1 sequences}). We use this characterization to prove that for $\Delta$, two members of the partition are either adjacent or their positions are at least $\Delta$ far apart. This means that for any ball of radius $\frac{\Delta}{2}$, all partition members that it intersects must be pairwise adjacent. I.e., they constitute a clique in the graph theory sense when the partition is viewed as an infinite graph with members being vertices and edges representing adjacency\footnote{Members are adjacent if their closures intersect.} of members. Because of this, it suffices to show that there are no cliques of size greater than $d+1$ in this partition graph. We do this by giving an explicit graph coloring of the partition using $d+1$ colors. This will show that the partition defined is a $(d+1, \Delta/2)$-secluded partition. Finally, it is easy to construct (we give examples) reclusive matrices with $\Delta = 1/d$.

\end{proof}



Recall from the introduction that a deterministic rounding scheme yields a partition of $\R^d$. Thus using the rounding scheme of Hoza and Klivans~\cite{hoza_preserving_2018}, we can obtain a $(d+1, \frac{1}{6(d+1)})$-secluded partition of $\R^d$ whose diameter is bounded above by 1. However, the members of this partition are not unit hypercubes (See \Autoref{fig:preserving_randomness_hoza-klivans_partition} in \Autoref{sec:hoza-klivans} for $d = 2$).  For the sake of completeness, we provide a proof of the theorem in \Autoref{sec:hoza-klivans}.

\begin{restatable}[Hoza-Klivans Partition]{theorem}{RestatableHozaKlivansPartitionThm}\label{:hoza-klivans-partition}
  Let $d\in\N$. Then there exists a $(d+1,\frac{1}{6(d+1)})$-secluded  partition $\P$ of $\R^{d}$ with diameter at most $1$.
\end{restatable}


\subsection{Lower Bound on Degree ($k$)} 
Our next set of results investigate  the optimality of the degree parameter ($k$)  of partitions. We first prove that the value of $d+1$ in \Autoref{:hypercube-partition-thm} is indeed optimal for hypercube partitions.

\begin{restatable*}[Optimality for Unit Hypercube Partitions]{theorem}{RestatableOptimalityForUnitHypercubePartitionsThm}\label{:unit-hypercube-optimality}
    Let $d\in\N$ and $\P$ a unit hypercube partition of $\R^{d}$. Then there exists $\vec{p}\in\R^{d}$ such that
  \[
    \abs{\mathcal{N}_{\closure{0}}(\vec{p})}\geq d+1
  \]
  Furthermore, $\P$ contains a $(d+1)$-clique.
\end{restatable*}

The above theorem states that there is a particular $\vec{p}\in\R^{d}$ so that for every choice of $\epsilon$, $B_{\epsilon}(\vec{p})$ intersects at least $d+1$ members of the partition, and thus the value $d+1$ from \Autoref{:hypercube-partition-thm} cannot be made any smaller. A $(d+1)$-clique in $\P$ means that there is a set of $d+1$ members of $\P$ that have pairwise intersecting closures (see \Autoref{:defn-adj}), and its existence follows trivially from the first claim of the theorem. This result is stronger than just claiming that such a partition $\P$ must have degree parameter at least $d+1$ (see the upcoming discussion of the three Optimality Theorems (or \Autoref{subsec:optimality-theorems} and \Autoref{subsec:gaps}) for details.)


As discussed in the prior section, in generalizing from hypercubes to more general partitions, we have to impose some condition to avoid trivialities, and we previously mentioned partitions with bounded measure or bounded diameter. We also consider one other condition which imposes a type of local finiteness.
Thus, we prove three different degree optimality theorems. Each has a slightly different requirement on the types of partitions under consideration (each bounding the size of the partition members in some way) and each has a slightly different sense in which the partitions are considered optimal.

In the statements below, $m$ denotes the Lebesgue measure (intuitively the volume of a set). We present all three theorems before discussing the proof outlines. The First Optimality Theorem states that for partitions of $\R^{d}$ in which all members are bounded in measure, then there is no secluded partition with degree less than $d+1$, because no matter the value of $\epsilon$, some point $\vec{p}$ can be found such that $\abs{\mathcal{N}_{\epsilon}(\vec{p})}\geq d+1$.

\begin{restatable*}[First Optimality Theorem]{theorem}{RestatableFirstOptimalityThm}\label{:first-optimality-thm}
   If $d\in\N$, and $\P$ is a partition of $\R^{d}$, and there exists $M\in(0,\infty)$ such that for all $X\in\P$, $X$ is Lebesgue measurable, and $m(X)<M$, then for any $\epsilon\in(0,\infty)$, there exists $\vec{p}\in\R^{d}$ such that
    \[
    \abs{\mathcal{N}_{\epsilon}(\vec{p})}\geq d+1.
    \]
\end{restatable*}

The Second Optimality Theorem states that under the stronger hypothesis that the members of the partition are bounded in diameter, then not only there is no secluded partition with degree less than $d+1$, but it is even false if we allow $\epsilon$ to be different for each point in the space. The $\epsilon$ function in the statement below should be viewed as some fixed $\epsilon$ for each point of $\R^{d}$.

\begin{restatable*}[Second Optimality Theorem]{theorem}{RestatableSecondOptimalityThm}\label{:second-optimality-thm}
    If $d\in\N$, and $\P$ is a partition of $\R^{d}$, and there exists $D\in(0,\infty)$ such that for all $X\in\P$, it holds that $\diam(X)<D$, then for any $\epsilon:\R^{d}\to(0,\infty)$, there exists $\vec{p}\in\R^{d}$ such that
    \[
    \abs{\mathcal{N}_{\epsilon(\vec{p})}(\vec{p})}\geq d+1.
    \]
\end{restatable*}

 The conclusion of the Second Optimality Theorem is  stronger than the conclusion of the First Optimality Theorem, but the hypothesis is also stronger (assuming all members are measurable); could it be, though, that the hypothesis of the First Optimality Theorem implies the conclusions of the Second Optimality Theorem and we just didn't find a proof of this?  We show in \Autoref{sec:clique-optimality} that this is not the case. There really is a ``gap'' between these two theorems.
 
 The Third Optimality Theorem uses a still stronger hypothesis and requires finiteness somewhere in the space. With this strengthened hypothesis, the conclusion can again be strengthened to say that we don't even really care about the $\epsilon$ neighborhoods at all because there is some point at the closure of $d+1$ members. The theorem hypothesis uses the notion of a {\em strict pairwise bound} $D$ (\Autoref{:strict-pairwise-bound}) which just means that for each member $X\in\P$, all points in $X$ are distance {\em strictly less} than $D$ apart\footnote{This is stronger than saying that the diameter of each member is $\leq D$, but weaker than saying that the diameter of each member is $<D$. Members can have diameter $D$, but can't have points that attain that diameter.} (with respect to $d_{max}$).

\begin{restatable*}[Third Optimality Theorem]{theorem}{RestatableThirdOptimalityThm}\label{:third-optimality-thm}
    If $d\in\N$, and $\P$ is a partition of $\R^{d}$, and there exists $D\in(0,\infty)$ such that for all $X\in\P$, $D$ is a strict pairwise bound for $X$ (it is sufficient but not necessary that $\diam(X)<D$), and if there exists some $\vec{\alpha}\in\R^{d}$ such that $\vec{\alpha}+[0,D]^{d}$ intersects finitely many members of $\P$, then there exists $\vec{p}\in\R^{d}$ such that
    \[
    \abs{\mathcal{N}_{\closure{0}}(\vec{p})}\geq d+1.
    \]
    Furthermore, $\P$ contains a $(d+1)$-clique.
\end{restatable*}

Just as with the first two Optimality Theorems, we will show in \Autoref{sec:clique-optimality} that the hypothesis of the Second Optimality Theorem does not imply the conclusion of the Third Optimality Theorem, so again there is a ``gap''. Now we turn to the proof outlines for the four theorems above.

\hypertarget{unit-hypercube-optimality-proof-outline}{}
\begin{proof}[Proof Outline for \nameref*{:unit-hypercube-optimality} (\Autoref*{:unit-hypercube-optimality})]
    As a corollary to the Third Optimality Theorem, we get the same conclusion for partitions whose members have a uniform upper bound on diameter and a uniform lower bound on the measure (see \Autoref{:diam-measure-cor}), and since unit hypercube partitions have this property, the conclusion follows.
\end{proof}
\begin{proof}[Alternate Proof Sketch]
    It is also possible to directly prove this result without using the heavy machinery of the three Optimality Theorems by showing that if $\vec{p}$ is one of the corners of some hypercube $X$, then it is at the closure of at least $d+1$ members. The argument goes by induction on $d$ and considers the $2^d$-many orthants locally around $\vec{p}$ showing that because $X$ takes up exactly one orthant, the only ways to fill up the other $2^d-1$ orthants requires at least $d$ other hypercubes.
\end{proof}

We give the outlines of the three Optimality Theorems in order of increasing complexity (which happens to be reverse order).

\hypertarget{third-optimality-thm-proof-outline}{}
\begin{proof}[Proof Outline for the Third Optimality Theorem]
  To prove this, let $H$ denote the closed hypercube $\vec{\alpha}+[0,D]^{d}$ which intersects only finitely many elements of $\P$. We use the partition $\P$ to induce a (finite) partition $\S$ on $H$:
  \[
  \S=\set{X\cap H: X\in\P,\;X\cap H\not=\emptyset}.
  \]
  $D$ remains a strict pairwise bound for all $X$ in the induced partition $\S$. This means that no member of $\S$ intersects opposing facets of the hypercube $H$, and thus no two corners of $H$ belong to the same member of $\S$. We show that this allows us to color all points of $H$ using $2^d$ colors (one associated with each corner) in such a way to satisfy the coloring properties of a generalized Sperner's lemma/KMM lemma, and such that all points within any member of $\S$ are assigned the same color. This will let us conclude that there is some point $\vec{p}$ of the space belonging to the closure of $d+1$ colors, and because there are only finitely many members of $\S$, we use the fact that the closure of a finite union of members is equal to the finite union of the closures, and so $\vec{p}$ belongs not just to the closure of $d+1$ colors, but to the closure of $d+1$ members of $\S$. The existence of a $(d+1)$-clique is a trivial consequence because all of these $d+1$ members have closures that contain $\vec{p}$ and so each pair of these members is adjacent (see \Autoref{:defn-adj}).
\end{proof}

The proof of the Second Optimality Theorem is quite similar to the previous proof. 

\hypertarget{second-optimality-thm-proof-outline}{}
\begin{proof}[Proof Outline for the Second Optimality Theorem]
  Let $H=[0,D]^{d}$ and let $\S$ denote the partition of $H$ induced by $\P$, observing that as before, the fact that each $X\in\P$ has $\diam(X)<D$ implies no member of $\P$ (and thus $\S$) intersects opposing facets of $H$. Using the same Sperner/KMM technique as before we can find some point $\vec{p}$ in the closure of $d+1$ colors. However, because the closure of an infinite collection of members may be larger than the infinite union of the closures of those members, it is possible that $\vec{p}$ does not belong to the closure of $d+1$ members of $\S$. However, no matter what the value of $\epsilon(\vec{p})$ is, this neighborhood of $\vec{p}$ contains an open set around $\vec{p}$, and that open set must intersect some member of $\S$ corresponding to each of the $d+1$ colors. This gives the result.
\end{proof}

The proof of the First Optimality Theorem is more involved.

\hypertarget{first-optimality-thm-proof-outline}{}
\begin{proof}[Proof Outline for the First Optimality Theorem]
  Having only a bound on the measure of the members of $\P$, it is possible that they extend arbitrarily far (even infinitely far) within $\R^{d}$, so the goal will be to approximate $\P$ by a nicer partition that has bounded diameter elements and then utilize the same Sperner/KMM techniques as used for the Second and Third Optimality Theorems. In fact, the approximating partition $\mathcal{A}$ that we construct will satisfy the hypothesis of the Third Optimality Theorem, so there is a point $\vec{p}$ belonging to the closure of $d+1$ members of $\mathcal{A}$. We then argue that the partition $\mathcal{A}$ is similar enough to $\P$ that the set $B_{\epsilon}(\vec{p})$ intersects at least $d+1$ members of $\P$.
  
  In more detail, the steps are as follows:
  \begin{enumerate}
      \item Partition $\R^{d}$ into a grid of $\epsilon$-diameter hypercubes
      \item Label each of these hypercubes with a member of $\P$ that it intersects a lot (i.e. the measure of the intersection is sufficiently large)
      \item Define a graph (equivalently a symmetric binary relation) on the hypercubes so that they are adjacent in the graph if they are close together and also share the same label
      \item Consider the transitive closure of the relation (i.e. connected components of the graph); show that each connected component has cardinality less than a fixed size, and so has bounded diameter
      \item Consider the partition $\mathcal{A}$ where members correspond to equivalence classes and apply the Third Optimality Theorem to get a point $\vec{p}$ at the closure of $d+1$ members of $\mathcal{A}$
      \item Show that each such element of $\mathcal{A}$ is a superset of a hypercube in the grid and that different elements of $\mathcal{A}$ which have $\vec{p}$ in their closure contain hypercubes with distinct labels
      \item Show that each of the $d+1$ distinctly labeled hypercubes above are contained in $B_{\epsilon}(\vec{p})$
  \end{enumerate}
\end{proof}

\subsection{Upper Bound on Tolerance ($\epsilon$)}

With the optimality of the parameter $k$ firmly established, we wish to consider whether our value of $\epsilon=\frac{1}{2d}$ in \Autoref{:hypercube-partition-thm} is maximal when we have $k$ minimized at $k=d+1$. Ideally, we want to consider optimality of the $\epsilon$ parameter not just for unit hypercube partitions but for a broader class of partitions. A natural class is partitions with the property that each member has a strict pairwise bound\footnote{While it might seem that considering partitions with members of diameter at most $D=1$ would be more natural than considering a strict pairwise bound of $1$, working through the proofs of the results below one can begin to see that actually, the strict pairwise bound is the ``right'' thing to do as that is the condition that works best in the proofs. In fact, trying to use the diameter leads to slightly weaker results.} (see \Autoref{:strict-pairwise-bound}) of $D=1$; all unit hypercube partitions belong to this class\footnote{This is at least true as we define unit hypercube partitions using half-open hypercubes (see \Autoref{:defn-unit-hypercube-partition}).}, and $D=1$ is the smallest value of $D$ we can use so that they do.
For the sake of comparisons between some results in the paper we didn't want to fix $D$ to be $1$ in the statement of the results below, but taking any of the results below as a standalone result, the value of $D$ is nothing more than a trivial scaling factor, so the reader can consider it to be fixed at $D=1$. The next result gives a trivial bound on the tolerance parameter, but gives a tight bound for $d=1$.

\begin{restatable*}[Trivial Tolerance ($\epsilon$) Bound]{proposition}{RestatableTrivialToleranceBound}\label{:optimal-diam-d1}
  If $d\in\N$, and $\P$ is a partition of $\R^{d}$, and there exists $D\in(0,\infty)$ such that for all $X\in\P$, $\diam(X)\leq D$, and if $\epsilon\in(0,\infty)$ such that for all $\vec{p}\in\R^{d}$,
  \[
    \abs{\mathcal{N}_{\epsilon}(\vec{p})}\leq 2^{d}
  \]
  then $\epsilon\leq\frac{D}2$.
\end{restatable*}
\begin{proof}[Proof Sketch]
    If $\epsilon>\frac{D}2$, we fix an arbitrary member $X$ of $\P$, and shift the hypercube $[-\epsilon, \epsilon]^d$ so that $X$ is completely contained in the interior. Then we observe that the $2^d$-many corners of the shifted $[-\epsilon, \epsilon]^d$ are distance greater than $D$ apart (in $d_{max}$), so they each belong to a distinct member of $\P$, and none of them belong to $X$. Taking $\vec{p}$ to be the center of the shifted $[-\epsilon, \epsilon]^d$, we have for contradiction that     $\abs{\mathcal{N}_{\epsilon}(\vec{p})}\geq 2^{d}+1 > 2^{d}+1$.
\end{proof}

By considering the specific value $d=1$, we get the following corollary which states that $\epsilon=\frac{D}{2d}=\frac{1}{2d}$ is indeed the optimal value for the class of partitions with strict pairwise bound of $D=1$ for dimension $d=1$.

\begin{restatable*}[Optimal Tolerance ($\epsilon$) in $\R^1$]{corollary}{RestatableOptimalToleranceROne}\label{:optimal-tolerance-R1}
  If $d=1$ and $\P$ is a partition of $\R^{d}=\R^{1}$, and there exists $D\in(0,\infty)$ such that for all $X\in\P$, $D$ is a strict pairwise bound for $X$, and if $\epsilon\in(0,\infty)$ is such that for all $\vec{p}\in\R^{d}$,
  \[
    \abs{\mathcal{N}_{\epsilon}(\vec{p})}\leq d+1=2
  \]
  then $\epsilon\leq\frac{D}2=\frac{D}{2d}$.
\end{restatable*}
\begin{proof}
Since $D$ is a strict pairwise bound for each member of $\P$, then every member of $\P$ has diameter at most $D$. Apply \Autoref{:optimal-diam-d1} with $d=1$.
\end{proof}

We can also show that the value of $\epsilon=\frac{D}{2d}$ is optimal for $d=2$, but \Autoref{:optimal-diam-d1} evaluated with $d=2$ does not give the correct value, so it requires a separate proof.

\begin{restatable*}[Optimal Tolerance ($\epsilon$) in $\R^2$]{proposition}{RestatableOptimalToleranceInRTwo}\label{:optimal-diam-d2}
  If $d=2$ and $\P$ is a partition of $\R^{d}=\R^{2}$, and there exists $D\in(0,\infty)$ such that for all $X\in\P$, $D$ is a strict pairwise bound for $X$, and if $\epsilon\in(0,\infty)$ is such that for all $\vec{p}\in\R^{d}$,
  \[
    \abs{\mathcal{N}_{\epsilon}(\vec{p})}\leq d+1=3
  \]
  then $\epsilon\leq\frac{D}4=\frac{D}{2d}$.
\end{restatable*}

For other dimensions, we do not know if $\epsilon=\frac{D}{2d}$ is optimal (though we conjecture that it is up to a constant factor; see \Autoref{:linear-conjecture}). The next result at least gives an upper bound that for any $d\in\N$, $\epsilon\leq\frac{D}{2\sqrt{d}}$\footnote{There are two useful ways to interpret this upper bound. First, the value of $\epsilon=\frac{1}{2d}$ that we achieve is within a factor of $\frac{1}{\sqrt{d}}$ of the optimal value $\epsilon'\leq\frac{1}{2\sqrt{d}}$ which gives a bound on the quality in terms of the dimension. Second, the value of $\epsilon=\frac{1}{2d}$ that we achieve is correct to within a square: $\epsilon=\frac{1}{2d}> \left(\frac{1}{2\sqrt{d}}\right)^2 \geq (\epsilon')^2$ which gives a dimension invariant way to view it.}.

\begin{restatable*}[Universal Tolerance ($\epsilon$) Bound]{theorem}{RestatableUniversalToleranceBound}\label{:sperner-upper-bound}
  If $d\in\N$, and $\P$ is a partition of $\R^{d}$, and there exists $D\in(0,\infty)$ such that for all $X\in\P$, it holds that $D$ is a strict pairwise bound for $X$, and if there exists $\epsilon\in(0,\infty)$ such that for all $\vec{p}\in\R^{d}$,
  \[
    \abs{\mathcal{N}_{\epsilon}(\vec{p})}\leq d+1,
  \]
  then $\epsilon\leq\frac{D}{2(\sqrt[d]{\sperner(d)}-1)}$ (if $d>1$). In particular, $\epsilon\leq\frac{D}{2\sqrt{d}}$ (for any $d$). 
\end{restatable*}

The $\sperner(d)$ quantity in the theorem statement is something that will be discussed in depth prior to proving the theorem. It will be sufficient for now to note that $\sperner(d)$ is bounded below by the minimum number of simplices needed in any dissection of the $d$-cube, and the ``in particular'' claim in the theorem uses a lower bound of $(d+1)^{\frac{d-1}2}$ for the dissection number. Any improvements that are made to the dissection number lower bound will translate directly to improvements of our $\epsilon$ upper bound.

\begin{proof}[Proof Methods of \Autoref*{:sperner-upper-bound}]
We apply a stronger version of Sperner's lemma (or, equivalently, a stronger version of the KKM lemma) to find multiple points within a fixed volume that are near $d+1$ different members of the partition. It is then argued that the $\epsilon$-balls around these points must be disjoint, and the value of $\epsilon$ is bounded from above using a volume/measure argument.
\end{proof}

\section{Related Work in Mathematics}
\label{sec:related}


We showed in \Autoref{sec:introduction} how our work is related to rounding schemes used in computational complexity theory, and in this section we connect our work to the rich history of and continued interest in partitions of $\R^{d}$, and in particular, partitions by unit hypercubes. Note that in essence, a partition by unit hypercubes is the same as a tiling or a packing of unit hypercubes. 
The only distinction is that in tilings and packings, the boundaries are ignored and in a partition they are not. Further, it is common in the literature to refer to $[0,1)^{d}$ and $[0,1]^{d}$ as a ``cube'', ``$d$-cube'', or ``hypercube''. Below, we discuss a number of questions that have been investigated in the literature regarding tilings and unit hypercubes. Most of these results deal with translated unit cubes (i.e. no rotation or the higher dimensional analogs of rotation).

The purpose of these examples is fourfold: (1) to demonstrate broad interest in hypercubes and hypercube partitions, (2) to show that even though questions about hypercubes may seem very simple, there remains active research in this area, (3) to preview a few questions related to some of our results, and (4) to demonstrate that there are many results which hold in $\R^{1}$, $\R^{2}$, and $\R^{3}$ but which may not hold in higher dimensions.

\vspace{-4mm}
\paragraph{Minkowski lattice cube-tiling conjecture (1907)}
Minkowski's conjecture \cite{minkowski_diophantische_1907} states that in any lattice tiling of $\R^{d}$ by translated unit cubes there exists a pair of cubes whose intersection is an entire $(d-1)$-dimensional face (e.g. in $\R^{2}$ there would be a pair of squares with a common edge, and in $\R^{3}$ there would be a pair of cubes sharing a common square side). A lattice tiling is one in which the centers of all of the cubes form a critical lattice (see Related Problems in \cite{noauthor_kellers_2021} and \cite{noauthor_lattice_2021} for definitions). It was proven true in 1941 by Haj\'os \cite{hajos_uber_1942}.

\paragraph{Keller's conjecture (1930)}
Keller's conjecture \cite{keller_uber_1930} is a generalization of Minkowski's conjecture, which relaxes the assumption that the cubes form a lattice. Thus, it states that in any tiling of $\R^{d}$ by translated unit cubes, there exists a pair of cubes whose intersection is an entire $(d-1)$-dimensional face. The complete resolution of this conjecture took substantial effort only being completely resolved in 2020. In 1940 Perron \cite{perron_uber_1940, perron_uber_1940-1} showed it was true for $d\leq6$. Szab\'{o} \cite{szabo_reduction_1986} recast the question in terms of periodic tilings in 1986 and then introduced the so-called Keller graphs along with Corr\'{a}di in 1990 \cite{corradi_combinatorial_1990}. In 1992, Lagarias and Shor \cite{lagarias_kellers_1992} used the Keller graphs to show that the conjecture is false for all $d\geq10$. This bound was refined by Mackey in 2002 \cite{mackey_cube_2002} showing that the conjecture is false for $d\geq8$. Progress on the only remaining case of $d=7$ was made by Debroni, Eblen, Langston, Myrvold, Shor, and Weerapurage in 2011 \cite{debroni_complete_2011}, and by Kisielewicz and {\L}ysakowska in 2014 \cite{kisielewicz_kellers_2014}, and by Kisielewicz in 2017 \cite{kisielewicz_towards_2017} and by {\L}ysakowska in 2018 \cite{noauthor_extended_nodate}. Finally, in 2020, Brakensiek, Heule, Mackey, and Narv\'{a}ez \cite{brakensiek_resolution_2020} determined that the conjecture was true for $d=7$ using automated satisfiability approaches which fully resolved the conjecture.

\vspace{-4mm}
\paragraph{Furtw\"angler's conjecture (1936)} 
Furtw\"angler's conjecture \cite{furtwangler_uber_1936} is another generalization of Minkowski's conjecture where instead of tilings, $k$-fold tilings are considered (a $k$-fold tiling is a collection of positions so that if a hypercube is placed at each position, then every point of $\R^{d}$ either belongs to the boundary of some hypercube, or belongs to exactly $k$ hypercubes). Furtw\"angler's conjecture states that in any $k$-fold lattice tiling of $\R^{d}$, there exists a pair of cubes whose intersection is an entire $(d-1)$-dimensional face, and he proved it for $d\leq3$. However, Haj\'os proved in 1942 \cite{hajos_uber_1942} that the conjecture was false for $d\geq4$. In 1979, Robinson \cite{robinson_multiple_1979} completely characterized the conjecture by proving for exactly which pairs $(k,d)$ the conjecture held and which it did not.

\vspace{-4mm}
\paragraph{Fuglede's set conjecture (1974) and functional analysis}
A set $\Omega\subset\R^{d}$ is called a spectral set if it has positive measure and if there is a basis of certain exponential functions for the space $L^{2}(\Omega)$ of square integrable functions (the set generating the basis is denoted $\Lambda$). Fuglede \cite{fuglede_commuting_1974} conjectured that a set was spectral if and only if it could be used to tile $\R^{d}$. Though this was proved false by Tao in 2004 \cite{tao_fugledes_2004}, earlier work by Lagarias, Reeds, and Wang in 2000 \cite{lagarias_orthonormal_2000} showed something similar for the special case of unit cubes. In particular, they showed that a set $\Lambda$ will generate a basis for $L^{2}([0,1]^{d})$ if and only if $\Lambda$ is the set of center positions of hypercubes in some partition of translated hypercubes. They used this result to show that extending an orthogonal set of functions to a basis is equivalent to extending a packing of hypercubes to a tiling. The ability to extend packings to tilings was also studied by Dutour, Itoh, and Poyarkov in 2018 \cite{dutour_cube_2018}, though in a different context.

\vspace{-4mm}
\paragraph{Countable partitions of $[0,1]$ by closed sets}
A problem solved by Sierpi\'{n}ski \cite{sierpinski_1918, countable_closed_partition_interval} was that there is no (non-trivial) partition of the unit interval into countably many closed sets. This extends trivially to show that there is no (non-trivial) partition the unit hypercube $[0,1]^{d}$ into countably many closed sets. This has some connection to our work. In particular, we state a conjecture which is equivalent to this problem when $d=1$, and thus provides a natural interesting generalization of this problem to higher dimensions.

\vspace{-4mm}
\paragraph{Coverings, dissections, and triangulations of the cube}
A covering of the cube is a set of simplices (using the vertices of the cube) so that the union of the simplices is the entire cube. A dissection is a covering with the additional requirement that the only overlap occurs at the boundary of the simplices. A triangulation is a dissection with the additional requirement that the intersection of any two simplices is either empty or a face of each. It has long been known that there is a triangulation of the $d$-cube $[0,1]^{d}$ using $d!$ simplices. In 1982, Sallee \cite{sallee_triangulation_1982, SALLEE1982211} gave lower bounds for how many simplices are needed in a triangulation, and many others have tried to bound the minimal number of simplices needed for a covering, a dissection, and a triangulation since then. In this paper, we will utilize the dissection number for one of our upper bounds and will utilize the dissection number lower bound of Glazyrin from 2012 \cite{glazyrin_lower_2012}.

\vspace{-4mm}
\paragraph{Unit cubes more broadly}
Many of the above results are discussed in the 2005 survey paper ``What Is Known About Unit Cubes'' by Zong \cite{zong_what_2005} and in his 2006 book \cite{zong_cube-window_2006}. In addition, many other properties of unit cubes are presented. Zong makes the case that despite the apparent simplicity of the $d$-cube,
there is much that remains unknown about it.
\section{Organization}
The remainder of the paper provides a complete description and proofs of the established results in detail.  In \Autoref{sec:notation}, we summarize notation that we will use. In \Autoref{sec:defn-basics} we present some basic definitions and results pertaining to partitions. In \Autoref{sec:reclusive-lattice-partitions} we construct the a class of partitions which we call reclusive partitions and prove the Hypercube Partition Theorem. In \Autoref{sec:clique-optimality} we prove three degree Optimality Theorems to establish that $k=d+1$ is the minimum value of $k$ for a broad range of reasonable partitions. In \Autoref{sec:epsilon} we first prove the optimality of the tolerance parameter for dimensions 1 and 2 (\Autoref{:optimal-diam-d1} and \Autoref{:optimal-diam-d2}).  In this section, we further show the $1/2\sqrt{d}$ upper bound the tolerance for higher dimensions. 
In \Autoref{sec:algorithm} we offer an application of secluded partitions to deterministic rounding and multipseudodeterminism. \Autoref{sec:future} contains concluding remarks and a discussion about future research directions.
\section{Notation}
\label{sec:notation}

We will use the following notation and conventions throughout this paper.
\begin{itemize}
\item We use $\N$ throughout to indicate the strictly positive integers.

\item For any $d\in\N$, we let $[d]$ denote the set of the first $d$ positive integers ($[d]=\set{1,2,\ldots,d}$).

\item If $d\in\N$ and $S\subseteq\R^{d}$ and $\vec{v}\in\R^{d}$, then we use the following notation to indicate shifting/translating the set $S$:
  \[
    \vec{v}+S = S+\vec{v}\defeq\set{\vec{s}+\vec{v}\colon \vec{s}\in S}.
  \]


\item For $d\in\N$ we define the metric $d_{max}$ on $\R^{d}$ which is the metric induced by the $l^{\infty}$ norm\footnote{Technically this is a different metric for each dimension $d$, but it is not necessary to include this in the notation as the dimension will be clear from context.}:
  \[
    d_{max}(\vec{x},\vec{y})\defeq\max_{i\in[d]}\abs{x_{i}-y_{i}}=\norm{\vec{x}-\vec{y}}_{\infty}.
  \]
  While we could use the $\norm{\cdot}_{\infty}$ norm notation throughout the paper, we choose to use the $d_{max}$ metric notation when possible.

\item For any $X\subseteq\R^{d}$ we denote the closure of $X$ with respect to the $d_{max}$ metric by $\clos{X}$. Since the topology induced by the $l^{\infty}$ norm on $\R^{d}$ is the same as the Euclidean topology (induced by the $l^{2}$ norm),
  and the Euclidean topology on $\R^{d}$ is the same as the product topology on $\R^{d}$, the closure is the same when taken with respect to any of these. In particular, if $X_{1},\ldots,X_{d}\subseteq\R$ and $X=\prod_{i=1}^{d}X_{i}$, then $\clos{X}=\prod_{i=1}^{d}\clos{X_{i}}$ (because the product of closures is the same as the closure of the product).

\item For any $\vec{x}\in\R^{d}$ and $\epsilon\in(0,\infty)$ we use $B^{\circ}_{\epsilon}(\vec{x})\defeq\set{\vec{y}\in\R^{d}:d_{max}(\vec{x},\vec{y})<\epsilon}$ to denote the open ball of radius $\epsilon$ around $\vec{x}$, and we use $\closure{B}_{\epsilon}(\vec{x})\defeq\set{\vec{y}\in\R^{d}:d_{max}(\vec{x},\vec{y})\leq\epsilon}$ to denote the closed ball of radius $\epsilon$ around $\vec{x}$, noting that $\closure{B}_{\epsilon}(\vec{x})=\closure{B^{\circ}_{\epsilon}(\vec{x})}$. The standard notation for the open ball is $B_{\epsilon}(\vec{x})$, but we will not use the open ball very often in this paper and prefer to use the circle superscript to emphasize when we do. (Any occurrence in the paper of a ball without the overline or the circle superscript is a typo.)

\item Though mentioned in the introduction, for ease of reference, we again define the $\epsilon$-neighborhoods. When it is understood that we are considering a particular partition $\P$, then for any point $\vec{p}\in\R^{d}$, and any $\epsilon\in(0,\infty)$, we let
\[
	\mathcal{N}_{\epsilon}(\vec{p})\defeq 
	\set{X\in\P:X\cap \closure{B}_{\epsilon}(\vec{p})\not=\emptyset}.
\]
\[
	\mathcal{N}_{\closure{0}}(\vec{p})\defeq
	\set{X\in\P:\closure{X}\ni\vec{p}}.
\]
An alternative expression of the former is as follows. This is justified in \Autoref{:alternate-neighborhood-defn} at the end of this section.
\[
    \mathcal{N}_{\epsilon}(\vec{p})=\set{\member(\vec{x}):\vec{x}\in\closure{B}_{\epsilon}(\vec{p})}.
\]


\item If the partition $\P$ of $\R^{d}$ is understood, we define a function
$\member:\R^{d}\to\P$ mapping each point $\vec{x}\in\R^{d}$ to the
unique member
containing it. Often this is denoted using the equivalence class notation of
$[\vec{x}]$, but in this paper we shall prefer to use the $\member$ function
notation for clarity since we will in general not be thinking of the partition
as being defined by an equivalence relation.

\item When we say ``countable'' we mean finite or countably infinite.

\end{itemize}

\begin{lemma}\label{:alternate-neighborhood-defn}
\[
	\mathcal{N}_{\epsilon}(\vec{p})=
	\set{\member(\vec{x}):\vec{x}\in\closure{B}_{\epsilon}(\vec{p})}.
\]
\end{lemma}
\begin{proof}
If $X\in\mathcal{N}_{\epsilon}(\vec{p})$, then $X\cap\closure{B}_{\epsilon}(\vec{p})\not=\emptyset$, so let $\vec{x}_{0}\in X\cap\closure{B}_{\epsilon}(\vec{p})$ which means $\member(\vec{x}_{0})=X$, so $X\in\set{\member(\vec{x}):\vec{x}\in\closure{B}_{\epsilon}(\vec{p})}$. Conversely, if $X\in\set{\member(\vec{x}):\vec{x}\in\closure{B}_{\epsilon}(\vec{p})}$ then there exists $\vec{x}_{0}\in\closure{B}_{\epsilon}(\vec{p})$ such that $X=\member(\vec{x}_{0})$, so in particular $\vec{x}_{0}\in X$ and thus $X\cap\closure{B}_{\epsilon}(\vec{p})\not=\emptyset$ so $X\in\mathcal{N}_{\epsilon}(\vec{p})$.
\end{proof}
\section{Definitions and Basic Partition Results}
\label{sec:defn-basics}

In this section, we will provide a variety of fairly basic claims about hypercube partitions that we will need later. We also introduce an example of a reclusive partition in order to present some geometric intuition before moving into the linear algebra perspective in the next section.


We first define a natural notion of adjacency for any partition of $\R^{d}$.
By definition, no two members of a partition contain the same element, but if we consider the closures of the members, then this need not hold.

\begin{definition}[Adjacent]\label{:defn-adj}
  Let $d\in\N$, and $\P$ be a partition of $\R^{d}$, and $X,Y\in\P$.
  We say that $X$ and $Y$ are adjacent if the intersection of their closures is non-empty.
  That is, we say $X$ and $Y$ are adjacent if $\clos{X}\cap\clos{Y}\not=\emptyset$.
  We denote this relation as $X\adj Y$.
\end{definition}


  As defined, adjacency is a reflexive relation---for any partition $\P$, for all $X\in\P$, we have $X\adj X$. Also note that despite the notation, adjacency is not an equivalence relation as it is not transitive in general.

Every partition of $\R^{d}$ naturally induces a graph in which adjacency in the graph aligns with adjacency in the partition.

\begin{definition}[Partition Graph]\label{:partition-graph}
  Let $d\in\N$ and $\P$ be a partition of $\R^{d}$. The partition graph of $\P$ is the graph $G$ whose vertex set is the set $\P$, and whose edge set contains the edge $\set{X,Y}$ if and only if $X\adj Y$ in the partition.
\end{definition}

Because adjacency is reflexive, every vertex in $G$ has a self-loop. However, this will not be important for us. We will usually not talk explicitly of this graph and will instead identify the partition and its graph structure.
We will frequently need to prove that two members of a partition are not adjacent, and we will use the following simple corollary.

\begin{corollary}[Non-Adjacency Corollary]\label{:non-adjacency}
  Let $d\in\N$, and $\P$ be a partition of $\R^{d}$, and $X,Y\in\P$.
  If for all $\vec{x}\in\clos{X}$ and $\vec{y}\in\clos{Y}$ it holds that $d_{max}(\vec{x},\vec{y})>0$ then $X$ and $Y$ are not adjacent.
\end{corollary}
\begin{proof}
  For proof by contrapositive, if $X$ and $Y$ are adjacent, then there exists $\vec{z}\in\clos{X}\cap\clos{Y}$, so let $\vec{x}=\vec{y}=\vec{z}$ and we have $d_{max}(\vec{x},\vec{y})=0$.
\end{proof}



We will be particularly interested in partitions of translated unit hypercubes as defined below.

\begin{definition}[Unit Hypercubes and Unit Hypercube Partitions]\label{:defn-unit-hypercube-partition}
  Let $d\in\N$.
  A set $X\subseteq\R^{d}$ is called a unit hypercube if there exists $\vec{a}\in\R^{d}$ such that $X$ can be expressed as
  \[
    X = \vec{a}+[0,1)^d.
  \]
  If so, this expression is unique, and we call $\vec{a}$ the representative corner of $X$ and denote it by $\rep(X)$.
  A partition $\P$ of $\R^{d}$ is called a unit hypercube partition if each member of $\P$ is a unit hypercube.
\end{definition}

It is worth noting that $\rep(X)\in X$.

  The terminology of our definition may be somewhat deceptive because what we call a unit hypercube partition requires these hypercubes to be translations of $[0,1)^d$.
  For example, if $\R^{2}$ is partitioned by a unit grid and the partition is rotated, then it would not be considered a unit hypercube partition by our definition. As mentioned before, some other sources
consider the closed unit hypercube
$\vec{a}+[0,1]^{d}$ or rotations of this set. If
we need a closed unit hypercube, we mention so
explicitly.

The following simple lemma shows that shifting a unit hypercube changes the representative corner in the natural way.

\begin{lemma}[Representative Corner Shift Lemma]\label{:rep-corner-shift}
  Let $d\in\N$, and $X\subset\R^{d}$ be a unit hypercube.
  Then for any vector $\vec{v}\in\R^{d}$,
  \[
    \rep(X+\vec{v})=\rep(X)+\vec{v}.
  \]
\end{lemma}
\begin{proof}
  This holds because if $X=\vec{a}+[0,1)^d$, then $X+\vec{v} = (\vec{a}+\vec{v})+[0,1)^d$.
\end{proof}

While the defined notion of adjacency may be the most natural, it will not be the easiest to work with, so the next two lemmas together give an equivalent characterization of adjacency of members of a unit hypercube partition.
The first lemma states that in a unit hypercube partition it is impossible to have two distinct unit hypercubes $X$ and $Y$ with representative corners closer than a distance of $1$ of each other, because otherwise the hypercubes would overlap.

\begin{lemma}[Non-Overlapping Lemma]\label{:non-overlapping} 
  Let $d\in\N$, and $\P$ be a unit hypercube partition of $\R^{d}$, and $X,Y\in\P$.
  Then either $X=Y$ or $d_{max}(\rep(X),\rep(Y))\geq1$.
\end{lemma}
\begin{proof}
  If $d_{max}(\rep(X),\rep(Y))<1$ then for all $i\in[d]$, $\abs{\rep(X)_{i}-\rep(Y)_{i}}<1$.
  Thus, for all $i\in[d]$, either
  \[
    \rep(X)_{i}\in[\rep(Y)_{i},\rep(Y)_{i}+1)\tag{if $\rep(X)_{i}\geq\rep(Y)_{i}$}
  \]
  or
  \[
    \rep(Y)_{i}\in[\rep(X)_{i},\rep(X)_{i}+1)\tag{if $\rep(X)_{i}\leq\rep(Y)_{i}$}.
  \]
  In either case, this implies
  \[
    [\rep(X)_{i},\rep(X)_{i}+1)\cap[\rep(Y)_{i},\rep(Y)_{i}+1)\not=\emptyset.
  \]
  Note that by definition of $\rep(X)$ and $\rep(Y)$, we can express $X$ and $Y$ as
  \begin{align*}
    X &= \prod_{i=1}^{d}[\rep(X)_{i},\rep(X)_{i}+1)\\
    Y &= \prod_{i=1}^{d}[\rep(Y)_{i},\rep(Y)_{i}+1).
  \end{align*}
  This implies that $X\cap Y\not=\emptyset$, and because $\P$ is a partition, this implies $X=Y$.
\end{proof}

The next lemma shows that in the case of of adjacent hypercubes, the above lower bound on the distance between the representative corners is tight.
The intuition is that if $X$ and $Y$ are unit hypercubes, then because the side lengths of the hypercubes are all 1, it is the case that adjacency of $X$ and $Y$ is equivalent to the representatives of $X$ and $Y$ being within a distance of 1 (in the $d_{max}$ metric, which measures the furthest coordinate).

\begin{lemma}[Adjacency by Representatives Lemma]\label{:adjacency-and-representatives} 
  Let $d\in\N$, and $\P$ be a unit hypercube partition of $\R^{d}$, and $X,Y\in\P$.
  Then $X\adj Y$ if and only if $d_{max}(\rep(X),\rep(Y))\leq1$.
\end{lemma}
\begin{proof}
  We begin with the reverse direction.
  If $d_{max}(\rep(X),\rep(Y))\leq1$ then for all $i\in[d]$, $\abs{\rep(X)_{j}-\rep(Y)_{j}}\leq1$.
  By similar reasoning as the previous lemma (but using the closed intervals instead of the half open intervals), we have for all $i\in[d]$,
  \[
    [\rep(X)_{i},\rep(X)_{i}+1]\cap[\rep(Y)_{i},\rep(Y)_{i}+1]\not=\emptyset.
  \]
  Since $X=\prod_{i=1}^{d}[\rep(X)_{i},\rep(X)_{i}+1)$ (and similarly for $Y$), we can express the closures as
  \begin{align*}
    \clos{X} &= \prod_{i=1}^{d}[\rep(X)_{i},\rep(X)_{i}+1]\\
    \clos{Y} &= \prod_{i=1}^{d}[\rep(Y)_{i},\rep(Y)_{i}+1].
  \end{align*}
  Thus $\clos{X}$ and $\clos{Y}$ intersect in each coordinate, so $\clos{X}\cap\clos{Y}\not=\emptyset$.
  Thus, by definition, $X\adj Y$.

  Now we prove the forward direction by contrapositive.
  Suppose $d_{max}(\rep(X),\rep(Y))>1$.
  This implies the existence of some $j\in[d]$ such that $\abs{\rep(X)_{j}-\rep(Y)_{j}}>1$.
  This means
  \[
    [\rep(X)_{j},\rep(X)_{j}+1]\cap[\rep(Y)_{j},\rep(Y)_{j}+1]=\emptyset
  \]
  thus $\clos{X}\cap\clos{Y}=\emptyset$, so $X$ and $Y$ are not adjacent.
\end{proof}

The previous two lemmas give rise to a number of equivalent notions of adjacency in a unit hypercube partition which we now formalize.

\begin{corollary}[Equivalent Definitions of Adjacency]\label{:equiv-defn-adjacency}
  Let $d\in\N$, and $\P$ be a unit hypercube partition of $\R^{d}$, and $X,Y\in\P$.
  The following statements are all equivalent:
  \begin{enumerate}
  \item\label{adj} $X\adj Y$
  \item\label{eq} $X=Y$ or $d_{max}(\rep(X),\rep(Y))=1$
  \item\label{zero-one} $d_{max}(\rep(X),\rep(Y))=0$ or $d_{max}(\rep(X),\rep(Y))=1$
  \item\label{leq} $d_{max}(\rep(X),\rep(Y))\leq1$
  \end{enumerate}

\end{corollary}
\begin{proof}
  That (\ref{adj}) and (\ref{leq}) are equivalent is the statement of the previous lemma.
  Statements (\ref{eq}) and (\ref{zero-one}) are equivalent because $X=Y$ if and only if $\rep(X)=\rep(Y)$ if and only if $d_{max}(\rep(X),\rep(Y))=0$.
  Thus, it suffices to show that (\ref{zero-one}) and (\ref{leq}) are equivalent.

  Trivially, (\ref{zero-one}) implies (\ref{leq}).
  To show that (\ref{leq}) implies (\ref{zero-one}), assume $d_{max}(\rep(X),\rep(Y))\leq1$; so by \Autoref{:non-overlapping}, either $X=Y$ implying $d_{max}(\rep(X),\rep(Y))=0$ implying (\ref{zero-one}), or $d_{max}(\rep(X),\rep(Y))\geq1$ implying $d_{max}(\rep(X),\rep(Y))=1$ (because $d_{max}(\rep(X),\rep(Y))\leq1$ by hypothesis) implying (\ref{zero-one}).



\end{proof}

The following proposition shows that in a hypercube partition, if a collection of hypercubes is a clique in the partition graph (i.e. they are all pairwise adjacent), then that actually implies that there is a point $\vec{p}$ which is at the closure of all of them, and so for any choice of $\epsilon\in(0,\infty)$, $\clos{B}_{\epsilon}(\vec{p})$ will intersect every one of these hypercubes. This does not hold for more general partitions. The usefulness of this result is that it is possible to show that a unit hypercube partition is not $(k,\epsilon$-secluded by finding a clique of size $k+1$. 

\begin{proposition}[Hypercube Cliques]\label{:hypercube-cliques-have-clique-points}
  Let $d,n\in\N$, and $\P$ be a unit hypercube partition of $\R^{d}$, and $X_{1},\ldots,X_{n}\in\P$ such that for all $i,j\in[n]$ $X_{i}\adj X_{j}$. Then there is a point $\vec{p}\in\R^{d}$ such that for all $i\in[n]$, $\vec{p}\in\clos{X_{i}}$.
\end{proposition}
\begin{proof}
Let $X_{1},\ldots,X_{n}$ as stated. Each $\clos{X_{i}}$ can be expresses as $\clos{X_{i}}=\clos{B}_{\frac12}(\vec{x}^{(i)})$ for some $\vec{x}^{(i)}\in\R^{d}$. Because each $X_{i}$ and $X_{j}$ are adjacent, by definition $\clos{X_{i}}\cap\clos{X_{j}}\not=\emptyset$, and because of each these sets is a metric ball of radius $\frac12$, it follows that $d_{max}(\vec{x}^{(i)},\vec{x}^{(j)})\leq1$. It is a property of the $d_{max}$ metric that if a set of vectors are all pairwise within distance $1$ of each other, then there is a point $\vec{p}$ which is within distance $\frac12$ of all of them\footnote{If the set of vectors is $S$, then the point $\vec{p}$ can defined coordinate-wise as $p_{i}=\frac12 (\sup_{\vec{x}\in S}x_{i} - \inf_{\vec{x}\in S}x_{i})$.}, so we can find such a $\vec{p}$ for the set $\vec{x}^{(i)},\ldots,\vec{x}^{(n)}$. Since $\vec{p}$ is distance at most $\frac12$ from each $\vec{x}^{(n)}$, that means $\vec{p}\in\clos{B}_{\frac12}(\vec{x}^{(n)})=\clos{X_{i}}$.
\end{proof}

\section{Reclusive Partitions}
\label{sec:reclusive-lattice-partitions}

In this section we define what we call reclusive partitions. The word reclusive is a synonym of secluded, but we will have a different technical definition. We have used, and will continue to use the term secluded to discuss generic partitions. We will use the term reclusive partition to talk about unit hypercube partitions which  have the very specific linear algebra structure that we develop in this section.

\subsection{An Example Reclusive Partition}
We will shortly be working with partitions of $\R^{d}$ from a linear algebraic perspective because that allows us to state the results very generally.
However, this makes the intuition of the geometry more difficult.
As a partial remedy for this, we first introduce a very specific reclusive partition of $\R^{d}$ for each $d\in\N$---these are the partitions that we first studied, and they are mathematically very convenient to work with.
We will not be interested in them as anything more than an example because the parameter value of $\epsilon$ that they achieve is only $\frac{1}{2^{d-1}}$
and as we have mentioned, some of the reclusive partitions will achieve as large as $\frac{1}{2d}$.
Nonetheless, these partitions capture the essential geometric idea of the construction of the more general reclusive partitions.

The following defines for all $d\in\N$ a partition of $\R^{d}$ which consists solely of half-open/half-closed unit hypercubes.
After presenting the definition, we elaborate on how to interpret it geometrically.

\begin{definition}[$\P_{d}$]\label{:example-partition}
  For each $d\in\N$, define\footnote{By this definition, in the case that $d=1$, $v_{d}=\langle 1 \rangle$.} $\vec{v_{d}}\defeq\langle\frac{1}{2^{d-1}},\ldots,\frac{1}{2^{d-1}},1\rangle$, the vector whose last entry is $1$ and all other $d-1$ entries are $\frac{1}{2^{d-1}}$.

  Define $\P_{1}$, which is a partition of $\R^{1}$, as follows:
  \[
    \P_{1}\defeq\set{[0,1)+n\colon n\in\Z}=\set{[0,1)+n\cdot\vec{v_{1}}\colon n\in\Z}.
  \]
  Then define $\P_{d}$, which is a partition of $\R^{d}$, inductively for all $d\in\N$ with $d>1$ as follows:
  \[
    \P_{d}\defeq\set{B\times[0,1)+n\cdot\vec{v_{d}}\colon n\in\Z,\; B\in\P_{d-1}}
  \]
\end{definition}

The following discussion motivates why we are interested in this partition and how to understand it geometrically.
The first partition ($d=1$) breaks up $\R^{d}=\R^{1}$ into unit intervals which are half open.
This partition has the property that for any fixed point $x\in\R^{1}$, if you consider all points with (Euclidean) distance less than or equal to $1/2^{d}=1/2$ from $x$, all such points belong to at most $d+1=2$ members of the partition $\P_{1}$ (see \Autoref{fig:simple-partition}).

Then consider how the second partition $\P_{2}$ is constructed by first examining only the members constructed when $n=0$.
In this case, each member is $B\times[0,1)$ for some $B\in\P_{1}$.
We think of this as extruding each member of the previous partition one unit into the newest dimension.
Restricted to $n=0$, this would partition $\R^{1}\times[0,1)$, so to capture all elements of $\R^{2}$, we need to make shifts not just for $n=0$ but for every integer.
The last index of $\vec{v_{d}}$ is $1$ to get integer shifts in the newest dimension so that for an arbitrary value of $n$ we get a partition of $\R^{1}\times[n,n+1)$.

Why is it that $\vec{v_{d}}$ is defined as it is?
If we had taken $\vec{v_{d}}=\langle 0,\ldots,0,1\rangle$ for example (so that $\vec{v_{2}}=\langle 0,1\rangle$), the definition above would still produce a partition.
However, it would not have the desired property that for any point $\vec{x}\in\R^{2}$ the points within a distance $1/2^{d}=1/4$ belong to at most $d+1=3$ members of $\P_{2}$.
For example (see \Autoref{fig:simple-partition}), the point $(1,1)$ would be arbitrarily close to points of the following four members: $[0,1)\times[0,1)$, $[0,1)\times[1,2)$, $[1,2)\times[0,1)$, $[1,2)\times[1,2)$.

To get this property, we shift the extrusions by a ``little bit'' in all of the other dimensions too in order to offset the ``seams'' or boundaries between members of the partition.
With each new partition we build from an old one, the amount of shift in each dimension is reduced (the entries in the vector $\vec{v_{d}}$ decrease as $d$ increases) so that shifts aren't ``undone'' by shifting too much and cycling back.
The construction of $\P_{3}$ has a similar intuition, and beyond that, we find it difficult to visualize.
For completeness, we next prove that the $\P_{d}$ are indeed partitions.
The proof will probably not lend insight to the rest of the paper and should be freely skipped.

\begin{claim}\label{:example-partition-proof}
  For each $d\in\N$, $\P_{d}$ is a partition of $\R^{d}$.
\end{claim}
\begin{proof}
  If $d=1$ (for an inductive base case), let $x\in\R^{1}$ be arbitrary and let $n=\floor{x}$ so $x\in[n,n+1)$ and $x\not\in[m,m+1)$ for any $m\not=n$, so $\P_{1}$ partitions $\R^{1}$.

  The inductive case follows similarly.
  Let $\vec{x}=\langle x_{i}\rangle_{i=1}^{d}\in\R^{d}$ be arbitrary.
  We want to prove the existence of unique $n\in\Z$ and $B\in\P_{d-1}$ such that $\vec{x}\in B\times[0,1)+n\cdot\vec{v_{d}}$.
  Note that by necessity $n=\floor{x_{d}}$ so that $x_{d}\in[0,1)+n\cdot 1=[n,n+1)$ (recall that the last coordinate of $\vec{v_{d}}$ is $v_{d_{d}}=1$).
  Then we see that $\vec{x}\in B\times[0,1)+n\cdot\vec{v_{d}}$ if and only if $\vec{x}-n\cdot\vec{v_{d}}\in B\times[0,1)$, and since we have already established the value of $n$, this holds if and only if $\langle x_{i}-n\cdot v_{d_{i}}\rangle_{i=1}^{d-1}\in B$.
  By the inductive hypothesis, there exists a unique $B\in\P_{d-1}$ such that this holds.
  Thus $\P_{d}$ partitions $\R^{d}$.
\end{proof}

Based on our discussion above, we hope we have provided the intuition that each member of any $\P_{d}$ is a unit hypercube with some amount of shift.
The $\P_{d}$ example partitions will be useful to keep in mind as we work with more general unit hypercube partitions.

\subsection{Motivating Properties}
\label{subsec:motivating-reclusive-properties}

While the inductive definition of the running example partitions $\P_{d}$ is useful, it is also useful to consider unit hypercube partitions from another perspective.
One can note that in the partition $\P_{d}$, the representative corner of each unit hypercube is an integer linear combination of the vectors $\vec{v_{1}},\ldots,\vec{v_{d}}$ as defined in \Autoref{:example-partition} (padded with zeros in the trailing entries as necessary which correspond to the higher dimensions).
The set of all integer linear combinations of a set of basis vectors for the vector space $\R^{d}$ is known as a lattice group (it is a group under vector addition).
Viewing the hypercube representatives as points within the lattice group will be useful.
In particular, this gives motivation to look at certain regularly structured unit hypercube partitions by examining a matrix associated with a set of basis vectors of $\R^{d}$.
For example, consider the example partition $\P_{d}$ for $d=5$.
If we embed the vectors $\vec{v_{1}},\ldots,\vec{v_{5}}$ from \Autoref{:example-partition} into $\R^{5}$ (by padding), and use those vectors as the columns of a matrix, then the matrix would be as follows (e.g. the first column is $\vec{v_{1}}$, the second column is $\vec{v_{2}}$, and so on with zeros padded at the end as necessary).
\[
  A=\left[
    \begin{array}{ccccc}
      1 & \frac{1}{2} & \frac{1}{4} & \frac{1}{8} & \frac{1}{16} \\
      0 & 1           & \frac{1}{4} & \frac{1}{8} & \frac{1}{16} \\
      0 & 0           & 1           & \frac{1}{8} & \frac{1}{16} \\
      0 & 0           & 0           & 1           & \frac{1}{16} \\
      0 & 0           & 0           & 0           & 1
    \end{array}
  \right]
\]
In fact, we could equivalently have defined $\P_5$ (and similarly for all $\P_{d}$ from the running example) to be the set of unit hypercubes whose representatives were integer linear combinations of the columns of this matrix; in other words, $\P_{5}$ could have been defined as the set of unit hypercubes whose representatives are given by $A\vec{n}$ for some $\vec{n}\in\Z^{d}$.

In light of this, we shall define a more structured version of unit hypercube partitions by defining them in terms of a matrix.
Observe the following four structural properties of the example matrix $A$ above:
\begin{enumerate}
\item The matrix $A$ for $\P_{5}$ above explicitly contains the structure of the partitions $\P_{d}$ for $d\leq5$ in the sense that the submatrix consisting of the first 4 rows and first 4 columns is the matrix associated with the partition $\P_{4}$. Similarly the submatrix consisting of the first 3 rows and first 3 columns is the matrix associated with the partition $\P_{3}$, and so on.
\item The matrix is upper triangular. The reason is that in the inductive definition of $\P_{d}$, the vector $\vec{v_{d}}$ is in $\R^{d}$, so the lengths of these vectors grows by one with each iteration of the construction, and each iteration of the construction adds a row and column.
\item The diagonal of the matrix is all 1's. This is because in the inductive definition of $\P_{d}$, there has to be a unit shift in the current dimension to accommodate that the members of the partition are unit hypercubes. For example, in the definition of $\P_{1}$, the members are $[0,1)+n\langle 1\rangle$ for all $n\in\Z$. If the vector $\langle 1\rangle$ was anything else, this would not be a partition. In the case of $\P_{2}$, each element is of the form $B\times[0,1)+n\langle\frac{1}{2},1\rangle$. Again, the last index of this vector must be 1 because we are extruding the elements of the previous partition by $[0,1)$. This holds in each dimension of the construction.
\item The entries in each row are strictly decreasing. This is because in the inductive definition of $\P_{d}$, we wanted to offset the ``seams'' of the hypercubes to prevent points in $\R^{n}$ from being ``close'' to too many hypercubes. For example, if we had defined $\P_{d}$ by taking each $\vec{v_{d}}$ to be the vector of all 0's aside from the last entry being 1, then the associated matrix would be the identity matrix. This would indeed generate a unit hypercube partition, but it would not have the property we ultimately desire of limiting the adjacencies (in fact, this is the partition in \Autoref{fig:simple-partition}).

\end{enumerate}

As mentioned earlier, while the $\P_{d}$ serve as nice examples, the exponentially decreasing nature of the shifts ($\frac{1}{2}, \frac{1}{4}, \frac{1}{8}, \frac{1}{16}, \ldots$) leads to a need for exponentially decreasing $\epsilon$ parameter.
Instead, we would like the shifts to change linearly and work with partitions that have matrices more like one of the ones below:

\[
  \left[
    \begin{array}{ccccc}
      1 & \frac{4}{5} & \frac{3}{5} & \frac{2}{5} & \frac{1}{5} \\
      0 & 1           & \frac{3}{5} & \frac{2}{5} & \frac{1}{5} \\
      0 & 0           & 1           & \frac{2}{5} & \frac{1}{5} \\
      0 & 0           & 0           & 1           & \frac{1}{5} \\
      0 & 0           & 0           & 0           & 1
    \end{array}
  \right]
  \qquad\text{or}\qquad
  \left[
    \begin{array}{ccccc}
      1 & \frac{4}{5} & \frac{3}{5} & \frac{2}{5} & \frac{1}{5} \\
      0 & 1           & \frac{4}{5} & \frac{3}{5} & \frac{2}{5} \\
      0 & 0           & 1           & \frac{4}{5} & \frac{3}{5} \\
      0 & 0           & 0           & 1           & \frac{4}{5} \\
      0 & 0           & 0           & 0           & 1
    \end{array}
  \right]
\]

In fact, we will define the matrices of interest in a general enough way to include all of the matrices shown so far.
Notice that both of these two matrices have the same structural properties mentioned above.


We now define the type of matrices that we will be interested in based on the ideas just discussed.
These matrices will be used to define partitions which have the property that hypercubes will be adjacent to a small number of other hypercubes; in the language of graph theory, the members will have few neighbors, so we call them reclusive matrices and reclusive partitions. Stated again for emphasis, we will have a connection between minimizing the parameter $k$ in the motivating question in the introduction, and minimizing the sizes of cliques in the partition graph.

\subsection{Construction}

\begin{definition}[Reclusive Matrix]
  Informally, a square matrix $A$ will be called a reclusive matrix if it is upper triangular, has only $1$'s on the main diagonal, and is strictly decreasing with positive entries in each row starting at the main diagonal.

  More formally, a square $n\times n$ matrix $A=(a_{ij})$ will be called a reclusive matrix if all of the following hold:
  \begin{enumerate}
  \item For all $i>j\in[n]$, $a_{ij}=0$.\hfill(Upper Triangular)
  \item For all $i\in[n]$, $a_{ii}=1$.\hfill($1$'s on Diagonal)
  \item For all $i\leq j<k\in[n]$, $a_{ij}>a_{ik}>0$.\hfill(Decreasing and Non-zero after Diagonal)
  \end{enumerate}
\end{definition}

\begin{remark}
  Reclusive matrices are invertible because they are upper triangular so the determinant is equal to the product of the diagonal entries which are all 1.
\end{remark}

We view $d\times d$ reclusive matrices as a natural way to build partitions of $\R^{d}$ as just discussed.
Before defining these partitions, we formalize the lattice group structure mentioned in the motivating discussion.

\begin{definition}[Lattice Group]\label{:lattice-group}
  For any invertible $d\times d$ matrix $A$, define the set $L_{A}\defeq\set{A\vec{v}\colon \vec{v}\in\Z^{d}}$.
\end{definition}

The set $L_{A}$ is a group under vector addition.
Since $A$ is an invertible linear map, it is actually an isomorphism between $L_{A}$ and $\Z^{d}$.

\begin{definition}[Reclusive Partition]
  If $A$ is a $d\times d$ reclusive matrix, then we associate to it a partition $\P_{A}$,
  called the reclusive partition for $A$, where
  \[
    \P_{A}\defeq\set{\vec{a}+[0,1)^{d}\colon\vec{a}\in L_{A}}.
  \]
\end{definition}

\begin{remark}
Some sources use the notation $L_{A}+[0,1)^{d}$ to indicate the set $\P_{A}$, but we will elect to not do so here.
\end{remark}

Notice that every member of $\P_{A}$ is a unit hypercube shifted by an element of the lattice group so that for all $X\in\P_{A}$, $\rep(X)\in L_{A}$ and conversely, for every $\vec{a}\in L_{A}$ there is a unit hypercube $X$ whose representative is $\vec{a}$.
The proof that this is a partition will not be presented as it is a direct consequence of \Autoref{:efficent-computation-of-representatives} which shows that for any element $\vec{x}\in\R^{d}$ there is a unique $X\in\P$ such that $\vec{x}\in X$ by explicitly computing the representative $\rep(X)$.

We will have use for yet another equivalent notion of adjacency, but to do so we must introduce two definitions for types of finite sequences (which we will apply to vectors). These defined sequences will play the essential role in proving that the reclusive partitions have the properties that we are looking for.

\begin{definition}[Alt-1 and Weak-Alt-1 Sequences]
  A finite sequence $\langle c_{i}\rangle_{i=1}^{n}$ is called alt-1 (alternating sequence of magnitude 1) if $\langle c_{i}\rangle_{i=1}^{n}=\langle (-1)^{i} \rangle_{i=1}^{n}$ or $\langle c_{i}\rangle_{i=1}^{n}=\langle (-1)^{i+1} \rangle_{i=1}^{n}$.

  A finite sequence $\langle c_{i}\rangle_{i=1}^{n}$ is called weak-alt-1 (weakly alternating sequence of magnitude 1) if all terms are $-1$, $0$, or $1$, and the subsequence $\langle c_{i_{j}} \rangle_{j=1}^{k}$ of non-zero terms is an alt-1 sequence.
\end{definition}

\begin{remark}
We consider the empty sequence to vacuously satisfy these definitions, so in
particular, a finite sequence of all zeros is considered weak-alt-1.
\end{remark}

The next lemma is more or less an adaption of the alternating sequence convergence theorem from calculus (c.f. \cite{Johnsonbaugh1981FoundationsOM}).
If we take a dot product of an alt-1 sequence and a strictly monotonic positive sequence, then we know what the sign of that dot product will be and can bound the magnitude.

\begin{lemma}\label{:alt-1-sums}
  Let $\langle c_{i}\rangle_{i=1}^{n}$ be an alt-1 sequence.
  Let $\langle a_{i} \rangle_{i=1}^{n}$ be a strictly decreasing (resp. strictly increasing) positive sequence.
  Letting $s=\sum_{i=1}^{n}c_{i}a_{i}$, the following hold:
  \begin{enumerate}
  \item\label{eq:first} If $n=1$ then $\abs{s}=a_{1}$ (resp. $\abs{s}=a_{n}$)
  \item\label{eq:second} If $n\geq 2$ then $a_{1}-a_{2}\leq\abs{s}\leq a_{1}$ (resp. $a_{n}-a_{n-1}\leq\abs{s}\leq a_{n}$)
  \item\label{eq:fourth} $\abs{s}>0$
  \item\label{eq:third} $\sign(s)=\sign(c_{1})$ (resp. $\sign(s)=\sign(c_{n})$)
  \end{enumerate}
\end{lemma}
\begin{proof}
  Note that (\ref{eq:fourth}) is implied by (\ref{eq:first}) and implied by (\ref{eq:second}) since $\langle a_{i}\rangle$ is strictly decreasing (resp. strictly increasing), so (\ref{eq:fourth}) need not be proven.

  We prove the ``increasing'' version of the statement because the inductive indexing is cleaner.
  This immediately implies the stated ``decreasing'' version by reversing $\langle a_{i}\rangle$ and reversing $\langle c_{i}\rangle$.

  We prove the ``increasing'' version by induction on $n$.
  If $n=1$, the claim holds trivially.
  Otherwise $n>1$ and assume for inductive hypothesis that the lemma holds for $n-1$.
  Then let
  \[
    s\defeq \sum_{i=1}^{n}c_{i}a_{i}\qquad\text{and}\qquad s'\defeq \sum_{i=1}^{n-1}c_{i}a_{i}\qquad\text{noting that}\qquad s = s' + c_{n}a_{n}.
  \]
  Thus, we have
  \begin{align*}
    s &=s' + c_{n}a_{n}\\
      &=\sign(s')\abs{s'} + c_{n}a_{n}\tag{Decomposition of $s'$}\\
      &=\sign(c_{n-1})\abs{s'} + c_{n}a_{n}\tag{Inductive hypothesis}\\
      &=-c_{n}\abs{s'} + c_{n}a_{n}\tag{$\sign(c_{n-1})=c_{n-1}=-c_{n}$ by alt-1 def'n}\\
      &=c_{n}(a_{n}-\abs{s'}).
  \end{align*}
  Taking the magnitude we have
  \begin{align*}
    \abs{s} &= \abs{c_{n}}\abs{a_{n}-\abs{s'}}\\
            &= \abs{a_{n}-\abs{s'}}\\
            &= a_{n}-\abs{s'}.\tag{$a_{n}$ and $\abs{s'}$ both non-negative}
  \end{align*}
  Since $\abs{s'}$ is non-negative, it follows that $\abs{s}\leq a_{n}$.
  Further, by inductive hypothesis, $\abs{s'}\leq a_{n-1}$ so again by the last line above, $\abs{s}=a_{n}-\abs{s'}\geq a_{n}-a_{n-1}$ which proves (\ref{eq:second}).

  Lastly, noting again that $a_{n}>a_{n-1}\geq\abs{s'}$ we have
  \begin{align*}
    \sign(s) &= \sign(c_{n}(a_{n}-\abs{s'}))\\
             &= \sign(c_{n})\sign(a_{n}-\abs{s'})\\
             &= c_{n}\cdot 1\tag{$c_{n}=\sign(c_{n})$ and $a_{n}-\abs{s'}>0$}
  \end{align*}
  which proves (\ref{eq:third}).
\end{proof}

The reason we required $\langle a_{i}\rangle$ to be strictly monotonic was because otherwise $\sign(c_{n})$ could be $0$.
The above lemma extends very naturally to weak-alt-1 sequences which have at least one non-zero term by applying the lemma to the subsequence of non-zero terms and the corresponding entries of $\langle a_i \rangle$ and we shall sometimes use it as such.



The following proposition will be the key to establishing the (final) equivalent definition of adjacency.
The $A\vec{c}$ in the statement will end up being $\rep(X)-\rep(Y)$ for unit hypercubes $X$ and $Y$, so this proposition will give an equivalent condition for $\norm{\rep(X)-\rep(Y)}_{\infty}\leq1$ (which by \Autoref{:equiv-defn-adjacency} is equivalent to $X$ and $Y$ being adjacent).

\begin{proposition}\label{:adj-norm-equiv}
  Let $A$ be a $d\times d$ reclusive matrix and $\vec{c}\in\Z^{d}$ (emphasis: $\vec{c}$ has integer coordinates).
  Then $\norm{A\vec{c}}_{\infty}\leq1$ if and only if $\vec{c}=\langle c_{j}\rangle_{j=1}^{d}$ is a weak-alt-1 sequence.
\end{proposition}
\begin{proof}
  Before proving either direction, let $\vec{x}=A\vec{c}$.
  Note that for any $i\in[d]$
  \begin{align*}
    x_{i} &= \sum_{j=1}^{d}a_{ij}c_{j}\tag{Def'n of matrix multiplication}\\
          &= \sum_{j=i}^{d}a_{ij}c_{j}\tag{If $j<i$ then  $a_{ij}=0$}.
  \end{align*}
  Note that $\langle a_{ij}\rangle_{j=i}^{d}$ is a strictly decreasing positive sequence.

  We begin by proving the reverse direction.
  If $\langle c_{j}\rangle_{j=1}^{d}$ is a weak-alt-1 sequence, then for any $i\in[d]$, the subsequence $\langle c_{j}\rangle_{j=i}^{d}$ is also a weak-alt-1 sequence.
  Then by \Autoref{:alt-1-sums} (applied to the subsequence of non-zero terms) and the expression of $\vec{x}$ above, $\abs{x_{i}}\leq a_{ii}=1$.
  Since this holds for all $i\in[d]$, $\norm{\vec{x}}\leq1$.

  For the reverse direction, assume that $\langle c_{j}\rangle_{j=1}^{d}$ is not a weak-alt-1 sequence, in which case we let $k\in[d]$ be the largest integer such that the subsequence $\langle c_{j}\rangle_{j=k}^{d}$ is not weak-alt-1.
  By the above expression of $\vec{x}$ we have
  \begin{align*}
    \norm{\vec{x}}_{\infty} &= \max_{i\in[d]}\abs{x_{i}}\tag{Def'n}\\
                            &\geq \abs{x_{k}}\\
                            &= \abs{\sum_{j=k}^{d}a_{kj}c_{j}}\tag{By expression of $x_{k}$}\\
                            &= \abs{c_{k} + \sum_{j=k+1}^{d}a_{kj}c_{j}}\tag{$A$ is reclusive, so $a_{kk}=1$}
  \end{align*}
  If $\abs{c_{d}}>1$ then $\langle c_{j}\rangle_{j=d}^{d}$ is not weak-alt-1, so $k=d$.
  Then the the above summation is empty, so this shows $\norm{\vec{x}}_{\infty}\geq \abs{c_{d}}>1$ and we are done in this case.
  Otherwise we may assume $\abs{c_{d}}\leq 1$; further, since $\vec{c}\in\Z^{d}$ this means $c_{d}$ is $-1$, $0$, or $1$ and thus $\langle c_{j}\rangle_{j=d}^{d}$ is trivially weak-alt-1 which implies $k\not=d$.
  Then the sequence $\langle c_{j}\rangle_{j=k+1}^{d}$ is non-empty (because $k\not=d$) and is weak-alt-1 (by design of $k$) and contains at least one non-zero term (because otherwise it would be trivially weak-alt-1)---let $n$ denote the number of non-zero terms.
  Let $\langle j_{i} \rangle_{i=1}^{n}$ be the sequence of non-zero terms of $\langle c_{j}\rangle_{j=k+1}^{d}$.
  Then $\langle c_{j_{i}}\rangle_{i=1}^{n}$ is an alt-1 sequence, and $\langle a_{kj_{i}}\rangle_{i=1}^{n}$ is a strictly decreasing sequence (since it is a subsequence of $\langle a_{kj}\rangle_{j=k+1}^{d}$ which is strictly decreasing because $A$ is reclusive).
  Thus, \Autoref{:alt-1-sums} applies to $s=\sum_{j=k+1}^{d}a_{kj}c_{j}=\sum_{i=1}^{n}a_{kj_{i}}c_{j_{i}}$ (with different cases if $n=1$ or $n\geq2$).
  We complete the proof with cases on the magnitude of $c_{k}$ (and handle the subcases of the value of $n$ as needed).

  Recall that $c_{k}\in\Z$ and note that $\abs{c_{k}}\not=0$ because $\langle c_{j}\rangle_{j=k+1}^{d}$ is weak-alt-1, so if $c_{k}=0$, then $\langle c_{j}\rangle_{j=k}^{d}$ would be weak-alt-1, but it is not by choice of $k$.
  So we consider two cases: $\abs{c_{k}}=1$ and $\abs{c_{k}}\geq2$.

  Case 1: If $\abs{c_{k}}\geq2$, then by \Autoref{:alt-1-sums} $\abs{s}\leq a_{kj_{1}}$ (regardless of the value of $n$).
    This gives the following inequalities:
    \begin{align*}
      \norm{\vec{x}}_{\infty} &\geq \abs{c_{k}+s}\tag{Work above}\\
                              &\geq \abs{c_{k}} - \abs{s}\tag{Triangle inequality}\\
                              &\geq 2 - \abs{s}\tag{Assumption on $\abs{c_{k}}$}\\
                              &\geq 2 - a_{kj_{1}}\tag{\Autoref{:alt-1-sums}}\\
                              &= 1 + (1 - a_{kj_{1}})\\
                              &> 1\tag{$j_{1}\geq k+1$ so $a_{kj_{1}}<1$}
    \end{align*}

    Case 2: If $\abs{c_{k}}=1$, then because $\langle c_{j}\rangle_{j=k}^{d}$ is not weak-alt-1, it implies that $\sign(c_{k})=\sign(c_{j_{1}})$.
    We then get the following inequalities:
    \begin{align*}
      \norm{\vec{x}}_{\infty} &\geq \abs{c_{k}+s}\tag{Work above}\\
                              &= \abs{c_{k}} + \abs{s}\tag{Same sign}\\
                              &= 1 + \abs{s}\tag{Assumption on $\abs{c_{k}}$}
    \end{align*}
    From here we have two cases depending on if $n=1$ or $n\geq2$.
    If $n=1$, then by \Autoref{:alt-1-sums} $\abs{s}=a_{kj_{1}}$, so from the above we have
    \begin{align*}
      \norm{\vec{x}}_{\infty} &\geq 1 + a_{kj_{1}}\\
                              &> 1\tag{$j_{1}\geq k+1$ so $a_{kj_{1}}>0$}
    \end{align*}
    If instead $n\geq2$, then by \Autoref{:alt-1-sums} $\abs{s}\geq a_{kj_{1}}-a_{kj_{2}}$, so from the above we have
    \begin{align*}
      \norm{\vec{x}}_{\infty} &\geq 1 + (a_{kj_{1}}-a_{kj_{2}})\\
                              &> 1\tag{$j_{2}>j_{1}\geq k+1$ so $a_{kj_{1}}-a_{kj_{2}}>0$}
    \end{align*}
\end{proof}

Upon inspection, one may note that we can improve upon the statement of the prior proposition.
In the proof above, if $\vec{c}$ was not a weak-alt-1 sequence, then not only was $\norm{A\vec{c}}_{\infty}>1$, but it could be bounded away from $1$.
This should not be surprising since $\set{A\vec{v}\colon \vec{v}\in\Z^{d}}$ is isomorphic to $\Z^{d}$ (as stated in the discussion of \Autoref{:lattice-group}).
Specifically, if $\vec{c}$ was not weak-alt-1, the above proof showed that one of the following four equations held:
\begin{align*}
  \norm{A\vec{c}}_{\infty}&\geq1 + 1\tag{$\abs{c_{d}}>1$}\\
  \norm{A\vec{c}}_{\infty}&\geq1 + (1-a_{kj_{1}})\tag{$\abs{c_{d}}\leq1$ and $\abs{c_{k}}\geq2$}\\
  \norm{A\vec{c}}_{\infty}&\geq1 + (a_{kj_{1}})\tag{$\abs{c_{d}}\leq1$ and $\abs{c_{k}}=1$ and $n=1$}\\
  \norm{A\vec{c}}_{\infty}&\geq1 + (a_{kj_{1}}-a_{kj_{2}})\tag{$\abs{c_{d}}\leq1$ and $\abs{c_{k}}=1$ and $n\geq2$}
\end{align*}
If we ignore the specifics of how $k$, $j_{1}$, and $j_{2}$ were found, and minimize over all possibilities, it leads to the following definition and corollary.

\begin{definition}[Reclusive Distance]\label{:reclusive-dist-def}
  If $A$ is a $d\times d$ reclusive matrix, define $\Delta_{A}$, the reclusive distance of $A$, as follows (taking $\min\emptyset = \infty$)
  \begin{align*}
    \delta_{1} &\defeq 1\\
    \delta_{2} &\defeq \min_{k\in[d]}\min_{k<j\leq d}(1-a_{kj})\\
    \delta_{3} &\defeq \min_{k\in[d]}\min_{k<j\leq d}(a_{kj})\\
    \delta_{4} &\defeq \min_{k\in[d]}\min_{k<j<j'\leq d}(a_{kj}-a_{kj'})\\\\
    \Delta_{A} &\defeq \min\set{\delta_{1},\delta_{2},\delta_{3},\delta_{4}}
  \end{align*}
\end{definition}

Observe that because $A$ is reclusive, $\Delta_{A}>0$.

\begin{corollary}\label{:reclusive-dist-cor}
  Let $A$ be a $d\times d$ reclusive matrix and $\vec{c}\in\Z^{d}$ (emphasis: $\vec{c}$ has integer coordinates).
  Then $\norm{A\vec{c}}_{\infty}\leq1$ if and only if $\vec{c}=\langle c_{j}\rangle_{j=1}^{d}$ is a weak-alt-1 sequence.
  Further, if $\norm{A\vec{c}}_{\infty}>1$, then $\norm{A\vec{c}}_{\infty}\geq1+\Delta_{A}$.
\end{corollary}
\begin{proof}
  The proof is implicit in the proof of \Autoref{:adj-norm-equiv}.
\end{proof}


A simple application of this corollary shows that if $X$ and $Y$ are non-adjacent hypercubes in a reclusive partition, then the distance between the representatives of the two hypercubes are separated by one plus the reclusive distance of the partition.

\begin{lemma}[Adjacent or Far Lemma (Representatives)]\label{:adjacent-or-far-rep}
  Let $d\in\N$, and $A$ be a $d\times d$ reclusive matrix, and $\P_{A}$ its reclusive partition, and $\Delta_{A}$ its reclusive distance.
  Let $X,Y\in\P_{A}$ such that $X$ and $Y$ are not adjacent.
  Then $d_{max}(\rep(X),\rep{Y})\geq 1+\Delta_{A}$.
\end{lemma}
\begin{proof}
  By definition of $\P_{A}$, there exists $\vec{m},\vec{n}\in\Z^{d}$ such that $A\vec{m}=\rep(X)$ and $A\vec{n}=\rep(Y)$ (in particular $\vec{m}=A^{-1}\rep(X)$ and $\vec{n}=A^{-1}\rep(Y)$).
  By \Autoref{:equiv-defn-adjacency}, because $X$ and $Y$ are not adjacent, $d_{max}(\rep(X),\rep(Y))>1$, so
  \begin{align*}
    1 &< d_{max}(\rep(X),\rep(Y))\\
      &= \norm{\rep(X)-\rep(Y)}_{\infty}\\
      &= \norm{A\vec{m}-A\vec{n}}_{\infty}\\
      &= \norm{A(\vec{m}-\vec{n})}_{\infty}.
  \end{align*}
  By \Autoref{:reclusive-dist-cor}, since $\norm{A(\vec{m}-\vec{n})}>1$ it must be that $\norm{A(\vec{m}-\vec{n})}\geq1+\Delta_{A}$.
\end{proof}

A similar result holds when considering general elements $\vec{x},\vec{y}$ of non-adjacent hypercubes $X$ and $Y$.
We can get this result because the location of $\vec{x}$ relative to $\rep(X)$ is similar to the location of $\vec{y}$ relative to $\rep(Y)$.

\begin{lemma}[Adjacent or Far Lemma (Points)]\label{:adjacent-or-far}
  Let $d\in\N$, and $A$ be a $d\times d$ reclusive matrix, and $\P_{A}$ its reclusive partition, and $\Delta_{A}$ its reclusive distance.
  Let $X,Y\in\P_{A}$ such that $X$ and $Y$ are not adjacent.
  For all $\vec{x}\in X$ and $\vec{y}\in Y$, it holds that $d_{max}(\vec{x},\vec{y})>\Delta_{A}$.
\end{lemma}
\begin{proof}
  By the definition of unit hypercubes, $\vec{x}\in\rep(X)+[0,1)^{d}$, so let $\vec{\alpha}\in[0,1)^{d}$ such that $\vec{x}=\rep(X)+\vec{\alpha}$.
  Similarly, let $\vec{\beta}\in[0,1)^{d}$ such that $\vec{y}=\rep(Y)+\vec{\beta}$.
  As in the proof of the prior lemma, let $\vec{m}=A^{-1}\rep(X)$ and $\vec{n}=A^{-1}\rep(Y)$.

  Since $\vec{\alpha},\vec{\beta}\in[0,1)^{d}$, it follows that $\vec{\beta}-\vec{\alpha}\in(-1,1)^{d}$ so $\norm{\vec{\beta}-\vec{\alpha}}_{\infty}<1$.
  Then we have the following:
  \begin{align*}
    d_{max}(\vec{x},\vec{y}) &= \norm{\vec{x}-\vec{y}}_{\infty}\\
                             &= \norm{(A\vec{m}+\vec{\alpha}) - (A\vec{n}+\vec{\beta})}_{\infty}\\
                             &= \norm{(A(\vec{m}+\vec{n})) - (\vec{\beta}-\vec{\alpha})}_{\infty}\\
                             &\geq \norm{(A(\vec{m}+\vec{n}))}_{\infty} - \norm{\vec{\beta}-\vec{\alpha}}_{\infty}\tag{Triangle inequality}\\
                             &\geq (1+\Delta_{A}) - \norm{(\vec{\beta}-\vec{\alpha})}_{\infty}\tag{By previous lemma}\\
                             &> (1+\Delta_{A}) - 1\tag{$\norm{\vec{\beta}-\vec{\alpha}}_{\infty}<1$}\\
                             &= \Delta_{A}
  \end{align*}
  Noting the strict inequality in the second to last line completes the proof.
\end{proof}

This lemma will be important later when we need a fixed bound on the distances between non-adjacent partition members.


Now that we have given a bound on how close non-adjacent hypercubes can be, we want to turn our attention to how many hypercubes can be pairwise adjacent.
In other words, we want to show a bound on the size of the largest clique in the partition graph of a reclusive partition.
We actually do something stronger and give an explicit coloring of the graph (an explicit coloring of the hypercubes).
If a graph can be properly colored with $n$ colors, then the size of the largest clique in the graph is at most $n$.




\begin{theorem}[Coloring Reclusive Partitions]\label{:coloring}
  Let $d\in\N$ and let $A$ be a $d\times d$ reclusive matrix. The graph of the reclusive partition $\P_{A}$ can be properly $(d+1)$-colored.
\end{theorem}
\begin{proof}
  Let $\chi=\left[1,2,\ldots,d\right]$ be a $1\times d$ matrix.
  Define the coloring function on the hypercubes as follows: 
  \begin{align*}
    &\mycolor:\P_{A}\to\Z_{d+1}\\
    &\mycolor(X)=\chi\;A^{-1}\;\rep(X)\mod(d+1).
  \end{align*}
  Recall that if $X\in\P_{A}$, then $\rep(X)=A\vec{m}$ for some $\vec{m}\in\Z^{d}$ (by definition of the reclusive partition $\P_{A}$), and $\vec{m}$ is unique because $A$ is invertible, so the definition of $\mycolor$ has an appropriate codomain because $(A^{-1}(A\vec{m}))=\vec{m}\in\Z^{d}$, and taken as a column vector it can be multiplied with $\chi$ to obtain an integer.

  Let $X$ and $Y$ be distinct hypercubes in $\P_{A}$ such that $X\adj Y$.
  We must show that $\mycolor(X)\not=\mycolor(Y)$.
  We do so by looking at the differences of the colors (and explicitly emphasize that it is being done $\mod (d+1)$).
  Let $\vec{m}=A^{-1}\rep(X)$ and $\vec{n}=A^{-1}\rep(Y)$.
  \begin{align*}
    \mycolor(X) - \mycolor(Y) \mod (d+1) &= \big[ (\chi\vec{m} \mod (d+1)) - (\chi\vec{n} \mod (d+1)) \big] \mod (d+1)\tag{definition}\\
                                         &= \big[ \chi\vec{m} - \chi\vec{n} \big] \mod (d+1)\tag{property of modular arithmetic}\\
                                         &= \big[ \chi(\vec{m}-\vec{n}) \big] \mod (d+1)\tag{linearity of $\chi$}
  \end{align*}
  Observe that by \Autoref{:adj-norm-equiv}, because $X\adj Y$ it follows that $(\vec{m}-\vec{n})$ is a weak-alt-1 sequence/vector.
  Further, $(\vec{m}-\vec{n})$ has at least one non-zero term (if it did not, then $\vec{m}=\vec{n}$ so $\rep(X)=\rep(Y)$ but we assumed $X$ and $Y$ were distinct).
  Note also that the matrix product of $\chi$ with $(\vec{m}-\vec{n})$ is really a dot product of the increasing positive sequence $\langle 1,2,\ldots,d\rangle$ with a weak-alt-1 sequence with a non-zero term. Thus, by \Autoref{:alt-1-sums} (applied to the (non-empty) subsequence of non-zero terms), we have
  \[
    0<\abs{\chi(\vec{m}-\vec{n})}\leq d.
  \]
  Thus, $\chi(\vec{m}-\vec{n}) \mod (d+1) \not= 0$ which proves that $\mycolor(X)\not=\mycolor(Y)$, so this is a proper coloring.
\end{proof}

In fact, the coloring above is tight---the chromatic number (the smallest number of colors that can be used to color the graph) is $d+1$.
To prove this, it suffices to show that $\P_{A}$ has a $(d+1)$-clique.

The following result actually follows as a corollary to the Optimality
Theorem for Unit Hypercube Partitions (\Autoref{:hypercube-partition-thm}), so the following proof is not strictly necessary; however, \Autoref{:hypercube-partition-thm} uses far more machinery, and it is fairly simple to find an explicit clique in the
partition.

\begin{proposition}[Chromatic Number]\label{:max-clique-tight}
  Let $d\in\N$, and $A$ be a $d\times d$ reclusive matrix, and $\P_{A}$ its reclusive partition.
  There exists a clique in $\P_{A}$ of size $d+1$.
\end{proposition}
\begin{proof}
  For each $i\in[d]$, let $\vec{e_{i}}\in\Z^{d}$ denote the $i$th standard basis vector (i.e. all zeros except for the $i$th term which is $1$).
  Let $V=\set{\vec{0}}\cup\set{\vec{e_{i}}\colon i\in[d]}$.
  Note that for any two vectors $\vec{v},\vec{w}\in V$, that $\vec{v}-\vec{w}$ is a weak-alt-1 vector.
  This is obvious if one of the vectors is $\vec{0}$ because then the difference is a standard basis vector which is all zeros except for a single term which is $1$, so it is weak-alt-1.
  Otherwise, $\vec{v}=\vec{e_{i}}$ and $\vec{w}=\vec{e_{j}}$ for some $i,j\in[d]$.
  If $i=j$, then the difference is $\vec{0}$ which is trivially weak-alt-1, and if $i\not=j$, then the difference is all zeros except for a term which is $+1$ and a term which is $-1$; such a vector must be a weak-alt-1 sequence.

  Let $R = \set{A\vec{v}\colon \vec{v}\in V}$ denote a set of representatives, and let $C = \set{X\in \P_{A}\colon\rep(X)\in R}$ be the set of unit hypercubes in $\P_{A}$ whose representatives are in $R$.
  Then $C$ (and $V$) have cardinality $d+1$, and we claim that $C$ is a clique.

  Consider hypercubes $X,Y\in C$, so $\rep(X),\rep(Y)\in R$, so $A^{-1}\rep(X), A^{-1}\rep(Y)\in V$.
  Let $\vec{m}=A^{-1}\rep(X)$ and $\vec{n}=A^{-1}\rep(Y)$.
  As described, since $\vec{m},\vec{n}\in V$, we have that $\vec{m}-\vec{n}$ is a weak-alt-1 sequence.
  By \Autoref{:adj-norm-equiv} this implies that $\norm{A(\vec{m}-\vec{n})}_{\infty}\leq1$.
  Thus,
  \begin{align*}
    d_{max}(\rep(X),\rep(Y)) &= \norm{\rep(X)-\rep(Y)}_{\infty}\\
                             &= \norm{A(\vec{m}-\vec{n})}_{\infty}\\
                             &\leq1
  \end{align*}
  and by \Autoref{:equiv-defn-adjacency}, $X$ and $Y$ are adjacent.
  Since this holds for any two hypercubes in $C$, it must be that $C$ is a $(d+1)$-clique.
\end{proof}


The following theorem is the main result of this section showing that
we have constructed partitions satisfying the motivating question with
$k=d+1$ and $\epsilon=\frac{\Delta_{A}}{2}$. Then the Hypercube Parition
Theorem (\Autoref{:hypercube-partition-thm}) mentioned in 
\Autoref{sec:results} will be a simple corollary of
this by demonstrating the existence of a reclusive partition with
$\Delta_{A}=\frac{1}{d}$.

\begin{theorem}[Partition Theorem]\label{:partition-thm}
  Let $d\in\N$, and $A$ be a $d\times d$ reclusive matrix, and $\P_{A}$ its reclusive partition, and $\Delta_{A}$ its reclusive distance.
  Then for any $\vec{p}\in\R^{d}$,
  \[
    \abs{\mathcal{N}_{\Delta_{A}/2}(\vec{p})}\leq d+1.
  \]
\end{theorem}
\begin{proof}
  It suffices to prove that $\mathcal{N}_{\Delta_{A}/2}(\vec{p})$ is a clique since by \Autoref{:coloring}, any clique contains at most $d+1$ hypercubes.
  Let $X,Y\in\mathcal{N}_{\Delta_{A}/2}(\vec{p})$ be arbitrary.
  By definition of $\mathcal{N}_{\Delta_{A}/2}(\vec{p})$, there exists $\vec{x}\in X$ such that $d_{max}(\vec{p},\vec{x})\leq\frac{\Delta_{A}}{2}$ and similarly there exists $\vec{y}\in Y$ such that $d_{max}(\vec{p},\vec{y})\leq\frac{\Delta_{A}}{2}$.
  By the triangle inequality of metrics,
  \[
    d_{max}(\vec{x},\vec{y}) \leq d_{max}(\vec{p},\vec{x}) + d_{max}(\vec{p},\vec{y}) \leq \frac{\Delta_{A}}{2} + \frac{\Delta_{A}}{2} = \Delta_{A}
  \]
  By the contrapositive of the Adjacent or Far Lemma (\Autoref{:adjacent-or-far}), since there exists $\vec{x}\in X$ and $\vec{y}\in Y$ such that $d_{max}(\vec{x},\vec{y})\leq\Delta_{A}$, it must be that $X$ and $Y$ are adjacent.
\end{proof}

In order to state the Hypercube Parition Theorem (\Autoref{:hypercube-partition-thm}) as promised, we first need a simple lemma. To motivate the following choice, consider again the definition of reclusive distance (\Autoref{:reclusive-dist-def}) and the discussion leading up to it.
To make the most of the Partition Theorem, we want to have a large reclusive distance, and that is accomplished by keeping three key quantities in a reclusive matrix large---for each row $k$ and arbitrary entries $j<j'$ within that row, we want the following to be large: $(1-a_{kj})$, $(a_{kj})$, and $(a_{kj}-a_{kj'})$.
The first discourages using matrix entries greater than $1$ (which is partially why reclusive matrices were defined to not allow that) and encourages using small entries; the second encourages using large entries; the third seems to encourage ``even spacing'' of the terms in a given row.

Based on these considerations, we make all entries in the matrix multiples of a common denominator.

\begin{lemma}\label{:d-reclusive}
  Let $d\in\N$, and $A$ be a $d\times d$ reclusive matrix, and $\Delta_{A}$ its reclusive distance.
  If for all $i,j\in[d]$ it holds that $a_{ij}$ is a multiple of $\frac{1}{k}$ (including $0$), then $\Delta_{A}\geq\frac{1}{k}$.
  Further, such reclusive matrices exist for $k\geq d$.
\end{lemma}
\begin{proof}
  The first claim follows from the definitions of reclusive matrix and reclusive distance.
  For the second claim, let $k\geq d$ and consider the $d\times d$ reclusive matrix $A$ which has $1$'s on the main diagonal and $0$'s in the lower triangle (as required for all reclusive matrices) and for any other entry, $a_{ij}=\frac{d-i+1}{k}$.
 \[
  A=\left[
    \begin{array}{ccccccc}
      1      & \frac{d-1}{k} & \frac{d-2}{k} & \frac{d-3}{k} & \cdots & \frac{2}{k} & \frac{1}{k} \\
      0      & 1             & \frac{d-2}{k} & \frac{d-3}{k} & \cdots & \frac{2}{k} & \frac{1}{k} \\
      0      & 0             & 1             & \frac{d-3}{k} & \cdots & \frac{2}{k} & \frac{1}{k} \\
      0      & 0             & 0             & 1             & \cdots & \frac{2}{k} & \frac{1}{k} \\
      \vdots & \vdots        & \vdots        & \vdots        & \ddots & \vdots      & \vdots \\
      0      & 0             & 0             & 0             & \cdots & 1           & \frac{1}{k} \\
      0      & 0             & 0             & 0             & \cdots & 0           & 1
    \end{array}
  \right]
\]
It is easily verified that this is in fact a reclusive matrix.
\end{proof}

We shall be particularly interested in reclusive $d\times d$ matrices $A$ as in the lemma with $k=d$ in which case the lemma says that $\Delta_{A}\geq\frac{1}{d}$, and by examination of the first row $\Delta_{A}\leq\frac{1}{d}$ giving equality.

\RestatableHypercubePartitionThm
\begin{proof}
 As just discussed, there exists a reclusive partition $\P_{A}$ with with reclusive distance $\Delta_{A}=\frac{1}{d}$, so the conclusion follows from \Autoref{:partition-thm}.
\end{proof}

\section{Fundamental Property of the Reclusive Definition}
\label{sec:fundamental-reclusive-property}

It would be a fair question for one to ask, ``Why bother developing this general notion of reclusive partitions if you end up only using one specific example?'' The initial reason for why we defined reclusive partitions is that we found ourselves looking to generalize the example of exponentially decaying shifts that we presented at the beginning of the section in order to move from $\epsilon$ with the denominator growing exponentially to $\epsilon$ with the denominator growing linearly. The generalization followed by noting the properties in Subsection~\ref{subsec:motivating-reclusive-properties}.

Initially, we expected this to be a much stronger generalization than we needed and expected that the order in which the shifts were applied would not really matter. However, as we will demonstrate in this subsection, we were incorrect, and the definition of a reclusive matrix really seems to capture something that seems fundamental to the properties we are interested in---particularly the property of considering partitions with minimal possible clique sizes.

Consider, for example, the following two reclusive matrices $A'$ and $A''$ (the shorthand notation denotes the value only of entries in the strict upper diagonal since the diagonal must be $1$'s and the subdiagonal must be $0$'s). 

\[
  A'=\left(\frac{d-j+1}{d}\right)=
  \left[
    \begin{array}{ccccccc}
      1      & \frac{d-1}{d} & \frac{d-2}{d} & \frac{d-3}{d} & \cdots & \frac{2}{d} & \frac{1}{d} \\
      0      & 1             & \frac{d-2}{d} & \frac{d-3}{d} & \cdots & \frac{2}{d} & \frac{1}{d} \\
      0      & 0             & 1             & \frac{d-3}{d} & \cdots & \frac{2}{d} & \frac{1}{d} \\
      0      & 0             & 0             & 1             & \cdots & \frac{2}{d} & \frac{1}{d} \\
      \vdots & \vdots        & \vdots        & \vdots        & \ddots & \vdots      & \vdots      \\
      0      & 0             & 0             & 0             & \cdots & 1           & \frac{1}{d} \\
      0      & 0             & 0             & 0             & \cdots & 0           & 1
    \end{array}
  \right]
\qquad\text{Example, $d=5$:}
  \left[
    \begin{array}{ccccccc}
      1      & \frac{4}{5}   & \frac{3}{5}   & \frac{2}{5}   & \frac{1}{5} \\
      0      & 1             & \frac{3}{5}   & \frac{2}{5}   & \frac{1}{5} \\
      0      & 0             & 1             & \frac{2}{5}   & \frac{1}{5} \\
      0      & 0             & 0             & 1             & \frac{1}{5} \\
      0      & 0             & 0             & 0             & 1
    \end{array}
  \right]
\]
\[
  A''=\left(\frac{d-j+1}{d-i+1}\right)=
  \left[
    \begin{array}{ccccccc}
      1      & \frac{d-1}{d} & \frac{d-2}{d}   & \frac{d-3}{d}   & \cdots & \frac{2}{d}   & \frac{1}{d}   \\
      0      & 1             & \frac{d-2}{d-1} & \frac{d-3}{d-1} & \cdots & \frac{2}{d-1} & \frac{1}{d-1} \\
      0      & 0             & 1               & \frac{d-3}{d-2} & \cdots & \frac{2}{d-2} & \frac{1}{d-2} \\
      0      & 0             & 0               & 1               & \cdots & \frac{2}{d-3} & \frac{1}{d-3} \\
      \vdots & \vdots        & \vdots          & \vdots          & \ddots & \vdots        & \vdots        \\
      0      & 0             & 0               & 0               & \cdots & 1             & \frac{1}{2}   \\
      0      & 0             & 0               & 0               & \cdots & 0             & 1
    \end{array}
  \right]
\qquad\text{Example, $d=5$:}
  \left[
    \begin{array}{ccccccc}
      1      & \frac{4}{5}   & \frac{3}{5}   & \frac{2}{5}   & \frac{1}{5} \\
      0      & 1             & \frac{3}{4}   & \frac{2}{4}   & \frac{1}{4} \\
      0      & 0             & 1             & \frac{2}{3}   & \frac{1}{3} \\
      0      & 0             & 0             & 1             & \frac{1}{2} \\
      0      & 0             & 0             & 0             & 1
    \end{array}
  \right]
\]

Then consider the matrices $B'$ and $B''$. 

\[
  B'=\left(\frac{j-1}{d}\right)=
  \left[
    \begin{array}{ccccccc}
      1      & \frac{1}{d}   & \frac{2}{d}   & \frac{3}{d}   & \cdots & \frac{d-2}{d} & \frac{d-1}{d} \\
      0      & 1             & \frac{2}{d}   & \frac{3}{d}   & \cdots & \frac{d-2}{d} & \frac{d-1}{d} \\
      0      & 0             & 1             & \frac{3}{d}   & \cdots & \frac{d-2}{d} & \frac{d-1}{d} \\
      0      & 0             & 0             & 1             & \cdots & \frac{d-2}{d} & \frac{d-1}{d} \\
      \vdots & \vdots        & \vdots        & \vdots        & \ddots & \vdots        & \vdots        \\
      0      & 0             & 0             & 0             & \cdots & 1             & \frac{d-1}{d} \\
      0      & 0             & 0             & 0             & \cdots & 0             & 1
    \end{array}
  \right]
\qquad\text{Example, $d=5$:}
  \left[
    \begin{array}{ccccccc}
      1      & \frac{1}{5}   & \frac{2}{5}   & \frac{3}{5}   & \frac{4}{5} \\
      0      & 1             & \frac{2}{5}   & \frac{3}{5}   & \frac{4}{5} \\
      0      & 0             & 1             & \frac{3}{5}   & \frac{4}{5} \\
      0      & 0             & 0             & 1             & \frac{4}{5} \\
      0      & 0             & 0             & 0             & 1
    \end{array}
  \right]
\]

\[
  B''=\left(\frac{i-j}{d-i+1}\right)=
  \left[
    \begin{array}{ccccccc}
      1      & \frac{1}{d}   & \frac{2}{d}   & \frac{3}{d}   & \cdots & \frac{d-2}{d}   & \frac{d-1}{d}  \\
      0      & 1             & \frac{1}{d-1} & \frac{2}{d-1} & \cdots & \frac{d-3}{d-1} & \frac{d-2}{d-1}\\
      0      & 0             & 1             & \frac{1}{d-2} & \cdots & \frac{d-4}{d-2} & \frac{d-3}{d-2}\\
      0      & 0             & 0             & 1             & \cdots & \frac{d-5}{d-3} & \frac{d-4}{d-3}\\
      \vdots & \vdots        & \vdots        & \vdots        & \ddots & \vdots          & \vdots         \\
      0      & 0             & 0             & 0             & \cdots & 1               & \frac{1}{2}    \\
      0      & 0             & 0             & 0             & \cdots & 0               & 1
    \end{array}
  \right]
\qquad\text{Example, $d=5$:}
  \left[
    \begin{array}{ccccccc}
      1      & \frac{1}{5}   & \frac{2}{5}   & \frac{3}{5}   & \frac{4}{5} \\
      0      & 1             & \frac{1}{4}   & \frac{2}{4}   & \frac{3}{4} \\
      0      & 0             & 1             & \frac{1}{3}   & \frac{2}{3} \\
      0      & 0             & 0             & 1             & \frac{1}{2} \\
      0      & 0             & 0             & 0             & 1
    \end{array}
  \right]
\]

Clearly $B'$ and $B''$ are not reclusive matrices\footnote{Well, they are for $d=1$ and $d=2$, but not for any $d\geq3$.} since the entries in the rows are not decreasing. Nonetheless, $B'$ and $B''$ define partitions of $\R^{d}$ in the same fashion as a reclusive matrix\footnote{The partition is the set of all unit hypercubes with representatives in the set $\set{B'\vec{v}\colon\vec{v}\in\Z^{d}}$ (resp. $\set{B''\vec{v}\colon\vec{v}\in\Z^{d}}$).}. For example, the partition associated with $B'$ in dimension $5$ is the one constructed by partitioning $\R^{1}$ into unit intervals, then extruding those intervals into unit squares in $\R^{2}$, copying the extrusion to multiple layers, and shifting each layer by $1/5$ of a unit more to the right than the previous layer; then extruding this partition of $\R^{2}$ into $\R^{3}$ and shifting each layer by $2/5$ of a unit more than the previous layer. In this sense, this partition is constructed in a {\em very} similar way to the partition for $A'$: layers are still offset by multiples of $1/5$, and the only difference is that smaller shifts happen first.

By Theorem~\ref{:coloring}, the partitions $\P_{A'}$ and $\P_{A''}$ do not contain any cliques of size $d+2$, so a natural question is whether the partitions associated with $B'$ and $B''$ also have this property since their constructions are so similar. Our intuition was that the answer to this question would ``yes'', but this is not the case. For $d\leq4$, the partitions for both $B'$ and $B''$ do not have cliques of size $d+2$, but for $d=5$ they do. 

We found this by exhaustive computer search\footnote{Because of the repetitive structure of the partition (since the underlying structure is a lattice), it suffices to check a sufficiently large but finite subset of the partition for cliques to determine the size of the largest clique in the whole partition.}, but our claim that for $d=5$, $B'$ and $B''$ each have a $d+2=7$ clique can easily be verified. Consider the following set of $7$ vectors in $\Z^{5}$:

\newcommand{\colvec}[5]{\left[\begin{array}{c}#1\\#2\\#3\\#4\\#5\end{array}\right]}
\[
  \colvec{ 0}{ 1}{-1}{ 0}{ 1},
  \colvec{ 0}{ 0}{-1}{ 0}{ 1},
  \colvec{ 0}{ 1}{ 0}{ 0}{ 0},
  \colvec{-1}{ 1}{-1}{ 0}{ 1},
  \colvec{-1}{ 0}{ 0}{ 0}{ 1},
  \colvec{-1}{ 0}{ 0}{ 1}{ 0},
  \colvec{ 0}{ 0}{ 0}{ 1}{ 0}
\]

Multiplying $B'$ with each vector results in the representative corner of the hypercubes associated with each vector (one can similarly compute them for $B''$):

\[
  \colvec{3/5}{7/5}{-1/5}{4/5}{1},
  \colvec{2/5}{2/5}{-1/5}{4/5}{1},
  \colvec{1/5}{1}{0}{0}{0},
  \colvec{-2/5}{7/5}{-1/5}{4/5}{1},
  \colvec{-1/5}{4/5}{4/5}{4/5}{1},
  \colvec{-2/5}{3/5}{3/5}{1}{0},
  \colvec{3/5}{3/5}{3/5}{1}{0}
\]

To see that the hypercubes with representative corners at these $7$ locations form a clique, check that the $d_{max}$ distance between any pair is exactly $1$ (and apply Corollary~\ref{:equiv-defn-adjacency}). By Proposition~\ref{:hypercube-cliques-have-clique-points}, and the discussion preceeding it, this shows that this partition has points where any ball of any radius centered at that point intersect all $7$ of these hypercubes.

It could be that there is some sufficiently large dimension $d$ such that $B'$ and $B''$ have cliques of size at most $d+1$, but we conjecture that this is not the case.

\begin{conjecture}[Maximum Clique Sizes in Non-Reclusive Partitions $B'$ and $B''$]\label{:max-clique-sizes}
  Based on our computations, we conjecture that the maximum clique size in the partitions associated with $B'$ and $B''$ is greater than $d+1$ for all $d\geq5$, and we know the maximum clique sizes for dimensions given in Table~\ref{tab:max-clique-sizes}.

    \begin{table}[H]
        \centering
        \begin{tabular}{ccccccccccccc}
          d   & 1  & 2  & 3  & 4  & 5  & 6  & 7  & 8  & 9  & 10 & 11 & 12 \\\hline
          B'  & 2  & 3  & 4  & 5  & 7  & 9  & 12 & 16 & 22 & 30 & 39 & 51 \\
          B'' & 2  & 3  & 4  & 5  & 7  & 9  & 11 & 16 & 21 & 28 & 36 & 47 \\
        \end{tabular}
        \caption{Maximum clique sizes in non-reclusive partitions $B'$ and $B''$ for various dimensions.}
        \label{tab:max-clique-sizes}
    \end{table}
\end{conjecture}

The above example demonstrates that our definition of reclusive partitions is not an arbitrary one and captures a certain structure of lattice based partitions that is sufficient to ensure that no large cliques exist. The key is that because our definition of reclusive partitions demands that the terms be decreasing in each row, we get an equivalent definition of adjacency in terms of weak-alt-1 sequences. The above example shows that this equivalence does not hold if we relax the decreasing requirement.
\section{Optimality of the Degree Parameter ($k$)}
\label{sec:clique-optimality}

Given our results that we can find a uniformly-sized neighborhood around every point that intersects at most $d+1$ members of the partition, it is natural to ask if this is optimal---can this value be improved from $d+1$ to $d$ or even smaller? In other words, is there a partition of $\R^{d}$ and a neighborhood around each point in $\R^{d}$ such that each of these neighborhoods intersects at most $d$ members of the partition? The answer is basically no; our partitions are optimal, and in this section we will discuss the technicalities of what we mean by ``basically'' and then prove the optimality formally.

It is obvious that this question above is only worth asking if some restrictions are added, because otherwise the partition of $\R^{d}$ which contains just one member ($\P=\set{\R^{d}}$) along with any neighborhoods trivially has the property that each neighborhood intersects only one member of the partition (because there is only one). The context above suggests that we are interested in partitions similar to tilings and packings, and so we will restrict to the context in which the size of the partition members is uniformly bounded in some way.

Two natural ways to bound the size of the members are by bounding the diameter and by bounding the measure. Note that bounding the diameter is a strictly stronger condition (assuming all members are measurable), because if we insist that each member of the partition has diameter at most $D$ in the $d_{max}$ metric, then the measure of each member is at most $(2D)^{d}$ (fix some point in the member, and all other points are within $D$ of it, so the member is contained in a ball of radius $D$ in the $d_{max}$ metric, and this ball has measure $(2D)^{d}$ because this ball is a $d$-dimensional hypercube with side length $2D$).

Thus, if we can answer this question by just restricting the measure, that would be ideal since it is a weaker hypothesis. Unfortunately, we will see that we cannot get the strongest version of the conclusion that we want just by bounding the measure, and so we will also consider diameter bounds. It will occur frequently in our theorems, though, that to get nicely generalized statements, diameter is not quite what we want. For example, consider the sets $[0,1)$ and $[0,1]$. Both sets have diameter $1$, but the the former has the stronger property that all pairs of points have distance strictly less than the diameter. In other words, diameter\footnote{In general $\diam(X)\defeq\sup\set{\textrm{distance}(x,y):x,y\in X}$} is defined as a supremum of distances, and the former set does not attain the supremum while the latter does. We want to distinguish between these cases, so we give the following definition.

\begin{definition}[Strict Pairwise Bound]\label{:strict-pairwise-bound}
  If $X$ is a subset of a metric space and $D\in(0,\infty)$ is a constant such that for all $x,y\in X$ it holds that $\mathrm{distance}(x,y)<D$, then we call $D$ a strict pairwise bound of $X$.
\end{definition}

\begin{remark}
  Unbounded sets do not have any strict pairwise bounds, and bounded sets have infinitely many strict pairwise bounds (if $D$ is a strict pairwise bound, then so is $D+\epsilon$ for any $\epsilon\in[0,\infty)$). 
  Also, we have the following implications:
  \[
    \diam(X)<D\quad\Longrightarrow\quad D\text{ is a strict pairwise bound of }X\quad\Longrightarrow\quad\diam(X)\leq D
  \]
  In general, the reverse implications do not hold which is exemplified by the sets $[0,D)$ and $[0,D]$.
\end{remark}

Because of the distinction just mentioned between a
strict bound on diameter, a strict pairwise bound, and a nonstrict bound on
diameter, we don't want to restrict to unit diameter partition members in this
section as we did in the last section, and instead will state many results for
a general diameter (or strict pairwise bound) $D\in(0,\infty)$. It is worth
noting that 
all of the results of the previous section used unit hypercubes, but if we
allow for $D$ diameter hypercubes (i.e., translations of $[0,D)^{d}$) then the
results also scale so that we can still attain the parameter $k=d+1$ and
$\epsilon=\frac{D}{2d}$.

In this section, we will prove the three different Optimality Theorems presented in \Autoref{sec:results}. 
 We begin with a more in depth discussion (and restatement) of these theorems and some of the corollaries that will follow. We will then introduce the necessary tools (primarily a variant of Sperner's lemma) for proving the Optimality Theorems, and prove them. Lastly, we will prove that there are ``gaps'' between each of these theorems which justifies the need for multiple variants.

\subsection{The Optimality Theorems}
\label{subsec:optimality-theorems}

In the statements below, let $m$ denote the Lebesgue measure of a set. For those unfamiliar with measure theory, sets that are ``not too bizarre'' are called Lebesgue measurable, and for any Lebesgue measurable set $X\subseteq\R^{d}$, $m(X)$ is a generalization of the volume of $X$. Sets that are not Lebesgue measurable (ones that are ``too bizarre'') don't have a well-defined volume/measure.

We will state each of the three theorems and then give the interpretation of each.

\RestatableFirstOptimalityThm

The conclusion of the First Optimality Theorem is that for partitions satisfying the hypothesis, there is no way to put uniformly sized neighborhoods (each a ball of radius $\epsilon$) centered at every point of $\R^{d}$ each intersecting only $d$ members of the partition. For any chosen $\epsilon$, there will be a point in the space where its neighborhood intersects at least $d+1$ partition members.

\RestatableSecondOptimalityThm

The $\epsilon$ function above should be viewed as some fixed $\epsilon$ for each point of $\R^{d}$. The conclusion of the Second Optimality Theorem is that for partitions satisfying the hypothesis, there is no way to put neighborhoods (each a ball of some radius) centered at every point of $\R^{d}$ each intersecting only $d$ members of the partition even if we allow the radius of the neighborhood to depend on which point it is centered at. There will be a point in the space where its neighborhood intersects at least $d+1$ partition members. This is a stronger conclusion than the First Optimality Theorem, but the hypothesis is also stronger (assuming all members are measurable) as discussed at the beginning of this section.

\RestatableThirdOptimalityThm

The conclusion of the Third Optimality Theorem is that for partitions satisfying the hypothesis, we don't even need to consider neighborhoods, because there is a single point in the closure of at least $d+1$ members of the partition, and thus for this specific point, any sized neighborhood with intersect $d+1$ members of the partition. This further implies that the partition has a $(d+1)$-clique (the intersection of closures of these $d+1$ members is non-empty because it contains the point $\vec{p}$). Again, this conclusion is stronger than the conclusions of the Second Optimality Theorem, but the hypothesis is also stronger.

\begin{remark}
Observe that in the First (resp. Second) Optimality Theorem, the strict inequality on the measure (resp. diameter) could be replaced with a non-strict inequality, and the statement would be equivalent since the actual value of $M$ (resp. $D$) is not used in the conclusion. However, this is not the case with the Third Optimality Theorem. Consider for example if $\P=\set{[0,1]^{d}}\cup\set{\set{\vec{x}}:\vec{x}\not\in[0,1]^{d}}$ (i.e. the partition where one member is the hypercube $[0,1]^{d}$, and every other member is a singleton set). Then, it holds for $D=1$ that for all $X\in\P$, $\diam(X)\leq1$ and taking $\vec{\alpha}=\vec{0}$, the set $\vec{\alpha}+[0,1]^{d}$ only intersects one member of the partition (and in fact intersects the closure of only one member of the partition since all members are already closed sets). However, because all members are closed, for any $\vec{p}\in\R^{d}$, the set $\mathcal{N}_{\closure{0}}(\vec{p})=\set{X\in\P:\closure{X}\ni\vec{p}}$ has cardinality $1<d+1$ because it contains only the set $\member(\vec{p})$. Thus, the use of the strict inequality of the statement will be necessary.
\end{remark}


The hypothesis added in the Third Optimality Theorem is that there is a local finiteness somewhere in the partition. Essentially this limits the resolution so that we can examine some portion of the partition and not deal with infinitely small sets. This finiteness condition arises in the natural context below.

\begin{corollary}\label{:diam-measure-cor}
  Let $d\in\N$ and $\P$ a partition of $\R^{d}$. If there exists $D\in(0,\infty)$ and $\mu\in(0,\infty)$ such that for all $X\in\P$, it holds that $X$ is Lebesgue measurable, and $\mu<m(X)$, and $\diam(X)<D$, then there exists $\vec{p}\in\R^{d}$ such that
  \[
    \abs{\mathcal{N}_{\epsilon}(\vec{p})}\geq d+1.
  \]
   Furthermore, $\P$ contains a $(d+1)$-clique.
\end{corollary}
\begin{proof}
  It suffices to prove that the hypotheses of the corollary imply the hypothesis of the Third Optimality Theorem.

  Let $\vec{\alpha}\in\R^{d}$ be arbitrary. Let $B=\vec{\alpha}+[0,D]^{d}$, and $\mathcal{B}=\set{X\in\P:X\cap B\not=\emptyset}$ (the members intersecting $B$). Let $B'=\vec{\alpha}+[0-D,D+D]^{d}$ (which should be viewed as ball containing $B$ and all points within distance $D$ of $B$ in the $d_{max}$ metric). Thus, for any $X\in\mathcal{B}$, because $\diam(X)<D$, we have $X\subseteq B'$, and thus $\bigsqcup_{X\in\mathcal{B}}X\subseteq B'$ (the square cup indicating a disjoint union).

   Consider the following volume argument (see the Appendix for details on the second inequality):
  \[
    (3D)^{d} = m(B') \geq m(\bigsqcup_{X\in\mathcal{B}}X) \geq \sum_{X\in\mathcal{B}}m(X)\geq \abs{\mathcal{B}}\mu.
  \]
  Dividing both sides by $\mu$ gives $\abs{\mathcal{B}}\leq\frac{(3D)^{d}}{\mu}$ which is finite and thus demonstrates that the hypotheses of the Third Optimality Theorem hold and completes the proof.
\end{proof}

The following is also a simple corollary, but we state it as a theorem to emphasize that in the case of unit hypercube partitions, we get the strongest of the results above.

\RestatableOptimalityForUnitHypercubePartitionsThm
\begin{proof}
  All members of $\P$ are measurable and have measure and diameter $1$, so apply \Autoref{:diam-measure-cor}.
\end{proof}

As we have mentioned before, the reason we believe the bounds on measure and diameter of the partition members to be very reasonable is that we are interested in partitions which are something like tilings. In a tiling, there is some finite collection of members used with translation and rotation (and other orthogonal maps). Because all norms on $\R^{d}$ are equivalent. If each member of a finite collection has finite diameter in some norm, then they also have finite diameter in the $\norm{\cdot}_{\infty}$ norm, and because the collection of members is finite, there is some value $D\in(0,\infty)$ so that all of these members have diameter less than $D$ in the $\norm{\cdot}_{\infty}$ norm. Also in a tiling, the members usually have non-zero measure, so again by finiteness there is some $\mu\in(0,\infty)$ which is a lower bound on the measure of each member (and measure is invariant under rotation and translation and other orthogonal maps). So by this reasoning we have the following simple corollary which is a slight generalization of the optimality theorem above for unit hypercubes.


\begin{corollary}[Tiling Corollary]
  Let $d\in\N$, and $\mathcal{F}$ be a finite collection of Lebesgue measureable subsets of $\R^{d}$ each with finite diameter and finite measure. Let $\P$ be a partition of $\R^{d}$ consisting only of translated orthogonal transformations of members of $\mathcal{F}$. Then there exists $\vec{p}\in\R^{d}$ such that
  \[
    \abs{\mathcal{N}_{\closure{0}}(\vec{p})}\geq d+1
  \]
  and $\P$ also contains a $(d+1)$-clique.
\end{corollary}
\begin{proof}
  The proof follows from the discussion above and \Autoref{:diam-measure-cor}.
\end{proof}

We will also state one more optimality theorem in light of the structure of \Autoref{:adjacent-or-far} (Adjacent or Far Lemma (Points)).

\begin{definition}
  If $d\in\N$ and $\P$ is a partition of $\R^{d}$, then we say that $\P$ has the adjacent-or-far property if there exists $\epsilon\in(0,\infty)$ such that for all $X,Y\in\P$ it is either the case that $X\adj Y$ or it is the case that for any $\vec{x}\in X$ and $\vec{y}\in Y$ that $d_{max}(\vec{x},\vec{y})>\epsilon$. Any such $\epsilon$ is called an adjacent-or-far constant for $\P$.
\end{definition}

The following result strengthens the both the hypothesis and conclusion of the Second Optimality Theorem but the conclusion is slightly weaker than that of the Third Optimality Theorem.

\begin{proposition}[Adjacent or Far Optimality Theorem]
  If $d\in\N$, and $\P$ is partition of $\R^{d}$ with the adjacent-or-far property, and there exists $D\in(0,\infty)$ such that for all $X\in\P$, $\diam(X)<D$, then $\P$ contains a $(d+1)$-clique
\end{proposition}

\begin{proof}[Proof Sketch]
  Use the same proof technique as the Second Optimality Theorem to restrict attention to $[0,D]^{d}$, induce a partition on it from $\P$ (which is non-spanning), and consider an admitted Sperner/KKM coloring to find a point $\vec{p}$ at the closure of $d+1$ colors. Taking $\epsilon$ to be an adjacent-or-far constant for $\P$, then $\closure{B}_{\epsilon/2}(\vec{p})$ intersects at least $d+1$ colors and thus at least $d+1$ members of $\P$. By the triangle inequality with the point $\vec{p}$, for any pair of these members $X,Y$, there is a point $\vec{x}\in X$ and $\vec{y}\in Y$ such that $d_{max}(\vec{x},\vec{y})\leq\epsilon$, so by the definition of the adjacent or far property, $X\adj Y$. Thus these $d+1$ members constitute a $(d+1)$-clique.
\end{proof}

\subsection{Tools for Proving the Optimality Theorems}

In the broader literature, when something like a Sperner coloring of a simplex is made continuous on the simplex (i.e. instead of coloring the vertices of some simplicial subdivision of the simplex, the entire simplex is colored) it is essentially a KKM covering. The Knaster-Kuratowski-Mazurkiewicz (KKM) lemma states that if a $(d+1)$-dimensional simplex is covered by a family $\set{C_{1},\ldots,C_{d+1}}$ of $d+1$ ``reasonable'' closed sets, then there is a point belonging to the intersection of all these closed sets $\bigcap_{i\in[d]}C_{i}$. The ``reasonableness'' is stated in terms analogous to the conditions on a Sperner coloring of a $(d+1)$-simplex. It is well known that the KKM lemma follows easily from Sperner's lemma using a sequence of subdivisions of finer and finer resolution and using the Bolzano–Weierstrass theorem (c.f. \cite{Kuhn66}), and it is known that the KKM lemma can also be used to prove Sperner's lemma (c.f. \cite{komiya_simple_1994}).
Thus, if we assign to every point within the simplex one of $d+1$ colors in a ``reasonable'' way and take $C_{i}$ to be the {\em closure} of the points of the $i$th color, then the KKM lemma guarantees the existence of some point in the intersection of all $C_{i}$ (which is a point belonging to the closure of all of the colors).

A number of variations on Sperner's lemma are known, and as with the standard case, they can be used to easily prove variations of the KMM lemma. Of particular interest to us will be the hypercube versions of Sperner's lemma or the KKM lemma in which each vertex/corner of the cube is given a unique color. The hypercube version of Sperner's lemma was first proved by Kuhn \cite{Kuhn66}, and generalized even further to convex polytopes by De Loera, Peterson, and Su \cite{loera_polytopal_2001}. Kuhn did not explicitly show the KKM/continuous variant for the cube, but it is easily proved using essentially the same arguments as in the case of the simplex, and De Loera \textit{et. al.} did explicitly mention connections to the KKM lemma. There were also convex polotope variants of the KKM lemma proven directly by van der Laan \textit{et. al.} \cite{van_der_laan_intersection_1999} without proving them via Sperner's lemma.

The next proposition follows almost immediately from the results alluded to in the prior paragraph (in particular, \cite[Corollary~3]{loera_polytopal_2001}, \cite[\textsection~2-3]{kuhn_hw_combinatorial_nodate}, or \cite{van_der_laan_intersection_1999}.

\begin{definition}[Sperner/KMM Coloring]\label{:sperner-kmm-defn}
  Let $d\in\N$ and $V=\set{0,1}^{d}$ denote a set of colors (which is exactly the set of vertices of $[0,1]^{d}$ so that colors and vertices are identified). Let $\chi:[0,1]^{d}\to V$ be a coloring function such that for any face $F$ of $[0,1]^{d}$, for any $\vec{x}\in F$, it holds that $\chi(\vec{x})\in F$ (informally, the color of $\vec{x}$ is one of the vertices in the face $F$). Such a function $\chi$ will be called a Sperner/KMM coloring.
\end{definition}

\begin{proposition}[Cubical Sperner/KMM lemma]\label{:sperner-kmm}
    Let $d\in\N$ and $V=\set{0,1}^{d}$ and $\chi:[0,1]^{d}\to V$ be a Sperner/KMM coloring. Then there exists a subset $J\subset V$ with $\abs{J}=d+1$ and a point $y\in[0,1]^{d}$ such that for all $j\in J$, $y\in\closure{\chi^{-1}(j)}$ (informally, $\vec{y}$ is in the closure of at least $d+1$ different colors).
\end{proposition}

This claim really does follow almost immediately from \cite[Corollary~3]{loera_polytopal_2001}, but we framed our proposition with notation and wording that is sufficiently different that we will offer a proof. We emphasize, though, that this proof amounts to nothing more than carefully working through this change of notation.

\begin{proof}
  For each vertex/color $\vec{v}\in V$, let $C^{\vec{v}}=\closure{\chi^{-1}(\vec{v})}$ denote the closure of the points assigned the color $\vec{v}$. Let $F$ be some face of $[0,1]^{d}$ and we will show that $F$ is covered by $\bigcup{C^{\vec{v}}:\vec{v}\in F\cap V}$ (that is, $F$ is covered by the collection $C^{\vec{v}}$ for each vertex $\vec{v}\in F$). To see this, let $\vec{x}\in F$ be arbitrary. Since the codomain of $\chi$ is $V$, we have trivially that $\chi(\vec{x})\in V$. Let $\vec{v}=\chi(\vec{x})$ denote this value to emphasize that this is a vertex of $[0,1]^{d}$. Also, $\vec{v}=\chi(\vec{x})\in F$ by hypothesis, and thus we have a particular $\vec{v}\in F\cap V$ such that
  \[
    C^{\vec{v}}=\chi^{-1}(\vec{v})=\chi^{-1}(\chi(\vec{x}))\ni\vec{x}
  \]
  which shows $F$ is covered as claimed.

  This gives the hypothesis of \cite[Corollary~3]{loera_polytopal_2001}, so taking $\vec{p}\in[0,1]^{d}$ arbitrarily, there must exist $J\subset V$ with $\abs{J}=d+1$ such that $\bigcap_{j\in J}C^{j}\not=\emptyset$ which proves our claim.
\end{proof}

For the next few result, we will use standard projection maps $\pi_{i}:[0,1]^{d}\to[0,1]$ defined by $\pi_{i}(\vec{x})\defeq x_{i}$ which maps points to the $i$th coordinate value, and we will also apply this map to sets and mean that $\pi_{i}(X)=\set{\pi_{i}(\vec{x}):\vec{x}\in X}$.

\begin{definition}[Non-Spanning]
  Let $d\in\R^{+}$ and $\mathcal{S}$ be a partition of $[0,1]^{d}$. We say that $\mathcal{S}$ is a non-spanning partition if it holds for all $X\in\mathcal{S}$ and for all $i\in[d]$ that either $\pi_{i}(X)\not\ni0$ or $\pi_{i}(X)\not\ni1$ (or both).
\end{definition}

Informally, a non-spanning partition of a hypercube does not contain any members which span the hypercube so as to intersect an opposite pair of opposite facets ($(d-1)$-dimensional faces). We will now show that such partitions admit a Sperner/KMM coloring which respects the structure of the partition (i.e. any two points in the same partition member are assigned the same color).

\begin{lemma}[Coloring Admission]\label{:color-admission}
  Let $d\in\N$, and $V=\set{0,1}^{d}$, and $\mathcal{S}$ a non-spanning partition of $[0,1]^{d}$. Then there exists a Sperner/KMM coloring $\chi:[0,1]^{d}\to V$ and a function $\chi_{\mathcal{S}}:\mathcal{S}\to V$ (which we will call a partition coloring) such that for all $\vec{x}\in\R^{d}$, $\chi(\vec{x})=\chi_{\mathcal{S}}(\member(\vec{x}))$.
\end{lemma}
\begin{proof}
  We begin by defining a coloring function $\chi_{\mathcal{S}}$ which colors members of the partition, and then define the coloring function for points which just uses the color of the containing member. Specifically, define $\chi_{\mathcal{S}}:\mathcal{S}\to V$ by
  \[
    \chi_{\mathcal{S}}(X)\defeq\langle c_{i}\rangle_{i=1}^{d}\quad\text{where }c_{i}=
    \begin{cases}
      0 & 0\in\pi_{i}(X)\\
      1 & 1\in\pi_{i}(X)\\
      0 & \text{otherwise}
    \end{cases}
  \]

  Observe that the cases are well-defined because $\mathcal{S}$ is non-spanning. Then define $\chi:[0,1]^{d}\to V$ by $\chi(\vec{x})\defeq\chi_{\mathcal{S}}(\member(\vec{x}))$. We will now show that this has the required properties of a Sperner/KMM coloring.

  Let $F$ be any face of $[0,1]^{d}$. Thus, $F$ can be expressed as $F=\prod_{i=1}^{d}F_{i}$ where each $F_{i}$ is one of three sets: $\set{0}$, $\set{1}$, or $[0,1]$. Let $\vec{x}\in F$ be arbitrary and let $\vec{c}=\chi(\vec{x})$; we must show that $\vec{c}\in F$ which we do by showing that $c_{i}\in F_{i}$ for each $i\in[d]$. There are three cases. If $F_{i}=\set{0}$, then $x_{i}=0$ (since $\vec{x}\in F$), so $0=x_{i}=\pi_{i}(\vec{x})\in\pi_{i}(\member(\vec{x}))$, in which case $c_{i}=0$ by definition. The case $F_{i}=\set{1}$ is analogous. The remaining case is that $F_{i}=[0,1]$, but in this case $c_{i}\in F_{i}$ trivially since $c_{i}$ is either $0$ or $1$.
\end{proof}

\begin{remark}
  Observe that in the previous few results, there is nothing special about the hypercube $[0,1]^{d}$; it was chosen only for convenience. It is trivial to extend the result to any rectangle $A=\prod A_{i}$ (where $A_{i}=[l_{i}, r_{i}]$ for some values $l_{i}$ and $r_{i}$). This can either be done by rewriting the proofs (since these sets are convex polytopes) or by using the natural scaling map (which is continuous and so preserves the intersections of closures). We will freely use the results above for other hypercubes.
\end{remark}

\subsection{Proofs of the Optimality Theorems}

We will now prove all three optimality theorems. All three proofs will have the same flavor, though it gets somewhat lost in notation. Essentially, we consider some sufficiently sized hypercube in the space, use the partition $\P$ to induce a partition $\mathcal{S}$ on the hypercube, argue that the induced partition is non-spanning (thus admitting a Sperner/KMM coloring), use the coloring to find a point $\vec{p}$ in the closure of $d+1$ colors, and then argue that $\vec{p}$ has the desired property when we transition from color classes back to partition members. For convenience, we restate each theorem before its proof.

The proof of the Third Optimality Theorem requires the fewest technical details, so we begin with its proof.

\RestatableThirdOptimalityThm* 
\begin{proof}
  Let $D$, and $\vec{\alpha}$ as in the theorem statement. Without loss of generality, we may assume that $\vec{\alpha}=\vec{0}$, that $D=1$, and that for all $X\in\P$, $\diam(X)<1$ by applying a continuous transformation\footnote{In further detail, let $H=\vec{\alpha}+[0,D^{+}]^{d}$ and $\phi:\R^{d}\to\R^{d}$ the continuous bijection defined by $\phi(\vec{x})=\frac{\vec{x}-\vec{\alpha}}{D}$. The map $\phi$ induces a partition $\P'$ of $\R^{d}$ (specifically $\P'=\set{\phi(X):X\in\P}$). Importantly, the claim that there exists $\vec{p}\in\R^{d}$ such that $\abs{\set{X\in\P:\closure{X}\ni\vec{p}}}\geq d+1$ is equivalent to the claim that there exists $\vec{p'}\in\R^{d}$ such that $\abs{\set{Y\in\P':\closure{Y}\ni\vec{p'}}}\geq d+1$; this is because $\phi$ is continuous, so for any $\vec{x}\in\R^{d}$ and any $X\in\P$, we have $\vec{x}\in\closure{X}$ iff $\phi(\vec{x})\in\closure{\phi(X)}$. Thus we can work in $\P'$ instead. It is easily verified that $\phi(\vec{\alpha})=\vec{0}$, and $\phi(H)=[0,1]^{d}$, and for any $X\in\P$ that $\diam(\phi(X))\leq\frac{D}{D^{+}}<1$}. Observe that $\P$ induces a finite partition $\mathcal{S}$ on $[0,1]^{d}$ (specifically, $\mathcal{S}=\set{X\cap[0,1]^{d}:X\in\P\text{ and }X\cap[0,1]^{d}}\not=\emptyset$). Further, $\mathcal{S}$ is a non-spanning partition; this is because some $Y\in\mathcal{S}$ had the property for some $i\in[d]$ that $\pi_{i}(Y)\ni0$ and $\pi_{i}(Y)\ni1$, then let $\vec{y},\vec{z}\in Y$ such that $y_{i}=\pi_{i}(\vec{y}=0)$ and $z_{i}=\pi_{i}(\vec{z}=1)$ so $\diam(Y)\geq d_{max}(\vec{z},\vec{y})\geq\abs{z_{i}-y_{i}}=1$ which would be a contradiction. Thus, by \Autoref{:color-admission}, letting $V=\set{0,1}^{d}$, there is a Sperner/KMM coloring $\chi:[0,1]^{d}\to V$ and partition coloring $\chi_{\mathcal{S}}:\mathcal{S}\to[0,1]^{d}$.

  For each $j\in V$, let $\mathcal{S}_{j}=\set{Y\in\mathcal{S}:\chi_{\mathcal{S}}(Y)=j}$ (i.e. the members of the partition mapped to color $j$) noting that $\chi^{-1}(j)=\bigcup_{Y\in\mathcal{S}_{j}}Y$. By \Autoref{:sperner-kmm} there is some $J\subseteq V$ with $\abs{J}\geq d+1$ and some $\vec{p}\in[0,1]^{d}$ with $\vec{p}\in\bigcap_{j\in J}\closure{\chi^{-1}(j)}$. This gives the following:
  \begin{align*}
    \vec{p} & \in\bigcap_{j\in J}\closure{\chi^{-1}(j)}\\
            & \in\bigcap_{j\in J}\closure{\bigcup_{Y\in\mathcal{S}_{j}}Y}\\
            & \in\bigcap_{j\in J}\bigcup_{Y\in\mathcal{S}_{j}}\closure{Y}\tag{Because $\mathcal{S}$ is a finite set}\\
  \end{align*}
  Thus, for each $j\in J$, $\vec{p}$ belongs $\closure{Y}$ for some $Y\in\mathcal{S}_{j}$; let $Y_{j}$ denote one such element. Trivially, if $j\not=j'$, then $\mathcal{S}_{j}$ and $\mathcal{S}_{j'}$ contain no common members, so $Y_{j}\not=Y_{j'}$ and thus

  \[
    \abs{\set{X\in\P:\closure{X}\ni\vec{p}}} \geq \abs{\set{Y\in\mathcal{S}:\closure{Y}\ni\vec{p}}} \geq \abs{\set{Y_{j}:j\in J}} \geq d+1
  \]

\end{proof}

The finiteness condition above was critical. In essence, we found $\vec{p}$ to belong to the closures of $d+1$ color classes, and since each color class consisted of only finitely many members, in fact $\vec{p}$ had to belong to the closure of not just the union of the members, but the closure of a single member (or multiple members).

\RestatableSecondOptimalityThm* 
\begin{proof}
  Let $D\in(0,\infty)$, $\epsilon:\R^{d}\to(0,\infty)$ as in the statement. As in the previous proof\footnote{We may use any $\vec{\alpha}$.}, we may assume without loss of generality that $D=1$ so that all members of $\P$ have diameter less than $1$, and we don't have to worry about scaling $\epsilon$ since it is arbitrary to begin with. As in the previous proof, $\P$ induces a non-spanning partition $\mathcal{S}$ on $[0,1]^{d}$ which admits Sperner/KMM coloring $\chi$, so letting $V=\set{0,1}^{d}$, there is some $J\subset V$ with $\abs{J}\geq d+1$ and some $\vec{p}\in\bigcap_{j\in J}\closure{\chi^{-1}(j)}$. Thus any open set containing $\vec{p}$ will intersect members of at least $d+1$ colors and so intersects at least $d+1$ members. Thus, in particular, $B^{\circ}_{\epsilon(\vec{p})}(\vec{p})$ intersects at least $d+1$ members of $\mathcal{S}$ and thus at least $d+1$ members of $\P$.
\end{proof}


The proof of the First Optimality Theorem is the most complicated, and we recommend looking at the \hyperlink{first-optimality-thm-proof-outline}{proof outline for the First Optimality Theorem} given earlier to understand the overall structure of the proof.

\RestatableFirstOptimalityThm* 
\begin{proof}
    The beginning of the proof will sound a bit strangely worded. This is because we want some bound throughout the proof of how many members of the partition an $\epsilon$-ball can intersect with; the natural approach is thus proof by contradiction, but we can get a stronger result from the proof than stated in the theorem if we are careful not to arrive at a contradiction in the end.

    Let $d\in\N$, and let $\P$ be a partition of $\R^{d}$ with Lebesgue measurable sets, and let $M\in(0,\infty)$ such that for all $X\in\P$, $m(X)<M$. Let $k\geq d+1$ be fixed throughout the proof. If there exists $\vec{p}\in\R^{d}$ such that $\abs{\mathcal{N}_{\epsilon}(\vec{p})}>k$ then we are done. Otherwise, we may assume that for all $\vec{p}\in\R^{d}$ that $\abs{\mathcal{N}_{\epsilon}(\vec{p})}\leq k$. (Note that since $k\geq d+1$, there {\em are} partitions such that for all $\vec{p}\in\R^{d}$ that $\abs{\mathcal{N}_{\epsilon}(\vec{p})}\leq k$---the reclusive partitions for example---and thus we have {\em not} stated a contradiction assumption.)

  Let $V=\epsilon\cdot\Z^{d}$ (representing vertices characterizing a scaled grid). For any $\vec{\alpha}\in V$, let $H(\vec{\alpha})=\vec{\alpha}+[-\frac{\epsilon}{2},\frac{\epsilon}{2})^{d}$ (standing for half open hypercube). Note that $\closure{H(\vec{\alpha})}=\closure{B}_{\epsilon/2}(\vec{\alpha})$. Now observe that $\set{H(\vec{\alpha}):\vec{\alpha}\in V}$ is a partition of $\R^{d}$, and in this way we view the vertex set $V$ as the centers of cells partitioning the space, and informally, we will identify the two for discussion.

  Now we want to argue that each cell (each $H(\vec{\alpha})$) intersects some member of the partition with high volume/measure. For any $\vec{\alpha}\in V$ we have $H(\vec{\alpha})\subseteq\closure{H(\vec{\alpha})}=\closure{B}_{\epsilon/2}(\vec{\alpha})\subseteq \closure{B}_{\epsilon}(\vec{\alpha})$ and thus $\set{X\in\P:X\cap H(\vec\alpha)\not=\emptyset} \subseteq \set{X\in\P:X\cap \closure{B}_{\epsilon}(\vec{\alpha})\not=\emptyset}$ where the latter has cardinality at most $k$ by assumption implying that the former has cardinality at most $k$. Then we have
  \[
    \epsilon^{d} = m(H(\vec\alpha)) = \sum_{X\in\P}m(X\cap H(\vec{\alpha})) = \sum_{X\in\P:X\cap H(\vec\alpha)\not=\emptyset}m(X\cap H(\vec{\alpha})).
  \]
  Since the index set has size at most $k$, by a simple averaging argument, there is some ${X\in\P}$ such that $m(X\cap H(\vec{\alpha}))\geq\frac{\epsilon^{d}}{k}$. Thus, there exists some function $P:V\to\P$ such that for all $\vec{\alpha}\in V$, $m(P(\vec{\alpha})\cap H(\vec{\alpha}))\geq\frac{\epsilon^{d}}{k}$. We view $P$ as providing a labeling of each cell with a member of $\P$ that is sufficiently similar (hence the name $P$); this function will be the key to approximating $\P$ with cells.

  Next, we provide a bound on how many cells can be mapped to a particular member of $\P$. For any $X\in\P$ we have
  \begin{align*}
    M &> m(X) \tag{By hypothesis}\\
      &= m\left(\bigsqcup_{\vec{\alpha}\in V}X\cap H(\vec{\alpha})\right) \tag{Decomposition of $\R^{d}$}\\
      &\geq m\left(\bigsqcup_{\vec{\alpha}\in P^{-1}(X)}X\cap H(\vec{\alpha})\right) \tag{Subsets have equal or smaller measure}\\
      &\geq \sum_{\vec{\alpha}\in P^{-1}(X)}m(X\cap H(\vec{\alpha})) \tag{\Autoref{:disjoint-uncountable}}\\
      &\geq \sum_{\vec{\alpha}\in P^{-1}(X)}\frac{\epsilon^{d}}{k} \tag{Def'n of $P$ and $X=P(\vec{\alpha})$ for $\vec{\alpha}\in P^{-1}(X)$}\\
      &= \abs{P^{-1}(X)}\frac{\epsilon^{d}}{k}
  \end{align*}
  which shows that $\abs{P^{-1}(X)}<\frac{kM}{\epsilon^{d}}$.

  Next, we will define three binary relations on the set of vertices in order to arrive at a useful equivalence relation that will let us approximate $\P$ closely enough by using the cells (see \Autoref{sec:binary-relations} for details on binary relations.). As with the first two Optimality Theorems, we want to control the diameter in the approximation and also utilize the labeling of $P$, so we will end up saying cells are equivalent (will be part of the same member of the approximating partition) if they have the same $P$ label and are also close/connected. With this goal, let
  \begin{align*}
    R_{P}&=\set{(\vec{\alpha},\vec{\beta})\in V^{2}:P(\vec{\alpha})=P(\vec{\beta})}\\
    R_{\epsilon}&=\set{(\vec{\alpha},\vec{\beta})\in V^{2}:d_{max}(\vec{\alpha},\vec{\beta})\leq\epsilon}.
  \end{align*}
  Note that both relations are reflexive and symmetric, and that $R_{P}$ is also transitive (so $R_{P}$ is an equivalence relation).
  Let $R_{\epsilon}^{t}$ denote the transitive closure of $R_{\epsilon}$ (so $R_{\epsilon}^{t}$ is an equivalence relation), and define $R=R_{P}\cap R_{\epsilon}^{t}$ which is also an equivalence relation. Observe that the equivalence classes of $R_{P}$ are exactly the sets $P^{-1}(X)$ for $X\in\range(P)$, and each equivalence class of $R$ is a subset of an equivalence class of $R_{P}$. By these two observations and the previous paragraph, each equivalence class of $R$ is finite and has cardinality (strictly) less than $\frac{kM}{\epsilon^{d}}$.

  We can now show a type of bound on ``diameter''. Let $(\vec{\alpha},\vec{\beta})\in R$ be arbitrary. Let $R_{0}=R_{P}\cap R_{\epsilon}$ and observe that $R$ can be equivalently expressed as the transitive closure of $R_{0}$, so that there must be a sequence $\langle \vec{x}^{(j)}\rangle_{k=0}^{N}$ for some $N\in\N$ with $\vec{x}^{(0)}=\vec{\alpha}$ and $\vec{x}^{(N)}=\vec{\beta}$ and for all $j\in[N]$ that $(\vec{x}^{(j-1)},\vec{x}^{(j)})\in R_{0}\subseteq R_{\epsilon}$. Because of the bound on the cardinality of each equivalence class of $R$ we may assume $N<\frac{kM}{\epsilon^{d}}-1$ (the $-1$ is because of the zero-based indexing). Thus, we have

  \begin{align*}
    d_{max}(\vec{\alpha},\vec{\beta}) &\leq \sum_{j=1}^{N}d_{max}(\vec{x}^{(k-1)},\vec{x}^{(j)})\tag{Triangle inequality}\\
                                      &\leq \sum_{j=1}^{N}\epsilon\tag{$(\vec{x}^{(k-1)},\vec{x}^{(j)})\in R_{\epsilon}$}\\
                                      &\leq N\epsilon\\
                                      &< \left(\frac{kM}{\epsilon^{d}}-1\right)\epsilon\tag{Constraint on $N$}
  \end{align*}

  We are now in a position to define the approximation partition. As is common notation, let $\faktor{V}{R}$ denote the family of equivalence classes of $R$ (note that for each $C\in\faktor{V}{R}$ we have $C\subseteq V$). Then define the approximation partition as $\mathcal{A}=\set{\bigsqcup_{\vec{\alpha}\in C}H(\vec{\alpha}):C\in\faktor{V}{R}}$ which is a partition of $R^{d}$. For any $X\in\mathcal{A}$, there is an unique equivalence class $C\in\faktor{V}{R}$ such that $X=\bigsqcup_{\vec{\alpha}\in C}H(\vec{\alpha})$, and we denote this class as $C_{X}$. We will extend the distance argument above the members of $\P$. Let $X\in\mathcal{A}$ be arbitrary and let $\vec{a},\vec{b}\in X$ be arbitrary. Then there must exist some $\vec{\alpha}\in C_{X}$ with $\vec{a}\in H(\vec{\alpha})$ and similarly, there must exist some $\vec{\beta}\in C_{X}$ (possibly the same as $\vec{\alpha}$) with $\vec{b}\in H(\vec{\beta})$.
  \[
    d_{max}(\vec{a},\vec{b}) \leq d_{max}(\vec{a},\vec{\alpha}) + d_{max}(\vec{\alpha},\vec{\beta}) + d_{max}(\vec{\beta},\vec{b}) < \frac{\epsilon}{2} + \left(\frac{dM}{\epsilon^{d}}-1\right)\epsilon + \frac{\epsilon}{2} = \frac{dM}{\epsilon^{d-1}}.
  \]
  Thus, (noting the strict inequalities above), $\mathcal{A}$ satisfies the hypothesis of the Third Optimality Theorem (with $D=\frac{kM}{\epsilon^{d-1}}$ and $\vec{\alpha}=\vec{0}$), so there is some point $\vec{p}\in\R^{d}$ such that $\abs{\set{X\in\mathcal{A}:\closure{X}\ni\vec{p}}}\geq d+1$. We let $\vec{p}$ denote such a point for the remainder of the proof.

  Let $\mathcal{N}=\set{X\in\mathcal{A}:\closure{X}\ni\vec{p}}$ denote this set and observe that for any $X\in\mathcal{N}$ we have
  \[
    \vec{p} \in \closure{X} = \closure{\bigcup_{\vec{\alpha}\in C_{X}}H(\vec{\alpha})} = \bigcup_{\vec{\alpha}\in C_{X}}\closure{H(\vec{\alpha})}
  \]
  where the second equality is because the set $C_{X}$ is finite. Thus, there must be some $\vec{\alpha}\in C_{X}$ such that $\vec{p}\in\closure{H(\vec{\alpha})}$; fix such an $\vec{\alpha}$ for $X$ and denote it $\vec{\alpha}_{X}$. We will now show that for each $X,Y\in\mathcal{N}$, that $P(\vec{\alpha}_{X})\not=P(\vec{\alpha}_{Y})$ and that $\closure{B}_{\epsilon}(\vec{p})\cap P(\vec{\alpha}_{X})\not=\emptyset$. This will be enough to conclude the proof because this gives an injection from $\mathcal{N}$ to $\mathcal{N}_{\epsilon}(\vec{p})=\set{X\in\P:X\cap \closure{B}_{\epsilon}(\vec{p})}$ proving that it has cardinality at least $d+1$.

  Let $X,Y\in\mathcal{N}$ be arbitrary with $X\not=Y$. Then $C_{X}\not=C_{Y}$. We will show that $P(\vec{\alpha}_{X})\not=P(\vec{\alpha}_{Y})$ by showing that $(\vec{\alpha}_{X},\vec{\alpha}_{Y})\not\in R_{L}$ which we do by showing that $(\vec{\alpha}_{X},\vec{\alpha}_{Y})\not\in R$ and $(\vec{\alpha}_{X},\vec{\alpha}_{Y})\in R_{\epsilon}^{t}$. It is immediate that $(\vec{\alpha}_{X},\vec{\alpha}_{Y})\not\in R$ because $\vec{\alpha}_{X}\in C_{X}$ and $\vec{\alpha}_{Y}\in C_{Y}$ which are different equivalence classes of $R$. Regarding $R_{\epsilon}^{t}$, we have (by the definition of $\vec{\alpha}_{X}$ and $\vec{\alpha}_{Y}$) that
  \[
    \closure{B}_{\epsilon/2}(\vec{\alpha}_{X}) = \closure{H(\vec{\alpha}_{X})} \ni \vec{p} \in \closure{H(\vec{\alpha}_{Y})} = \closure{B}_{\epsilon/2}(\vec{\alpha}_{Y})
  \]
  so by the triangle inequality we have $d_{max}(\vec{\alpha}_{X},\vec{\alpha}_{Y})\leq\epsilon$ which means $(\vec{\alpha}_{X},\vec{\alpha}_{Y})\in R_{\epsilon}\subseteq R_{\epsilon}^{t}$. Thus we have established that $P(\vec{\alpha}_{X})\not=P(\vec{\alpha}_{Y})$.

  The last thing we need is to show that for any $X\in\mathcal{N}$ we have $\closure{B}_{\epsilon}(\vec{p})$ intersecting $P(\vec{\alpha}_{X})$. Since $\vec{p}\in\closure{H(\vec{\alpha}_{X})}=\closure{B}_{\epsilon/2}(\vec{\alpha}_{X})$, it follows that $\closure{B}_{\epsilon/2}(\vec{\alpha}_{X})\subseteq \closure{B}_{\epsilon}(\vec{p})$ (i.e. all points within $\epsilon/2$ of $\vec{\alpha}_{X}$ are within $\epsilon$ of $\vec{p}$ by the triangle inequality because $\vec{\alpha}_{X}$ is within $\epsilon/2$ of $\vec{p}$). Thus $H(\vec{\alpha}_{X})\subseteq \closure{B}_{\epsilon}(\vec{p})$. We will intersect both sides of this containment with $P(\vec{\alpha}_{X})$ recalling that by the definition of $P$, $P(\vec{\alpha}_{X})$ is a member of $\P$ such that $m(P(\vec{\alpha}_{X})\cap H(\vec{\alpha}_{X}))\geq\frac{\epsilon^{d}}{k}$, so in particular, this intersection is not empty. Thus
  \[
    \emptyset \not= P(\vec{\alpha}_{X})\cap H(\vec{\alpha}_{X})\subseteq P(\vec{\alpha}_{X})\cap \closure{B}_{\epsilon}(\vec{p}).
  \]
  In words, the $\epsilon$ ball around $\vec{p}$ intersects the member $P(\vec{\alpha}_{X})$ of $\P$. This completes the proof as we have shown that $\abs{\set{X\in\P:X\cap \closure{B}_{\epsilon}(\vec{p})}}\geq d+1$.
\end{proof}

\begin{remark}
  In the above proof, we have actually shown something stronger than the statement of the First Optimality Theorem. Not only are we guaranteed the existence of a point $\vec{p}$ where the $\epsilon$ ball intersects at least $d+1$ members of $\P$, but if we know some bound $k$ for the partition, then we are guaranteed a point where the intersection with at least $d+1$ members has substantial measure (at least $\frac{\epsilon^{d}}{k}$). This result also applies to the Second and Third Optimality Theorems (if the partitions are restricted to measurable members) because those theorems have otherwise stronger assumptions in the hypothesis.
\end{remark}

In this section we have shown that in our motivating question (\Autoref{ques:motivating}), the value $k=d+1$ is optimal not just for hypercube partitions, but also for any partition which has a uniform upper bound on the measures (or diameters) of the members (the First and Second Optimality Theorems). Further, if we strengthen this hypothesis to require that the diameters be uniformly upper bounded and require some finiteness, then we can actually conclude that there is a single point at the closure of $d+1$ members of the partition, which also implies that the partition has a $(d+1)$-clique (the Third Optimality Theorem). However, now that we know the optimal value of $k$, we can, in a sense, get this same strong conclusion even under the weaker assumptions if we insist that our partitions have the property of the motivating question for the optimal value of $k=d+1$.

\begin{proposition}\label{:boosting-with-optimal-k}
  If $d\in\N$, and $\P$ is a partition of $\R^{d}$, and there exists $D\in(0,\infty)$ such that for all $X\in\P$, it holds that $D$ is a strict pairwise bound for $X$, and if there exists $\epsilon\in(0,\infty)$ such that $\P$ has the property that for all $\vec{q}\in\R^{d}$,
  \[
    \abs{\mathcal{N}_{\epsilon}(\vec{q})}\leq d+1
  \]
  then there exists $\vec{p}\in\R^{d}$ such that
  \[
    \abs{\mathcal{N}_{\epsilon}(\vec{p})}=d+1.
  \]
  Furthermore, $\P$ contains a $(d+1)$-clique.
\end{proposition}
\begin{proof}
  Using the same techniques as in the proof of the Second Optimality theorem, pick $\vec{\alpha}\in\R^{d}$ and consider the hypercube $H=\vec{\alpha}+[0,D]^{d}$ and the non-spanning partition $\S$ of it induced by $\P$, and an admitted Sperner/KKM coloring so that there is some $\vec{p}\in H$ at the closure of at least $d+1$ colors. On the other hand, $\closure{B}_{\epsilon}(\vec{p})$ intersects at most $d+1$ members by assumption, and thus it intersects exactly one member of each of these $d+1$ colors. Further, for any $\delta\in(0,\epsilon]$, the ball $\closure{B}_{\delta}(\vec{p})$ will intersect exactly these same $d+1$ members, so that $\vec{p}$ belongs to the closure of these $d+1$ members (and it follows as before that these $d+1$ members form a clique).
\end{proof}

The significance of this result is that it makes a nice connection between the neighborhood property of the motivating question and the clique property of the partition graph.

\subsection{Optimality Theorem Gaps}
\label{subsec:gaps}

We now present two examples that demonstrate that there really is a ``gap'' between the three different Optimality Theorems. The conclusions get stronger with each successive version of the theorem, but each time the hypotheses were also made stronger (with the exception that the second and third theorems remove the requirement that the members be measurable). This begs the question of whether all three versions are necessary---for example, could it be that the hypothesis of the First Optimality Theorem are mathematically sufficient to imply the conclusions of the Third Optimality Theorem and we just did not find a proof? The answer is no; the strengthening of the hypotheses really is necessary and we prove this by providing two (counter)examples.

\begin{proposition}[First and Second Optimality Gap]
  The hypothesis of the First Optimality Theorem does not imply the conclusion of the Second Optimality Theorem for any $d\in\N$ with $d\geq2$. (And it does for $d=1$.)
\end{proposition}
\begin{proof}
  We will construct a partition which satisfies the hypothesis of the First Optimality Theorem and which does not satisfy the conclusion of the Second Optimality Theorem. 

    Recall that with respect to the $d_{max}$ metric, the ball of radius $r$ centered at the origin is $\closure{B}_{r}(0)=[-r,r]^{d}$. This ball (which is a hypercube) has diameter $2r$ and measure/volume $(2r)^{d}$.

  Consider the sequence $\set{r_{n}}_{n=1}^{\infty}$ where $r_{n}=\frac{1}{2}n^{1/d}$. Then $\closure{B}_{r_{n}}(0)$ has measure $n$. Let $S_{1}=\closure{B}_{r_{1}}(0)$ and inductively for $n>1$, let $S_{n}=\closure{B}_{r_{n}}\setminus \closure{B}_{r_{n-1}}$. All of the $S_{n}$ are disjoint by construction, and because $\set{r_{n}}_{n=1}^{\infty}$ increases without bound, they form a partition of $\R^{d}$. Further, the measure of $S_{1}$ is $1$, and for $n>1$, the measure of $S_{n}$ is $n-(n-1)=1$, the difference in measures of the two balls.

  Thus, every member of this partition has a measure of $1$, and because the $S_{n}$ are concentric, each with non-zero ``width'', for every point $\vec{p}\in\R^{d}$, there exists some $\epsilon(\vec{p})$ such that $\closure{B}_{\epsilon(\vec{p})}(\vec{p})$ intersects at most $2$ members of the partition. This shows that for $d\geq2$, the conclusion of the Second Optimality Theorem does not hold\footnote{This also shows that the conclusion of the Third Optimality Theorem does not hold for $d\geq2$ either since it is stronger, but we will prove this for all $d\in\N$ in the next proposition.}. For clarity, this construction works with $d=1$ as well, but if $d=1$, then $d+1=2$, so for each $\vec{p}\in\R^{d}$, $\closure{B}_{\epsilon(\vec{p})}(\vec{p})$ intersects at most $2=d+1$ members of the partition which does not contradict the conclusion of the Second Optimality Theorem.

  In the case $d=1$, suppose for contradiction that there is some partition $\P$ of $\R$ and some $M$ such that for all $X\in\P$, $m(X)<M$, (so the hypothesis of the First Optimality Theorem is satisfied) and suppose that there is some $\epsilon:\R\to(0,\infty)$ such that for all $\vec{p}\in\R^{d}$ that $\abs{\mathcal{N}_{\epsilon(\vec{p})}(\vec{p})}\leq1$ (so that the conclusion of the Second Optimality Theorem is not satisfied). Fix some $X\in\P$ and observe that for any $\vec{x}\in X$, $\closure{B}_{\epsilon(\vec{x})}(\vec{x})$ clearly intersects $X$, and so does not intersect any other member of $P$ by assumption (since $\abs{\mathcal{N}_{\epsilon(\vec{p})}(\vec{p})}\leq1$, so $\closure{B}_{\epsilon(\vec{x})}(\vec{x})$ intersects at most one member of the partition). Thus $\closure{B}_{\epsilon(\vec{x})}(\vec{x})\subseteq X$, and so $B^{\circ}_{\epsilon(\vec{x})}(\vec{x})\subseteq X$. Since $\vec{x}\in X$ was arbitrary, this shows that $X$ is an open set. Since $m(X)$ is finite, $X\not=\R$ and so $X$ is not closed (because $\R$ is a connected set). Fix some $\vec{p}\in\closure{X}\setminus X$. Then $\closure{B}_{\epsilon(\vec{p})}(\vec{p})$ intersects two members of $\P$---namely $X$ and $\member(\vec{p})$---which is a contradiction. Thus, the conclusion of the Second Optimality Theorem follows from the hypothesis of the First Optimality Theorem in the case $d=1$.
\end{proof}





\begin{proposition}[Second and Third Optimality Gap]
  The hypothesis of the Second Optimality Theorem does not imply the conclusion of the Third Optimality Theorem for any $d\in\N$. (Also, the hypothesis of the First Optimality Theorem does not imply the conclusions of the Third Optimality Theorem for any $d\in\N$.)
\end{proposition}
\begin{proof}
  We will construct a partition which satisfies the hypotheses of both the First and Second Optimality Theorems and which does not satisfy the conclusion of the Third Optimality Theorem.

  Consider the partition of singletons $\P=\set{\set{\vec{x}}:\vec{x}\in\R^{d}}$. All sets are measurable and have diameter $0$ (and thus measure $0$), and because each member is already closed, for any point $\vec{p}\in\R^{d}$, we have $\mathcal{N}_{\closure{0}}(\vec{p})=\set{X\in\P:\closure{X}\ni\vec{p}}=\set{X\in\P:X\ni\vec{p}}=\set{\vec{p}}$ which has size $1<d+1$. Thus, the conclusion of the Third Optimality Theorem does not hold.
\end{proof}

One might notice that this example demonstrating the gap between the Second and Third Optimality Theorems seems rather contrived for two reasons: (1) all members of the partition have measure 0, and (2) the partition consists of uncountably many members. If we consider a partition in which no member has measure 0, then this actually implies that the partition has countably many members (see the Appendix for details), so insisting on countable partitions is a weaker condition. We believe that countable partitions are a very natural restriction and we briefly explore this and make some connections to partitions of closed sets.

As we have shown, the hypothesis of the Second Optimality Theorem does not imply the conclusion of the Third Optimality Theorem, but the example that we gave demonstrating this seemed contrived as discussed above, and we believe that the finiteness condition that we added to that hypothesis in the Third Optimality Theorem may stronger than necessary---we wonder if replacing the finite requirement with a countable requirement in the Third Optimality Theorem is mathematically sufficient to justify the same conclusion.

\begin{conjecture}[Stronger Third Optimality Theorem]
    If $d\in\N$, and $\P$ is a partition of $\R^{d}$, and there exists $D\in(0,\infty)$ such that for all $X\in\P$, $D$ is a strict pairwise bound of $X$, and if there exists some $\vec{\alpha}\in\R^{d}$ such that $\vec{\alpha}+[0,D]^{d}$ intersects {\em countably} many members of $\P$,\\
    then there exists $\vec{p}\in\R^{d}$ such that
                           \[
                           \abs{\mathcal{N}_{\closure{0}}(\vec{p})}\geq d+1.
                           \]
\end{conjecture}

We can prove that the case of $d=1$ in the above conjecture is true; not only that, we can prove that it is equivalent to a known result about the inability to partition the unit interval with countably many closed sets. In the statement below, non-trivial means that the partition contains at least two members.

\begin{theorem}[\cite{sierpinski_1918}]
  There is no non-trivial partition of $[0,1]$ by countably many closed sets.
\end{theorem}

A proof of the above theorem utilizing the Baire Category Theorem can be found in \cite{countable_closed_partition_interval} and seems to be originally attributed to \cite{sierpinski_1918}. The natural generalization of this theorem to higher dimensions is not interesting because it is an immediate consequence of the theorem that there is no countable non-trivial partition of $[0,1]^{d}$ for any $d\in\N$. If there was such a partition $\S$, then pick two points $\vec{x},\vec{y}\in[0,1]^{d}$ belonging to different members of $\S$ and consider the convex hull (i.e. the line segment between them) which is isomorphic to $[0,1]$. Then $\S$ would induce a non-trivial countable partition of $[0,1]$ by closed sets giving the contradiction. However, we will prove now that the theorem above is equivalent to the case $d=1$ of our conjecture. Thus our conjecture (if true) serves as an interesting and fairly natural generalization of the above theorem.

\begin{remark}
  We will be talking about the topology on $\R$ as well as the subspace topology on $[0,1]$, but because $[0,1]$ is closed, every subset of $[0,1]$ is closed in the subspace topology if and only if it is closed in the topology on $\R$, so we need not distinguish between these and can say sets are closed without ambiguity.
\end{remark}

\renewcommand{\S}{\mathcal{S}}
\begin{proof}[Proof of equivalence]
  First assume the theorem and we will prove the conjecture for $d=1$. Let $\P$ be a partition of $\R^{1}$ with $D$, $\alpha$ as in the conjecture hypothesis. We may assume that $\alpha=0$ and $D=1$ without loss of generality\footnote{See the footnote in the proof of the Third Optimality Theorem.}. Then $\P$ induces a partition $\S$ on $[0,1]$ (specifically $\S=\set{X\cap[0,1]:X\in\P,\;X\cap[0,1]\not=\emptyset}$), and this partition is non-trivial because if $\S$ contains just one member then the member is $[0,1]$ which implies there is some member $X\in\P$ with $X\supseteq[0,1]$ which contradicts that $1$ is a strict diameter for all members of $\P$). $\S$ is countable as it has cardinality at most that of $\P$. By the theorem, $\S$ contains some set $Y$ which is not closed, so let $\vec{p}\in\closure{Y}\setminus Y\not=\emptyset$. Let $M=\member_{\S}(\vec{p})$. Since $\vec{p}\in M$ and $\vec{p}\not\in Y$, we have $M\not=Y$, so $M,Y\in\set{X\in\S:\closure{X}\ni\vec{p}}$ which shows this set has cardinality at least two. It follows that $\set{X\in\P:\closure{X}\ni\vec{p}}$ has cardinality at least $2$ as well\footnote{Take $Y'$ to be the member of $\P$ such that $Y'\cap[0,1]=Y$, and similarly for $M'$. Then $M'\not=Y'$ because otherwise $Y=Y'\cap[0,1]=M'\cap[0,1]=M$. Also, $\closure{M'}\supseteq M'\supseteq M\ni\vec{p}$, and $\closure{Y'}\supseteq\closure{Y}\ni\vec{p}$ so that $Y',M'\in\set{X\in\P:\closure{X}\ni\vec{p}}$} which proves the conjecture for $d=1$.

For the other direction, assume that the theorem is false and we will show that our conjecture is false with $d=1$. Let $\S$ be a non-trivial countable partition of $[0,1]$ by closed sets. We will modify $\S$ to construct $\S'$ which will have the distance requirements we need. If $\member_{\S}(0)\not=\member_{\S}(1)$, then let $\S'=S$. Otherwise, let $M=\member_{\S}(0)=\member_{\S}(1)$; we will split $M$ into disjoint closed sets. Since $\S$ is non-trivial, it contains at least two members, and thus $[0,1]\setminus M$ is non-empty, and since $M\ni0,1$ we have that $(0,1)\setminus M$ is non-empty and is an open set (in both topologies). Let $a$ be an arbitrary element of this open set, and let $\epsilon>0$ such that $(a-\epsilon,a+\epsilon)\subseteq(0,1)\setminus M$. In other words, $[0,a-\epsilon]\sqcup[a+\epsilon,1]\subseteq M$. Let $M_{0}=M\cap[0,a-\epsilon]$ and $M_{1}=M\cap[0,a+\epsilon]$, and let $\S'=\S\setminus{M}\cup\set{M_{0},M_{1}}$ (i.e. remove the member $M$ from $\S$ and add back two members $M_{0}$ and $M_{1}$). Since $M_{0}$ and $M_{1}$ are closed and disjoint and $M_{0}\sqcup M_{1}=M$ we have that $\S'$ is also a non-trivial countable partition of $[0,1]$ by closed sets, and has the additional property that for any $Y\in\S'$ and $x,y\in Y$, $d_{max}(x,y)<1$ (i.e. all members of $\S'$ have $1$ as a strict diameter). Now consider the partition $\P$ of $\R^{1}$ given by $\P=\S'\cup\set{\set{x}:x\in\R^{1}\setminus[0,1]}$ (i.e. the partition using the same members as $\S'$ and using singletons elsewhere). This satisfies the hypothesis of the conjecture with $\alpha=0$ and $D=1$, but all members of $\P$ are closed sets so $\set{X\in\P:\closure{X}\ni\vec{p}}=\set{X\in\P:X\ni\vec{p}}=\set{\member_{\P}(\vec{p})}$ which has cardinality $1$, thus the conclusion of the conjecture would not hold.
\end{proof}

\section{Upper Bounds on the Tolerance Parameter ($\epsilon$)}
\label{sec:epsilon}

In a particular reclusive partition, it was possible to center around every point in $\R^{d}$ a closed ball of radius $\frac{1}{2d}$ (in the $d_{max}$ metric) which would only intersect $d+1$ members of the partition. In the previous section, we argued that this is optimal if there is a bound on the size of the partition members, (and that otherwise there are trivial counterexamples). 
We view this value as the primary concern, and now that we know $d+1$ is optimal, we consider a secondary concern which is the value of $\epsilon$ in the motivating question. 

A simple argument shows that if all elements of $\P$ have diameter at most $D$, then it must be that $\epsilon\leq\frac{D}2$.
\RestatableTrivialToleranceBound
\begin{proof}
  We show that for any $\epsilon>\frac{D}2$, there is some point $\vec{p}\in\R^{d}$ such that $\closure{B}_{\epsilon}(\vec{p})$ intersects $2^{d}+1$ members of $\P$. Let $\epsilon>\frac{D}2$ and fix any $X\in\P$. Then for each $i\in[d]$, let $a_{i}=\inf\set{x_{i}:\vec{x}\in X}$ and let $b_{i}=\sup\set{x_{i}:\vec{x}\in X}$ noting that $b_{i}-a_{i}\leq D$. Thus $X\subseteq\prod_{i=1}^{d}[a_{i},b_{i}]\subseteq\prod_{i=1}^{d}[a_{i},a_{i}+D]=\vec{a}+[0,D]^{d}$. Let $\vec{p}=\langle a_{i}+\frac{D}2\rangle_{i=1}^{d}$, so $X\subseteq \closure{B}_{\frac{D}2}(\vec{p})\subset \closure{B}_{\epsilon}(\vec{p})$. For any distinct $\vec{\alpha},\vec{\beta}\in\set{-\epsilon,\epsilon}^{d}$, we have $\vec{\alpha}+\vec{p},\vec{\beta}+\vec{p}\in\closure{B}_{\epsilon}(\vec{p})$, and $\vec{\alpha}+\vec{p},\vec{\beta}+\vec{p}\not\in X$, and $d_{max}(\vec{p}+\vec{\alpha}, \vec{p}+\vec{\beta})=d_{max}(\vec{\alpha},\vec{\beta})=2\epsilon>D$ which implies $\vec{p}+\vec{\alpha}$ and $\vec{p}+\vec{\beta}$ belong to different members of $\P$. Thus $\abs{\set{X\in\P:X\cap\closure{B}_{\epsilon}(\vec{p})\not=\emptyset}}\geq2^{d}+1$.
\end{proof}

By considering the specific value $d=1$, we get the following corollary.

\RestatableOptimalToleranceROne
\begin{proof}
Since $D$ is a strict pairwise bound for each member of $\P$, then every member of $\P$ has diameter at most $D$. Apply \Autoref{:optimal-diam-d1} with $d=1$.
\end{proof}
 
Note that the reclusive partitions can be easily scaled from unit hypercubes to $D$-sidelength hypercubes, in which case a value of $k=d+1$ and $\epsilon=\frac{D}{2d}$ can be achieved. The above proposition shows that for $\R^{1}$ (when $d=1$), partitions with members with strict pairwise bound of $D$ (or diameter at most $D$) must have $\epsilon\leq\frac{D}{2}=\frac{D}{2d}$, and thus (certain) reclusive partitions attain the maximal value of $\epsilon$ when $k$ is minimized at $d+1$.

The argument above did not use any interesting properties of the partitions though. All we did was to say that if the radius $\epsilon$ is larger than $\frac{D}2$, then it is possible to strictly contain an entire member within a hypercube and allow the corners far enough apart so that they must each belong to a different member. In order to get better bounds than this, we will need some additional properties. What we will show in this section is that the statements of the Optimality Theorems were actually quite a bit weaker than they could have been. In the first two optimality theorems, we picked an aribtrary $D$-sidelength hypercube and showed the existence of a point $\vec{p}$ with the desired properties in that cube. Obviously we could have picked any cube, and so there are infinitely many points $\vec{p}$ as in those two theorems. Not only this, the results of \cite{de_loera_polytopal_2002} show that because the $d$-dimensional hypercube is a $d$-dimensional polytope with $2^{d}$ vertices, we can actually find (basically\footnote{There is a caveat that the results of \cite{de_loera_polytopal_2002} are stated in terms of how many subsets of $d+1$ colors have an intersection, so if say $d+1+n$ colors intersect, that accounts for ${d+1+n}\choose{d+1}$ of the $2^{d}-d$ points. When we actually apply these results in this section, the hypothesis that the parameter $k$ from our motivating question is optimal at $d+1$ will ensure we don't have to worry about this multiplicity.}) $2^{d}-d$ such points in each $D$-sidelength hypercube. Further, their proof actually gives an even stronger result that we can use to get a general upper bound of $\epsilon\leq\frac1{2\left((d+1)^\frac{d-1}{2d}-1\right)}$, and prove that the value of $\epsilon=\frac{D}{2d}$ is optimal for $d=2$.

Before moving on, we show that the (quite ugly) upper bound we just mentioned is bounded above by and also asymptotically equivalent to the much cleaner bound $\frac{1}{2\sqrt{d}}$.

\begin{lemma}\label{:asymptotic-one-over-two-root-d}
The function $f(x)=\frac1{2\left((x+1)^\frac{x-1}{2x}-1\right)}$ is asymptotically equivalent to the function $g(x)=\frac{1}{2\sqrt{x}}$ (i.e. $\lim_{x\to\infty}\frac{g(x)}{f(x)}=1$). Furthermore, for $x>1$, $f(x)<g(x)$.
\end{lemma}
\begin{proof}
Note first that
\[
\lim_{x\to\infty}x^{1/x}=\lim_{x\to\infty}e^{\ln(x)/x}=e^0=1
\]
and so
\[
\lim_{x\to\infty}x^{(-1)/(2x)}=\lim_{x\to\infty}\frac{1}{\sqrt{x^{1/x}}}=1.
\]

We will apply the squeeze theorem to the following inequalities which hold for $x>0$.
\begin{align*}
 1 < \left(\frac{x}{x+1}\right)^{1/2} &= \frac{x^{1/2}}{(x+1)^{1/2}} \\
                                  &\leq \frac{x^{1/2}}{(x+1)^{(x-1)/(2x)}-1} \qquad=\frac{g(x)}{f(x)} \\
                                  &\leq \frac{x^{1/2}}{x^{(x-1)/(2x)}-1} \\
                                  &\leq \frac{1}{x^{(-1)/(2x)}-\frac{1}{\sqrt{x}}} \tag{divide top and bottom by $\sqrt{x}$}
\end{align*}
The limit as $x\to\infty$ of the first and last expressions is $1$, so by the squeeze theorem, $\lim_{x\to\infty}\frac{g(x)}{f(x)}=1$. Furthermore, since $1<\frac{g(x)}{f(x)}$ (for $x>1$) and $f(x)$ is positive for all $x>1$, it follows that $f(x)<g(x)$ for $x>0$.
\end{proof}


\begin{definition}[Simplicial Subdivision]
  A simplicial subdivision (sometimes called a triangulation) of a hypercube $H$ is a finite set of $d$-simplices $\Sigma=\set{\sigma_{1},\ldots,\sigma_{m}}$ such that $H = \bigcup_{i=1}^{m}\sigma_{i}$, and for every $i,j\in[m]$ with $i\not=j$, $\sigma_{i}\cap\sigma_{j}$ is either empty, or a face of both $\sigma_{i}$ and of $\sigma_{j}$. The vertices of a simplicial subdivision are the vertices of all its $d$-simplices; that is, $V(\Sigma)\defeq\bigcup_{i=1}^{m}V(\sigma_{i})$.
\end{definition}


The following definition is comparable to \Autoref{:sperner-kmm-defn}.

\begin{definition}[Sperner Coloring]\label{:sperner-coloring-defn}
  Let $d\in\N$, and $H=[0,1]^{d}$, and $V(H)=\set{0,1}^{d}$ denote a set of colors (which is exactly the set of vertices of $[0,1]^{d}$ so that colors and vertices are identified). Let $\Sigma=\set{\sigma_{1},\ldots,\sigma_{m}}$ be a simplicial subdivision of $H$. Let $\chi:V(\Sigma)\to V(H)$ be a coloring function such that for any face $F$ of $[0,1]^{d}$, for any $\vec{x}\in F\cap V(\Sigma)$, it holds that $\chi(\vec{x})\in F$ (informally, the color of $\vec{x}$ is one of the vertices in the face $F$). Then $\Sigma$ along with $\chi$ is called a Sperner coloring of $H$.

  Further, for any $\sigma\in\Sigma$, we define $\mathrm{colorset}(\sigma)=\set{\chi(\vec{x}):\vec{x}\in V(\sigma)}$.
\end{definition}


The above definitions are easily generalized to convex polytopes (but we will not need the more general definitions except for the discussion in this paragraph), and with the generalized definitions, De Loera, Peterson, and Su showed in \cite{loera_polytopal_2001} that for any convex polytope $P$ in $d$ dimensions with $n$ vertices, any Sperner coloring $(\Sigma,\chi)$ of $P$ will have at least $n-d$ ``fully colored simplices''. That is, $\Sigma$ must contain $n-d$ simplices $\sigma_{1},\ldots\sigma_{n-d}$ such that $\abs{\mathrm{colorset}(\sigma_{i})}=d+1$ (recall that $\sigma_{i}$ only has $d+1$ vertices, so this means that each vertex of $\sigma_{i}$ is assigned a different color). In fact, they show that if $i\not=j$, then $\mathrm{colorset}({\sigma_{i}})\not=\mathrm{colorset}({\sigma_{j}})$. From another perspective, they say that there are at least $(n-d)$ many $(d+1)$-cardinality subsets of the colors, each being the colorset of some simplex in the subdivision.

They argue that this result ($n-d$) is tight when $P$ is an arbitrary convex polytope, but that improvements may be made when restricting to specific types of convex polytopes (e.g. hypercubes).
In fact, their bound can be improved significantly when restricting to hypercubes by combining the results of \cite{loera_polytopal_2001} with other results about simplicial decompositions of hypercubes.

Next, we give a definitional name to the best possible parameter for each dimension $d$. We could get away without defining this value in our paper and just use bounds on this value (since that is all we will use anyway), but we want to make this quantity explicit since it relates nicely to some other areas of research regarding the $d$-dimension hypercube.

\begin{definition}[Sperner Number]\label{:sperner-number-defn}
  Let $d\in\N$, and $H=[0,1]^{d}$, and let $S_{d}$ be the set of all values $s\in\N\cup\set{0}$ such that the following statement holds:
  \begin{quotation}
    \noindent For any Sperner coloring $(\Sigma,\chi)$ of $H$, there exists at least $s$-many $(d+1)$-cardinality sets $J_{1},\ldots,J_{s}\subseteq V(H)$ such for each $J_{i}$, $\Sigma$ contains a simplex $\sigma_{i}$ such that $\mathrm{colorset}(\sigma_{i})=J_{i}$.
  \end{quotation}
  We call the minimum value of $S_{d}$ the Sperner number of $[0,1]^{d}$ or the $d$th Sperner number, and we denote it by $\sperner(d)\defeq\min(S_{d})$.
\end{definition}

Since $[0,1]^{d}$ is a convex polytope in $d$ dimensions with $2^{d}$ vertices, it follows immediately from \cite{loera_polytopal_2001} that any $\sperner(d)\geq 2^{d}-d$. In fact, they implicitly showed something stronger than this. By \cite[Thm.~1]{loera_polytopal_2001} and the comment following the proof of \cite[Cor.~3]{loera_polytopal_2001} on pages 18-19 of the July 2001 version 8, the collection of colorsets of simplices in any simplicial subdivision induces a ``face-to-face simplicial cover'' of $[0,1]^{d}$ using those colorsets to define simplices (in the terminology from \cite{glazyrin_lower_2012}, a ``dissection'' of $[0,1]^{d}$). Thus, $\sperner(d)$ is at least as large as the number of simplices needed in a dissection of $[0,1]^{d}$. In other words, the dissection number of the $d$-cube gives a lower bound on the Sperner number.

We can also get a trivial upper bound on the $d$th Sperner number by noting that, by definition, it can be no larger than the size of the minimal cardinality simplicial subdivision of $[0,1]^{d}$. Similarly, it can be no larger than the size of the minimal cardinality triangulation\footnote{We use the term triangulation as in \cite{glazyrin_lower_2012}, which is different from how it is used in \cite{loera_polytopal_2001}. By triangulation, we mean a simplicial subdivision $\Sigma$ such that $V(\Sigma)=V([0,1]^{d})$ (i.e. the only vertices used in the subdivision are ones from the original hypercube).} of $[0,1]^{d}$---this is because any triangulation is a valid simplicial subdivision.

Summarizing the above three paragraphs using the notation in \cite{glazyrin_lower_2012}, we have the following chain of inequalities for properties of $[0,1]^{d}$.
We emphasize that all quantities below are with respect to using only vertices of the $[0,1]^{d}$ (i.e. no extra vertices are allowed). For example, Below \textit{et. al.} show in \cite{below_minimal_2000} that the use of extra vertices can drastically reduce the necessary size of a triangulation.

\begin{equation}
  \label{eq:cube-ineq}
  \mathrm{cover}(d) \leq \mathrm{dis}(d) \leq \sperner(d) \leq \mathrm{triang}(d)
\end{equation}

Glazyrin showed in \cite{glazyrin_lower_2012} that $(d+1)^{\frac{d-1}{2}}\leq\mathrm{dis}(d)$ and Orden and Santos showed in \cite{orden_asymptotically_2003} that $\mathrm{triang}(d)\in O(0.816^{d}d!)$ which are the best known asymptotic bounds to date, and they provide upper and lower\footnote{For $d<5$, the lower bound of $\sperner(d)\geq 2^{d}-d$ can be used instead of $\sperner(d)\geq(d+1)^{\frac{d-1}{2}}$, however, noting that $\sperner(d)$ is an integer, this actually only gives an improved lower bound in the case $d=3$, providing a bound of $5$ instead of $4$.} bounds on the $d$th Sperner number.

While it is known for general polytopes that the dissection number and the triangulation number are not equal (see an example in \cite{below_minimal_2000}), for hypercubes it is still an open question if $\mathrm{dis}(d)$ equals $\mathrm{triang}(d)$ or not. We provide the first few values of $\sperner(d)$ which are exactly the values where $\mathrm{triang}(d)$ is know to equal $\mathrm{dis}(d)$ (see \cite[Table~1]{glazyrin_lower_2012}).

\begin{table}[h]
  \centering
  \begin{tabular}{|c|c|}
    d&$\sperner(d)$\\\hline
    1&1\\
    2&2\\
    3&5\\
    4&16\\
  \end{tabular}
  \caption{Known values of $\sperner(d)$.}
\end{table}


The next lemma is really the result that we want regarding the quantity $\sperner(d)$ because we have interest not in colorings of simplicial decompositions, but colorings of the entire hypercube as in \Autoref{:sperner-kmm-defn}. Due to the compactness of $[0,1]^{d}$, we can transfer the defining property of $\sperner(d)$ from Sperner colorings of simplicial subdivisions to Sperner/KKM colorings. The technique to do so is the same one that is used in many proofs that use variations of Sperner's lemma to prove variations of the KKM lemma.

\begin{lemma}\label{:polytope-kmm}
  Let $d\in\N$ and $H=[0,1]^{d}$ so $V(H)=\set{0,1}^{d}$. In any Sperner/KKM coloring $\chi:[0,1]^{d}\to V(H)$, there are $\sperner(d)$ distinct $(d+1)$-cardinality sets $J_{1},\ldots,J_{\sperner(d)}\subseteq V(H)$ such that $\closure{\cap_{v\in V}\chi^{-1}(v)}\not=\emptyset$.
\end{lemma}
\begin{proof}
  The proof is identical to the proof of \cite[Cor.~3]{loera_polytopal_2001} with the exception that we know each simplicial subdivision in the sequence (what they call triangulations) contains not just (in their notation) ``$c(P)$'' different colorsets, but in fact $\sperner(d)$ different colorsets.
\end{proof}

With this lemma, we can now show that the value of $\epsilon=\frac{D}{2d}$ is optimal for $d=2$.

\RestatableOptimalToleranceInRTwo
\begin{proof}
  We show that for any $\epsilon>\frac{D}4$, there is some point $\vec{p}\in\R^{d}$ such that $\closure{B}_{\epsilon}$ intersects at least $d+2$ members of $\P$. Let $\epsilon>\frac{D}4$. By the Second Optimality Theorem there is some $\vec{p_{0}}\in\R^{d}$ such that $\abs{\set{X\in\P:X\cap \closure{B}_{\epsilon}(\vec{p_{0}})\not=\emptyset}}\geq d+1=3$ (so along with the hypothesis we actually have cardinality exactly $d+1=3$). Consider the hypercube $H=\closure{B}_{\frac{D}2}(\vec{p_{0}})$, and let $\S$ be the partition of $H$ induced by $\P$, and let $\chi$ be an admitted Sperner/KKM coloring. By \Autoref{:boosting-with-optimal-k}, since $\sperner(2)=2$, there exists distinct points $\vec{x},\vec{y}\in H$ (one of them possibly equal to $\vec{p_{0}}$) each belonging to the closure of $d+1$ members of $\S$ (and thus $\P$), denoted $\mathcal{N}_{\vec{x}}$ and $\mathcal{N}_{\vec{y}}$, and $\mathcal{N}_{\vec{x}}\not=\mathcal{N}_{\vec{y}}$. Thus, one of these points (wlog $\vec{x}$) is at the closure of a different set of $d+1$ members of $\S$ than $\vec{p_{0}}$ is. Since $\vec{x}\in H$, $d_{max}(\vec{x},\vec{p_{0}})\leq\frac{D}2$. Consider the midpoint $\vec{c}=\frac{\vec{x}+\vec{p_{0}}}{2}$ observing that $d_{max}(\vec{x},\vec{c})\leq\frac{D}4$ and $d_{max}(\vec{p_{0}},\vec{c})\leq\frac{D}4$, and since $\epsilon>\frac{D}4$, $\closure{B}_{\epsilon}(\vec{c})$ contains an open set around $\vec{x}$ and an open set around $\vec{p_{0}}$, and thus intersects the $d+1$ members containing $\vec{x}$ in their closure and also the $d+1$ members containing $\vec{p_{0}}$ in their closure. Because these sets are not the same, $\closure{B}_{\epsilon}(\vec{c})$ intersects at least $d+2$ members of $\S$ and thus $\P$.
\end{proof}

Thus, we have shown that value of $\epsilon$ we are able to achieve with the reclusive partitions is optimal in $\R^{1}$ and $\R^{2}$, but these proofs were both special cases that resulted from the $\epsilon$ values in question being large enough relative to the diameter to easily argue about. We will now provide an upper bound on $\epsilon$ for all dimension $d\geq2$, and while the technique will be more general than the technique for $\R^{2}$, it will use the same basic idea of locating points at the closures of $d+1$ members, arguing about distances between these points, and limiting the number of occurences of such points based on $\epsilon$.

\begin{lemma}\label{:upper-bound-lemma}
  If $d\in\N$, and $\P$ is a partition of $\R^{d}$, and there exists $D\in(0,\infty)$ such that for all $X\in\P$, it holds that $D$ is a strict pairwise bound for $X$, and if there exists $\epsilon\in(0,\infty)$ such that for all $\vec{p}\in\R^{d}$,
  \[
    \abs{\mathcal{N}_{\epsilon}(\vec{p})}\leq d+1
  \]
  then for any $\vec{\alpha}\in\R^{d}$, the hypercube $H=\vec{\alpha}+[0,D]^{d}$ contains at least $\sperner(d)$ distinct points $\vec{s}_{1},\ldots\vec{s}_{\sperner(d)}\in H$ such that for each $i,j\in[\sperner(d)]$,
  \[
    \abs{\mathcal{N}_{\closure{0}}(\vec{s}_{i})}=d+1,
  \]
  and if $i\not=j$, then $\mathcal{N}_{\closure{0}}(\vec{s}_{i})\not=\mathcal{N}_{\closure{0}}(\vec{s}_{j})$.
\end{lemma}
\begin{proof}
  As in other proofs, consider the set $H=\vec{\alpha}+[0,D]^{d}$, consider the partition $\S$ induced by $\P$ which is non-spanning and admits a Sperner/KKM coloring. So by \Autoref{:polytope-kmm}, there are at least $\sperner(d)$ distinct $(d+1)$-cardinality sets $J_{1},\ldots,J_{\sperner(d)}\subseteq V(H)$ such that $\closure{\cap_{v\in V}\chi^{-1}(v)}\not=\emptyset$ (i.e. for each set $J_{i}$ of $d+1$ colors, there is some point in the closure of all colors in $J_{i}$). For each $i$, let $\vec{s}_{i}\in\closure{\cap_{v\in V}\chi^{-1}(v)}$ be arbitrary (i.e. $\vec{s}_{i}$ is in the closure of all colors in $J_{i}$).

  We now justify that if $i\not=j$, the $\vec{s}_{i}\not=\vec{s_{j}}$. By exactly the same argument as \Autoref{:boosting-with-optimal-k}, we have $\abs{\set{X\in\P:\closure{X}\ni\vec{s}_{i}}}=d+1$. It follows immediately that $\set{X\in\P:\closure{X}\ni\vec{s}_{i}}\not=\set{X\in\P:\closure{X}\ni\vec{s}_{j}}$ because $J_{i}\not=J_{j}$ and the set $\set{X\in\P:\closure{X}\ni\vec{s}_{i}}$ contains exactly one member of each color in $J_{i}$ (and similarly for $J_{j}$).
\end{proof}

\RestatableUniversalToleranceBound
\begin{proof}
  For any $\vec{\alpha}\in\R^{d}$, consider the hypercube $H=\vec{\alpha}+[0,D]^{d}$, and (as in many of our other proofs) an admitted Sperner/KKM coloring of an induced partition. Then let $\vec{s}_{1},\ldots\vec{s}_{\sperner(d)}\in H$ be points as in \Autoref{:upper-bound-lemma}. It follows that for any $i\not=j$ that $d_{max}(\vec{s}_{i},\vec{s}_{j})\leq\epsilon$ because if they were not, then the midpoint $\vec{c}=\frac{\vec{s}_{i}+\vec{s}_{j}}{2}$ would have the property that $\closure{B}_{\epsilon}(\vec{c})$ contains an open set around $\vec{s}_{i}$ and around $\vec{s}_{i}$, and thus $\closure{B}_{\epsilon}(\vec{c})$ would intersect strictly more than $d+1$ members of $\P$ which would contradict the hypothesis.

  Thus, the points $\vec{s}_{i}$ give rise to closed balls $\closure{B}_{\epsilon}(\vec{s}_{i})$ which are pairwise disjoint. Further, because $\vec{s}_{i}\in H=\vec{\alpha}+[0,D]^{d}$, it follows that $\closure{B}_{\epsilon}(\vec{s}_{i})\subseteq \vec{\alpha}+[-\epsilon,D+\epsilon]^{d}$. This gives the following volume/measure comparison argument:
  \[
    \sperner(d)\cdot(2\epsilon)^{d} = m\left(\bigsqcup_{i=1}^{\sperner(d)}\closure{B}_{\epsilon}(\vec{s}_{i})\right) \leq m(\vec{\alpha}+[-\epsilon,D+\epsilon]^{d}) = (D+2\epsilon)^{d}.
  \]
  Taking $d$th roots of both sides and manipulating the equations, we have the stated inequality for $\epsilon$ for $d\geq2$ (with $d=1$, there would be a division by $0$ because $\sperner(1)=1$).

  To show the ``in particular'' statement, for $d\geq2$, increase the bound by replacing $\sperner(d)$ with the lower bound $(d+1)^{\frac{d-1}{2}}\leq\sperner(d)$ (\Autoref{eq:cube-ineq} and \cite{glazyrin_lower_2012}) (note that for $d=1$ this would give division by zero). Lastly then apply \Autoref{:asymptotic-one-over-two-root-d}. 
  
  To show the ``in particular'' statement, for $d=1$, note that $d=\sqrt{d}$, so apply \Autoref{:optimal-tolerance-R1}.
\end{proof}

Note that it is not an issue that the above bound does not hold for $d=1$ because we already have exact bounds for $d=1$ and $d=2$ in \Autoref{:optimal-diam-d1} and \Autoref{:optimal-diam-d2}. Nonetheless, we don't believe that the $-1$ term in the above expression is necessary, but we have not yet come up with an argument that removes it.

Since the first four values of $\sperner(d)$ are known, it is the first bound in the theorem above that gives the better bound. For convenience, we summarize the bounds on $\epsilon$ for $d\in\set{1,2,3,4}$.

\renewcommand{\arraystretch}{1.5}
\begin{table}[h]
  \centering
  \begin{tabular}{|c|c|c|c|}\hline
    d & $\sperner(d)$ & $\epsilon$ upper bound                             & Reason \\\hline
    1 & 1             & $\frac{D}2$                                           & \Autoref{:optimal-diam-d1} \\\hline
    2 & 2             & $\frac{D}{4}$                                         & \Autoref{:optimal-diam-d2} \\\hline
    3 & 5             & $\frac{D}{2(\sqrt[3]{5}-1)}\approx\frac{D}{1.419952}$ & \Autoref{:sperner-upper-bound} \\\hline
    4 & 16            & $\frac{D}{2(\sqrt[4]{16}-1)}=\frac{D}2$               & \Autoref{:sperner-upper-bound} \\\hline
  \end{tabular}
  \caption{Summary of upper bounds on $\epsilon$.}
\end{table}

While the above gives an upper bound whose denominator grows only as the square root of $d$, we conjecture that the true maximal value of $\epsilon$ has a tighter bound with the denominator growing linearly in $d$. If this conjecture is true, that means that reclusive partitions can achieve within a constant factor of the maximum possible value of $\epsilon$.

\begin{conjecture}[Linear Universal Tolerance ($\epsilon$) Conjecture]\label{:linear-conjecture}
  If $d\in\N$, and $\P$ is a partition of $\R^{d}$, and there exists $D\in(0,\infty)$ such that for all $X\in\P$, $\diam(X)\leq D$, and if $\epsilon\in(0,\infty)$ such that for all $\vec{p}\in\R^{d}$,
  \[
    \abs{\mathcal{N}_{\epsilon}(\vec{p})}\leq d+1
  \]
  then $\epsilon=\frac{D}{\Omega(d)}$.
\end{conjecture}

In fact, if it holds that $\sperner(d)\geq\frac{d!}{c^{d}}$ for some constant $c\in(0,\infty)$, or even if $\sperner(d)\geq\left(\frac{d}{c'}\right)^{d}$, then the conjecture above holds. This is consistent with the current bounds mentioned earlier in this section due to \cite{glazyrin_lower_2012} and \cite{orden_asymptotically_2003}:
\[
  (d+1)^{\frac{d-1}{2}}\leq\mathrm{dis}(d) \leq \sperner(d) \leq \mathrm{triang}(d)\in O(0.816^{d}d!).
\]
To see that the conjecture would hold in this case, note that we can use Stirling's approximation as a lower bound for the factorial case above:
\[
  \left(\frac{d}{ce}\right)^{d} \leq \frac{\sqrt{2\pi d}\left(\frac{d}{e}\right)^{d}}{c^{d}} \leq \frac{d!}{c^{d}} \leq \sperner(d).
\]
So the first case above reduces to the second case.

So assuming $\sperner(d)\geq\left(\frac{d}{c'}\right)^{d}$ for some constant $c'\in(0,\infty)$, and using this lower bound in the previous upper bound for $\epsilon$, we have
\[
  \epsilon \leq \frac{D}{2(\sqrt[d]{\sperner(d)}-1)} \leq \frac{D}{2(\frac{d}{c}-1)}
\]
showing that under these assumptions $\epsilon=\frac{D}{\Omega(d)}$.
\section{Application of Reclusive Partitions to Deterministic Rounding}
\label{sec:algorithm}
\renewcommand{\P}{\mathcal{P}}

Throughout this section, let $A$ denote a $d\times d$ reclusive matrix. Other notation from \Autoref{sec:reclusive-lattice-partitions} will be used as well. Our first goal of this section will be to show that given $\vec{x}\in\R^{d}$ we can efficiently compute $\rep(X)$ for the unique $X\in\P_{A}$ such that $\vec{x}\in X$.
In other words, we can quickly map points of $\R^{d}$ to the much sparser set of representative points $\set{\rep(X)\colon X\in\P_{A}}=L_{A}$.
First observe a simple fact.

\begin{fact}\label{:floor}
  Let $a\in\R$ and $b\in[0,1)$ such that $a-b\in\Z$.
  Then $a-b=\floor{a}$.
\end{fact}
\begin{proof}
  If $a-b>\floor{a}$ then $a-b\geq\floor{a}+1$, so $a-\floor{a}\geq1+b\geq1$ which is a contradiction.
  If $a-b<\floor{a}$ then $a-b\leq\floor{a}-1$, so $a-\floor{a}\leq b-1<0$ implying that $a<\floor{a}$ which is a contradiction.
  Thus $a-b=\floor{a}$.
\end{proof}

\begin{proposition}[Efficient Computation of Representatives]\label{:efficent-computation-of-representatives}
  Let $d\in\N$, and $A$ be a $d\times d$ reclusive matrix, and $\P_{A}$ its reclusive partition.
  For any $\vec{x}\in\R^{d}$, let $X\in\P_{A}$ be the unique hypercube such that $\vec{x}\in X$.
  Then $\rep(X)$ can be efficiently\footnote{By efficient, we mean that the computation can be done with $O(d)$ matrix multiplications with matrices of size $d\times d$.} computed in terms of $\vec{x}$ and $A$.\\
\end{proposition}

\begin{remark}
  For additional intuition of the following proof, see the proof of \Autoref{:example-partition-proof} which is an inductive proof instead of an algebraic one.
\end{remark}

\begin{proof}
  Recall that $\rep(X)=A\vec{m}$ for a unique $\vec{m}\in\Z^{d}$, so it will suffice to compute $\vec{m}$ because $\rep(X)$ can then be computed via a single matrix multiplication.
  We show by induction that if we have computed $m_{i}$ for all $i>k$, then we can compute $m_{k}$.
  The main reason that we can do this is that $A$ is a triangular matrix, so the technique for computing each $m_{i}$ has the flavor of Gaussian elimination.
  The inductive base case is that we have not computed any $m_{i}$.


  By the definition of $\P_{A}$ and $\rep(X)$, we have that $\vec{x}\in\rep(X)+[0,1)^{d}$, so let $\vec{\alpha}=\vec{x}-\rep(X)$ so $\vec{\alpha}\in[0,1)^{d}$.
  Now we consider just the $k$th coordinate.
  \begin{align*}
    x_{k} &= \alpha_{k}+\rep(X)_{k}\\
          &= \alpha_{k}+(A\vec{m})_{k}\\
          &= \alpha_{k}+\sum_{i=1}^{d}a_{ki}m_{i}\tag{Def'n of matrix multiplication}\\
          &= \alpha_{k}+\sum_{i=k}^{d}a_{ki}m_{i}\tag{$A$ is reclusive, so $a_{ki}=0$ for $i<k$}\\
          &= \alpha_{k}+m_{k}+\sum_{i=k+1}^{d}a_{ki}m_{i}\tag{$A$ is reclusive, so $a_{kk}=1$. Summation might be empty}\\
  \end{align*}
  We now reformulate in terms of $m_{k}$.
  \begin{align*}
    m_{k} &= x_{k}-\sum_{i=k+1}^{d}a_{ki}m_{i}-\alpha_{k}\tag{Solve for $m_{k}$}\\
          &= \floor{x_{k}-\sum_{i=k+1}^{d}a_{ki}m_{i}}\tag{\Autoref{:floor}}
  \end{align*}
  Thus, $m_{k}$ can be computed as a floor in terms of $A$, $\vec{x}$, and the already known $m_{i}$ for $i>k$.
  As mentioned, we can return vector $\rep(X)=A\vec{m}$.

  Altogether, this computation requires $O(d^{2})$ additions and multiplications. We need $O(d)$ to compute each $m_{i}$ and $i$ takes $d$ many values, and we need $O(d^{2})$ operations to compute $A\vec{m}$.
\end{proof}

Having shown that we can efficiently compute the representative corners, we turn to an application in pseudodeterministic computations. Consider a multi-valued function such as $f(1^n)=$the set of $n$-bit primes. We know a probabilistic polynomial-time for this function.  A deficiency of this algorithm is that two different runs of the algorithm may produce two different valid outputs. Is there a probabilistic algorithm that outputs a canonical prime number? I.e, most of the random choices of the algorithm will produce the same output. Motivated by this, Gat and Goldwasser~\cite{gat_probabilistic_2011} defined the notion of {\em pseudodeterministic algorithms}. 

\begin{definition}
Let $f$ be a total multi-valued function\footnote{A multi-valued function maps inputs to non-empty sets of outputs.}. We say that $f$ admits a {\em polynomial-time, pseudodeterministic algorithm}, if there is a probabilistic polynomial-time algorithm $A$ such that for every $x$, there exists $v \in f(x)$ such that $\Pr[A(x) = v] \geq 2/3$.
\end{definition}

Goldreich~\cite{Goldreich19} generalized the notion to $k$-pseudodeterministic algorithms.

\begin{definition}
Let $f$ be total multi-valued function. We say that $f$ admits a {\em polynomial-time, $k$-pseudodeterministic algorithm}, if there is a probabilistic polynomial-time algorithm $A$ such that for every $x$, there exists a set $S_x \subseteq f(x)$, and $\Pr[A(x) \in S] \geq \frac{k+1}{k+2}$.
\end{definition}

Suppose there is some function $f:\bitst\to\R^{d}$ (for some $d\in\N$) that admits a $(\epsilon, \delta)$ additive approximation algorithm $M$. I.e, 
\[\Pr[|M(x) - f(x)|_{\infty} \leq \epsilon] \geq 1-\delta\]
A natural question that arises is whether we can make such algorithms pseudodeterministic. Using reclusive partitions,  we show $M$ can be we can arrive a $(d+1)$ modified to $M'$ such that $M'$ is a $(d+1)$-pseudodeterministic algorithm (with the same probability guarantee of $1-\delta$) at a small loss on approximation guarantee. The algorithm $M'$ is describe below.

\begin{algorithm}
  \caption{$(d+1)$-Pseudodeterministically approximating $f\colon\bitst\to\R^{d}$}  \label{alg:pd}
  ~\\
  Let $d\in\N$, and $A$ be a $d\times d$ reclusive matrix, $\P_{A}$ its reclusive partition, and $\Delta_{A}$ its reclusive distance.\\
  Let $f\colon\bitst\to\R^{d}$ be any function, and let $\epsilon\in[0,\infty)$ and $\delta\in[0,1]$, and let $M$ be an efficient $(\epsilon,\delta)$-approximation algorithm for $f$ with respect to the $d_{max}$ metric (i.e. for all $x\in\bitst$, $\pr[d_{max}(A(x),f(x))\leq\epsilon]\geq1-\delta$).\\
  Let $\phi\colon\R^{d}\to\R^{d}$ be the bijection defined by $\phi(\vec{x})\defeq\dfrac{\Delta_{A}}{\epsilon}\cdot\vec{x}$.\\
  Let $M'$ be the following algorithm (which takes an input in $\bitst$).
  ~\\
  \begin{algorithmic}[1]
    \Procedure{M'(x)}{}
    \State $\vec{\alpha}\gets M(x)$\Comment{Run $M$ just once on the input}
    \State Determine the unique $X\in\P_{A}$ such that $\phi(\vec{\alpha})\in X$\Comment{This does not happen explicitly}
    \State $\vec{r}\gets\rep(X)$\Comment{Computed as in \Autoref{:efficent-computation-of-representatives}}
    \State\Return $\phi^{-1}(\vec{r})$\Comment{$\phi^{-1}$ is just multiplication}
    \EndProcedure
  \end{algorithmic}
\end{algorithm}

Note that in \Autoref{alg:pd}, we really do allow $\epsilon$ and $\delta$ to be as general as stated.
However, in practice, because $M$ must be an $(\epsilon,\delta)$-approximation algorithm, both $\epsilon$ and $\delta$ will be small.
The domain is chosen as $\bitst$ since this is the formal input of algorithms, but in practice the domain of the algorithm could be a different set.

\begin{proposition}[Algorithm Guarantees]\label{:algorithm-guarantees}
  \Autoref{alg:pd} is an efficient $(d+1)$-pseudodeterministic $(\epsilon',\delta)$-approximation algorithm for $f$ where $\epsilon'=\epsilon(\frac{1}{\Delta_{A}}+1)$.
\end{proposition}
\begin{proof}
  Let all notation be as in \Autoref{alg:pd}.
  Note that for any $\vec{a},\vec{b}\in\R^{d}$ we have $d_{max}(\vec{a},\vec{b})=\frac{\epsilon}{\Delta_{A}}d_{max}(\phi(\vec{a}),\phi(\vec{b}))$ via linearity of norms and the linearity of $\phi$:
  \begin{align*}
    d_{max}(\vec{a},\vec{b}) &= \norm{\vec{a}-\vec{b}}_{\infty}\\
                             &= \norm{\frac{\epsilon}{\Delta_{A}}\left(\frac{\Delta_{A}}{\epsilon}\vec{a}-\frac{\Delta_{A}}{\epsilon}\vec{b}\right)}_{\infty}\\
                             &= \frac{\epsilon}{\Delta_{A}}\norm{\phi(\vec{a})-\phi(\vec{b})}_{\infty}\\
                             &= \frac{\epsilon}{\Delta_{A}}d_{max}(\phi(\vec{a}),\phi(\vec{b}))
  \end{align*}

  By the defining property of $M$, for any $x\in\bitst$, with probability at least $1-\delta$, we have $d_{max}(\vec{\alpha},f(x))\leq\epsilon$ which, by the prior comment, holds if and only if $d_{max}(\phi(\vec{\alpha}),\phi(f(x)))\leq\Delta_{A}$.
  In this case we can bound the error of the approximation:
  \begin{align*}
    \text{error} &= d_{max}(\phi^{-1}(\vec{r},\;f(x))\\
                 &= \frac{\epsilon}{\Delta_{A}}d_{max}(\phi(\phi^{-1}(\vec{r})),\;\phi(f(x)))\tag{By the opening comment}\\
                 &= \frac{\epsilon}{\Delta_{A}}d_{max}(\vec{r},\;\phi(f(x)))\tag{$\phi\circ\phi^{-1}$ is identity map}\\
                 &\leq \frac{\epsilon}{\Delta_{A}}\left[d_{max}(\vec{r},\;\phi(\vec{\alpha})) + d_{max}(\phi(\vec{\alpha}),\;\phi(f(x)))\right]\tag{Triangle Inequality}\\
                 &\leq \frac{\epsilon}{\Delta_{A}}\left[d_{max}(\vec{r},\;\phi(\vec{\alpha})) + \Delta_{A} \right]\tag{By the second comment above}\\
                 &\leq \frac{\epsilon}{\Delta_{A}}\left[1 + \Delta_{A} \right]\tag{$\vec{r}\in X$ and $\phi(\vec{\alpha})\in X$, and $X$ is a unit hypercube}\\
                 &=\epsilon\left(\frac{1}{\Delta_{A}}+1\right)
  \end{align*}
  In other words, on any input $x\in\bitst$, with probability at least $(1-\delta)$, the value $\vec{r}$ returned by $M'$ will be within a distance $\epsilon'=\epsilon\left(\frac{1}{\Delta_{A}}+1\right)$ of the true function value $f(x)$.
  This proves that $M'$ is an $(\epsilon',\delta)$-approximation of $f$.

  That $M'$ is $(d+1)$-pseudodeterministic follows from the Partition Theorem (\Autoref{:partition-thm}).
  Using the notation $\vec{c}$, $N$, and $\mathcal{S}$ from that theorem, let $\vec{c}=\phi(f(x))$.
  As stated above, with probability at least $1-\delta$ we have $d_{max}(\phi(\vec{\alpha}),\phi(f(x)))\leq\Delta_{A}$, in which case $\phi(\vec{\alpha})\in N$, so the unique hypercube $X\in\P_{d}$ containing $\phi(\vec{\alpha})$ intersects $N$, and so $X\in\mathcal{S}$.
  Thus, with probability at least $1-\delta$, $\phi^{-1}\left(\vec{r}\right)=\phi^{-1}\left(\rep(X)\right)$ is one of at most $d+1$ values.
  This proves that $M'$ is $(d+1)$-pseudodeterministic\footnote{To meet the technical definition of $k$-pseudodeterminism of Goldreich, it should be required that $1-\delta\geq\frac{k+1}{k+2}$, so in this case with $k=d+1$, it should be required that $1-\delta\geq\frac{d+2}{d+3}$}.

  The efficiency of $M'$ follows from the efficiency of $M$ as well as \Autoref{:efficent-computation-of-representatives}.

\end{proof}

We show another application in the context of sample complexity of pseudodeterministic algorithms. Let $f_1, \ldots, f_d$ be functions from $\{0, 1\}^n$ to $[0, 1]$. 
Consider algorithms that have a blackbox access to these functions. I.e, the algorithm can generate a query $q$ ask for the values of $f_i(q)$ for some $1 \leq i \leq d$. The goal of the algorithm is to obtain a $(\epsilon, \delta)$-approximation to the vector $\langle\ex(f_{i})\rangle_{i=1}^{d}$. 
with respect to $d_{max}$ metric. This means that with probability at least $1-\delta$, the vector $\vec{\alpha}\in\R^{d}$ returned by the algorithm is such that $d_{max}(\vec{\alpha}, \langle\ex(f_{i})\rangle_{i=1}^{d})\leq\epsilon$ which is equivalent to saying that with probability at least $1-\delta$, it holds for all $i$ that vector $\abs{\alpha_{i}-\ex(f_{i})}\leq\epsilon$.

Goldreich proved that there is $(d+1)$-pseudodeterministic algorithm for this task, and this algorithm has a sample complexity of $\widetilde{O}(d^{4})$ . 
Using reclusive partitions, we show that the sample complexity can be improved to  $\widetilde{O}(d^{2})$ samples.

\RestatableApplicationToPseudodeterministicAlgorithms
\begin{proof}
  First define an algorithm $M$ as follows.
  $M$ selects uniformly at random $m=O\left(\frac{(d+1)^{2}}{\epsilon^{2}}\cdot\log\left(\frac{d}{\delta}\right)\right)$ points $x_{1},\ldots,x_{m}\in\nbits$, queries each function $f_{i}$ on all of those values and outputs the vector $\vec{v}$ where $v_{i}=\frac{1}{m}\sum_{j=1}^{m}f(x_{j})$ is the observed sample average.
  Then we have for each $i\in[d]$ that
  \[
    \pr\left[\abs{v_{i}-\ex f_{i}}>\frac{\epsilon}{d+1}\right]<\frac{\delta}{d}
  \]
  (this follows because $O(\frac{1}{\epsilon'^{2}}\cdot \log(1/\delta'))$ samples is sufficient for $(\epsilon', \delta')$-approximating the average of a single function $f:\nbits\to[0,1]$).
  Thus, by a union bound, we have
  \[
    \pr\left[\exists i\text{ s.t. }\abs{v_{i}-\ex f_{i}}>\frac{\epsilon}{d+1}\right]<\delta
  \]
  and taking the complement we have the desired result that
  \[
    \pr\left[d_{max}\left(\vec{v},\langle \ex f_{i}\rangle_{i=1}^{d}\right)\leq\frac{\epsilon}{d+1}\right]>1-\delta.
  \]
  Thus, $M$ will $(\frac{\epsilon}{d+1}, \delta)$-approximate $\langle\ex(f_{i})\rangle_{i=1}^{d}$ with respect to the $d_{max}$ metric.

  To complete the proof of the claim, let $A$ be a $d\times d$ reclusive matrix as in \Autoref{:d-reclusive} so that the reclusive distance is $\Delta_{A}=\frac{1}{d}$.
  Apply \Autoref{alg:pd} with $M$ and $A$ to obtain an algorithm $M'$ which $(d+1)$-pseudodeterministically $(\epsilon, \delta)$-approximates $\langle\ex(f_{i})\rangle_{i=1}^{d}$.
\end{proof}

Both Goldreich's algorithm and our own begin by taking samples to obtain an $(\epsilon',\delta)$-approximation of the averages, and then apply a rounding technique to add pseudodeterminism.
The difference is that Goldreich's rounding technique requires $\epsilon' = \frac{\epsilon}{d^{2}}$ whereas our algorithm only requires $\epsilon'=\frac{\epsilon}{d+1}$.
Goldreich's requirement for a better initial approximation requires more samples than ours does.
The reason for this could be that Goldreich's rounding technique is randomized, while ours is deterministic.



\section{Conclusions and Future Work}\label{sec:future}



One remaining open question from this work is to either prove \Autoref{:linear-conjecture} (possibly by improving the bound in \Autoref{:sperner-upper-bound}) or to offer a construction with better tolerance ($\epsilon$). In particular, improvements of \cite[Thm.~1]{loera_polytopal_2001} (along with the extra properties they proved that we discussed after our \Autoref{:sperner-number-defn}) would decrease our upper bound on the tolerance by increasing the lower bound on $\sperner(d)$ in \Autoref{:sperner-number-defn}. In fact, improvement of \cite[Thm.~1]{loera_polytopal_2001} just for the case that the polytope $P$ is a hypercube would yield improvements on our upper bounds for generic partitions, because we only every applied their results to hypercubes (because balls in the $d_{max}$ metric are hypercubes).
As a simplification of this open question, one might ask if the conjectured bound on the tolerance is at least optimal for partitions of unit hypercubes. 

A second open question that is of interest to us and is relevant to rounding schemes is the following: if we consider secluded partitions with degree $k=\poly(d)$ instead of exactly $k=d+1$, what improvements are possible for the tolerance $\epsilon$? For example, in the results of Hoza and Klivans, a degree of $\poly(d)$ would have been fine. Viewed another way, one might consider the best possible degree as a function of the tolerance---given some tolerance $\epsilon$, what is the smallest $k$ such that there exists a $(k,\epsilon)$-secluded partition with members of diameter at most $1$. A simple result is that for a tolerance of $\epsilon>\frac12$, $k\geq 2^d$ because for any point $\vec{p}$, the $2^d$ corners of $\clos{B}_\epsilon(\vec{p})$ are all $d_{max}$ distance more than $1$ apart, so because the partition has members of diameter at most $1$, each corner must belong to different member of the partition, so $\abs{\mathcal{N}_\epsilon(\vec{p})}\geq 2^d$. Thus, our results in this work show that for $\epsilon\leq\frac{1}{2d}$ one can achieve degree $k=d+1=\poly(d)$, and that is optimal, but for $\epsilon>\frac12$, $k\geq 2^d=\exponential(d)$, and we think it would be interesting to know the best value of $k$ for $\epsilon\in(\frac{1}{2d},\frac12)$ or at least know the largest $\epsilon$ such that $k=\poly(d)$.

A third area that is open is to better understand which lattice partitions are $(d+1,\epsilon)$-secluded (for some $\epsilon$). We initially defined the notion of a reclusive matrix and reclusive partition in order to construct some partition that was $(d+1,\frac{1}{O(d)})$-secluded, and we expected that our definition imposed far more requirements on the matrix than were necessary (i.e. that there would be lots of partitions that did not meet our definition, but had similar structural properties). However, we discussed in \Autoref{sec:fundamental-reclusive-property} our conceptual understanding of the construction of these partitions by successive extrusions and shifts and expected that this would be robust to re-ordering the entries in a row of the matrix; however, we then showed by a specific example that doing so resulted in partitions that were not $(d+1,\epsilon)$-secluded because there were points on the boundary of more than $d+1$ members of the partition for dimensions greater than $5$. This suggests to us that there is something deeper going on with these constructions than we initially realized, and it would be interesting to understand exactly what sets of basis vectors have the property that the integer linear combinations give the positions of hypercubes in a $(d+1,\epsilon)$-secluded partition (for some $\epsilon$).
\newpage
\appendix

\section{Rounding Schemes in Prior Work}
\label{sec:rounding-schemes-in-prior-work}

In this section, we will discuss in some detail how rounding is used in a number of publications and what properties of the rounding schemes are important in each of these papers. Not all of them benefit from our main construction and bounds, but we mention them nonetheless to highlight that there are a variety of perspectives one may reasonably take on what constitutes a good rounding scheme. Further, we think viewing each of these schemes as a partition (or distribution of partitions) highlights which publications have common goals in designing their rounding schemes. We begin by looking at a very simple rounding scheme, but though it is very simple, it shows up as a significant part of numerous publications. We have found that each of these publications independently walks through the construction, and we hope to demonstrate that each of these is doing the same thing under a different guise.

\subsection{A Very Simple Rounding Scheme}

Recalling that a deterministic rounding scheme for $\R$ is just a function $f:\R\to\R$, arguably, the most basic deterministic rounding scheme in $\R$ is the floor function, $\floor{\cdot}$, which maps every real number to the largest integer that is not larger that it\footnote{One could also consider the ceiling function, but floor tends to be used more often in practice as we shall see.}. If one considers the partition induced by this deterministic rounding scheme, it is the partition of half-open unit intervals $\P_{\floor{\cdot}}=\set{[n,n+1):n\in\Z}$. There are three simple modifications one might wish to make to this rounding scheme. 

First, one may want a ``scaled'' version. In the floor scheme, values might be rounded by as much as $1$, but one might wish to have values rounded by at most $\alpha$ for some $\alpha\in(0,\infty)$. This can be accomplished by a modified floor function $\floor{\cdot}_{\alpha}:\R\to\R$ defined by $\floor{x}_{\alpha}\defeq\alpha\floor{x/\alpha}$. This function maps every real number to the largest integer multiple of $\alpha$ that is not larger than it. The partition induced by $\floor{x}_{\alpha}$ is $\P_{\floor{\cdot}_{\alpha}}=\set{[\alpha n,\alpha(n+1)):n\in\Z}$.

The second modification that one might want is a ``shift'' of the floor scheme. For example, maybe it is desirable that $0.98$ and $1.213$ are rounded to the same value, and so one could (for example) choose the function $f:\R\to\R$ defined by $f(x)=\floor{x-0.3}$ so that $f(0.98)=f(1.213)=0$. More generally, one could pick any $\beta\in\R$ to shift by. This shift can be combined with a scaling $\alpha\in(0,\infty)$ to define the deterministic rounding scheme $\floor{\cdot}_{\alpha,\beta}$ given by $\floor{x}_{\alpha,\beta}\defeq\floor{x-\beta}_{\alpha}+\beta=\alpha\floor{(x-\beta)/\alpha)}$. The partition induced from this rounding scheme is $\P_{\floor{\cdot}_{\alpha,\beta}}=\set{[\alpha n+\beta,\alpha(n+1)+\beta):n\in\Z}$
\footnote{
  To see this, observe that $\P_{\floor{\cdot}_{\alpha,\beta}}$ is in fact a partition of $\R$ and that for any $n\in\Z$, if $x\in[\alpha n+\beta,\alpha(n+1)+\beta)$, then $(x-\beta)/\alpha\in[n,n+1)$ so $\floor{(x-\beta)/\alpha}=n$ so $\floor{x}_{\alpha,\beta}=\alpha n$. Thus all points in any member of $\P_{\floor{\cdot}_{\alpha,\beta}}$ map to the same value, and points in two different members map to different values (i.e. $n$ and $n'$).
}.
Note, that by this definition, the difference between $x$ and $\floor{x}_{\alpha,\beta}$ will become relatively large as $\beta$ is taken to be large, so typically $\beta$ will only take values in $[0,\alpha)$.

The third modification that one might want to make to the floor scheme is to have a ``different representative''. In the floor function, each value in the interval $[n,n+1)$ is mapped/rounded to $n$, but it might make sense to map/round these values to some other point in the interval such as the midpoint (or it might even be desirable to map/round them to a point not in the interval). We can combine this with the scaling and shifting. Let $\alpha$, $\beta$ as before and $\gamma\in\R$ (it will be typical that $\gamma$ is small, and to round to the midpoint we will let $\gamma=\alpha/2$). Define the deterministic rounding scheme $\floor{\cdot}_{\alpha,\beta,\gamma}$ by $\floor{x}_{\alpha,\beta,\gamma}\defeq\floor{x}_{\alpha,\beta}+\gamma=\alpha\floor{(x-\beta)/\alpha)}+\gamma$. The partition induced by this rounding scheme is the same as $\P_{\floor{\cdot}_{\alpha,\beta}}$ because changing the value assigned to each member does not change the member.

The shift modification discussed above is typically most useful when applied in the context of a randomized rounding scheme (a distribution of functions) rather than a deterministic rounding scheme (a single function). The idea is that it is often desirable that for any fixed pair of points $x,y\in\R$ which are ``sufficiently close'', then it holds with ``sufficiently high probability'' (over the selection of function $f$ from the distribution) that $f(x)=f(y)$. For example, fix some $\alpha\in(0,\infty)$ and consider the set of functions $\set{\floor{\cdot}_{\alpha,\beta}:\beta\in[0,\alpha)}$ with distribution corresponding to $\beta$ being distributed uniformly over $[0,\alpha)$. This gives a randomized rounding scheme with the following property: For any $\epsilon\in[0,\alpha)$, for any $x,y\in\R$ with $\abs{x-y}\leq\epsilon$, the probability that $\floor{x}_{\alpha,\beta}=\floor{y}_{\alpha,\beta}$ is greater than or equal to $\epsilon/\alpha$.

Intuitively this is because $x$ and $y$ end up in different members of the partition $\P_{\floor{\cdot}_{\alpha,\beta}}$ if and only if one of the boundaries of that partition separate $x$ and $y$ which happens with probability $\abs{x-y}/\alpha\leq \epsilon/\alpha$.

We view the randomized rounding scheme above (for any distribution of $\beta\in\R$) as a distribution of partitions of $\R$ by half-open $\alpha$-length intervals, and a value $x\in\R$ is randomly rounded by randomly obtaining a partition in the distribution, determining which member/interval of that partition contains $x$, and then returning the minimum value of that member/interval.

The ideas above easily generalizes to $\R^{d}$ for any $d\in\N$. One can view this generalization as being the above in each coordinate or (equivalently) as a vector version: for $\alpha\in(0,\infty)$ and a vector $\vec{\beta}=\langle{\beta_{i}}\rangle_{i=1}^{d}$, and a vector $\vec{\gamma}=\langle{\gamma_{i}}\rangle_{i=1}^{d}$ define $\floor{\cdot}_{\alpha,\vec{\beta},\vec{\gamma}}:\R^{d}\to\R^{d}$ coordinatewise in the expected way: $\floor{\vec{x}}_{\alpha,\vec{\beta},\vec{\gamma}}\defeq\langle\floor{x_{i}}_{\alpha, \beta_{i}, \gamma_{i}}\rangle_{i=1}^{d}$. If $\vec{\gamma}$ is not specified, it will be assumed to be $\vec{0}$. The partition induced by this scheme is $\P_{\floor{\cdot}_{\alpha,\vec{\beta},\vec{\gamma}}}=\set{\vec{\gamma}+[\alpha n,\alpha(n+1))^{d}:n\in\Z}$. In other words, the partition induced by rounding each coordinate is a grid of unit hypercubes with some shift applied to the grid.


\subsection{The Randomized Rounding Scheme of Saks and Zhou}
\label{sec:saks-zhou}

\newcommand{\opnorm}[1]{\norm{#1}_{\infty-\mathrm{op}}}
\newcommand{\entrynorm}[1]{\norm{#1}_{\infty-\mathrm{entry}}}

The rounding scheme used by Saks and Zhou is the basic rounding scheme just introduced\footnote{There is a small caveat that they consider only rounding matrices in $[0,1]^{d\times d}$ and requiring them to be rounded to a value in $[0,1]^{d\times d}$, but they just ensure everything is rounded down in each coordinate and then take $0$ if the value was negative.}. We briefly state the parameters of their scheme.

Let $d\in\N$ and $t=O(\log d)$ and $D=O(\log d)$. Let $\epsilon=d\cdot2^{-t}\cdot2^{-D}$. This will not be of much importance in this paper, but we want to highlight that this basic rounding scheme is used in multiple papers, so we briefly mention the parameters of the scheme for Saks and Zhou. Let $\alpha=2^{-t}$ so $\alpha=1/\poly(d)$. Let $S=\set{\frac{0}{2^{D}},\frac{1}{2^{D}},\ldots,\frac{2^{D}-2}{2^{D}},\frac{2^{D}-1}{2^{D}}}$ and let $\beta$ be uniformly distributed over $S$ and let $\vec{\beta}$ be the $d^{2}$ length vector in which every entry is $\beta$ (i.e. $\vec{\beta}=\langle \beta \rangle_{i=1}^{d^{2}}$). Let $\vec{n}\in\R^{d^2}$. Then with probability at least $1-\frac{O(d^3)}{2^{D}}$ (over the choice of $\beta$) it holds for all $\vec{m}\in B^{\circ}_{\epsilon}$ (w.r.t. the $l^{\infty}$ norm/$d_{max}$ metric) that $\floor{\vec{n}}_{\alpha,\vec{\beta}}=\floor{\vec{m}}_{\alpha,\vec{\beta}}$. In other words, with high probability, the entire $\epsilon$-ball of vectors around $\vec{n}$ are rounded to the same value. The reason is that for any coordinate $i\in[d]$, there are at most $O(d)$ values of $\beta\in S$ such that $\floor{n_i}_{\alpha, \beta}\not=\floor{m_i}_{\alpha, \beta}$ so the result holds by a union bound over the $d^2$ coordinates.

The notion of distance that Saks and Zhou were interested in, though, is not the $l^\infty$ norm, but the operator norm on matrices induced by the $l^{\infty}$ norm on vectors. This norm can be defined in either of these two well-known equivalent ways. Let $M$ be an $d\times d$ matrix:
\[
    \opnorm{M}\defeq\sup\set{\norm{M\vec{x}}_{\infty}:\vec{x}\in\R^{d}\text{ and }\norm{\vec{x}}_{\infty}=1}
\]
or
\[
    \opnorm{M}\defeq\max_{i\in[d]}\sum_{j=1}^{d}\abs{M_{i,j}}.
\]
If $M$ is just viewed as the obvious vector $\vec{m}$ in $\R^{d^2}$, then it is easy to see using the second definition above that
\[
  \norm{\vec{m}}_\infty\leq\opnorm{M}\leq d\norm{\vec{m}}_\infty.
\]
Thus, for any matrix $N$ it holds with the above probability that for all matrices $M$ within distance $\epsilon'=\cdot2^{-t}\cdot2^{-D}=\epsilon/d$ of $N$ w.r.t. the operator norm $M$ is rounded to the same value as $N$ (they are rounded as they would be if they were viewed as $d^2$ length vectors).

\subsection{The Randomized Rounding Scheme of Goldreich}

In \cite[Algorithm~9]{Goldreich19}, Goldreich uses the basic rounding scheme discussed above as well
\footnote{
  Goldreich uses $t$ to denote the dimension that we refer to as $d$, uses $\epsilon$ to denote what we call $\alpha$, and uses $\tau$ to denote what we call $\beta$. Further, Goldreich is is proving the property we are about to discuss in the context of learning the averages of $t$-many functions which is a detail showing up in the proof that is not needed for how we will state this property.
}.
However, unlike Saks and Zhou, Goldreich's goals in using the partition are very relevant to our work in this paper. For an arbitrary $\epsilon\in(0,\infty)$ let $\alpha=\epsilon$. Goldreich selects $\beta$ uniformly at random from the set $\set{-(j-0.5)\cdot\frac{\alpha}{10d^2}}$ and takes $\vec{\beta}$ to be the vector of length $d$ in which every coordinate is $\beta$ (i.e. $\vec{\beta}=\langle\beta\rangle_{i=1}^d$) and then applies the function $\floor{\cdot}_{\alpha,\vec{\beta}}:\R^{d}\to\R^{d}$. 

Goldreich shows that this randomized rounding scheme has the following property: For any point $\vec{x}\in\R^{d}$, there is a set $S_{\vec{x}}$ of cardinality at most $d+1$ such that with high probability (at least $1-\frac{1}{d+3}$ for $d>12$) over the choice of $\beta$, it will hold that
\[
\forall\vec{y}\in B_\epsilon(\vec{x}) :\floor{\vec{y}}_{\alpha,\vec{\beta}}\in S_{\vec{x}}
\]
where the $\epsilon$-ball is with respect to $l^\infty$/$d_{max}$. In other words, Goldreich shows that for these parameters of the basic grid rounding scheme, for any $\epsilon$-ball, there is a set of $d+1$ members of the induced partition, and it will hold with high probability that that ball intersects no member other than these. We emphasize the order of quantifiers---for any ball there is a high probability that this occurs, but there is $0$ probability that this occurs for all balls simultaneously because no matter which $\beta$ is chosen, the induced partition is a grid, so the $\epsilon$-ball at the corner of member will intersect $2^d$ different members.

Our work in this paper shows that this property that Goldreich desires can be achieved with a deterministic rounding scheme and that the randomness is not required. In other words, in \Autoref{sec:reclusive-lattice-partitions}, we construct a partition in each dimension $\R^{d}$ (which gives a deterministic rounding scheme) such that {\em every} ball of an appropriate radius $\epsilon$ intersects at most $d+1$ members of the partition.
\subsection{The Deterministic Rounding Scheme of Hoza and Klivans}
\label{sec:hoza-klivans}

In \cite[Section~2]{hoza_preserving_2018}, Hoza and Klivans have the same goal as Goldriech---ensuring that for any $\epsilon$-ball, there are very few values that all points in that ball are rounded to. However, Hoza and Klivans do this with a deterministic rounding scheme, and the induced partition of this rounding scheme has quite good parameters regarding our motivating question. The analysis of the rounding scheme in their paper is somewhat obscured by other technical aspects that were relevant to other ideas they were discussing but are not necessary for the analysis of partition. For this reason, we will present their scheme here doing our best to preserve the notation that they used (so one can compare our presentation with their paper if desired) while also casting it in a way that is consistent with the perspective we take; we will then prove that the induced partition of $\R^d$ can be scaled to a $(d+1, \frac{1}{6(d+1)})$-secluded partition with all members having diameter at most $1$.

Let $\epsilon\in(0,\infty)$ and $d\in\N$. Let $\mathcal{I}$ be a partition of $\R$ by intervals of length $2\epsilon(d+1)$ which are closed on the left and open on the right. Fix an arbitrary point $x\in\R$ and consider the interval $[x-\epsilon,x+\epsilon]=x+[-\epsilon,\epsilon]$. Because this interval has length $2\epsilon$ (and is closed) and every interval in $\mathcal{I}$ has length $2\epsilon(d+1)$ (and is half open), it follows that there is exactly one value $\Delta\in[d+1]$ such that the interval $[x+(2\epsilon\Delta)-\epsilon,x+(2\epsilon\Delta)+\epsilon]=x+2\epsilon\Delta+[-\epsilon,\epsilon]$ intersects two intervals in $\mathcal{I}$, and for every other $\Delta\in[d+1]$, this interval is a subset of some interval of $\mathcal{I}$ (which interval that is may depend on $\Delta$)
\footnote{
The sketch of the reason for this is that
\[
\bigcup_{\Delta\in[d+1]}[x+(2\epsilon\Delta)-\epsilon,x+(2\epsilon\Delta)+\epsilon]=[x+\epsilon,x+2(d+1)\epsilon+\epsilon]
\]
which is a closed interval of length $2\epsilon\Delta$ and thus intersects exactly 
two intervals of $\mathcal{I}$, say $I_{left}=[\alpha-2\epsilon(d+1),\alpha)$ 
and $I_{right}=[\alpha,\alpha+2\epsilon(d+1))$. 
The point $\alpha$ is either contained in the interior or right boundary of $[x+(2\epsilon\Delta_{0})-\epsilon,x+(2\epsilon\Delta_{0})+\epsilon]$ 
for some $\Delta_{0}$ (if not, then $\bigcup_{\Delta\in[d+1]}[x+(2\epsilon\Delta)-\epsilon,x+(2\epsilon\Delta)+\epsilon]$ 
would not intersect $I_{left}$). This is the unique $\Delta_{0}$ such 
that the interval intersects both $I_{left}$ and $I_{right}$.
}.

Now consider the partition $\mathcal{G}$ of $\R^{d}$ where each member is a $d$-fold product of intervals of $\mathcal{I}$. That is,
\[
\mathcal{G}\defeq\set{\prod_{i=1}^{d}I^{(i)}:I^{(1)},\ldots,I^{(d)}\in\mathcal{I}}.
\]
Each member of $\mathcal{G}$ is a hypercube, and up to translation, the set of centers of these hypercubes is $2\epsilon\Z^d$ (i.e. $\mathcal{G}$ should be interpreted as a grid of hypercubes). Let $\vec{1}$ denote the vector such that every entry is a $1$, and define $\Lambda=\set{(2\epsilon\Delta)\cdot\vec{1}:\Delta\in[d+1]}$. We claim that for any point $\vec{x}\in\R^{d}$, there exists at least one $\vec{\lambda}\in\Lambda$ such that $\clos{B}_{\epsilon}(\vec{x}+\vec{\lambda})$ is a subset of a member of $\mathcal{G}$ (intuitively, $\vec{x}$ can be shifted by one of these values, so that it is $\epsilon$-far into the interior of some member). This is because $\clos{B}_{\epsilon}(\vec{x}+\vec{\lambda})$ (which is a hypercube) is a subset of a member of $\mathcal{G}$ (all of which are hypercubes) if and only if for all $i\in[d]$ it holds that $[x_{i}+\lambda_{i}-\epsilon,x_{i}+\lambda_{i}+\epsilon]=[x_{i}+(2\epsilon\Delta)-\epsilon,x_{i}+(2\epsilon\Delta)+\epsilon]$ is a subset of some member of $\mathcal{G}$. By what we showed, for each coordinate $i\in[d]$, there is exactly one $\Delta\in[d+1]$ (and thus one $\vec{\lambda}\in\Lambda$) such that this does not hold in coordinate $i$, and so there are at most $d$-many $\vec{\lambda}$'s for which this does not hold on some coordinate. Thus, there must be at least one $\vec{\lambda}\in\Lambda$ (i.e. at least one $\Delta\in[d+1]$) for which the containment holds for all coordinates $i\in[d]$.

With these properties established, let $s:\R^{d}\to\Lambda$ be a function mapping each point $\vec{x}$ to one of the $\vec{\lambda}\in\Lambda$ that has the containment property above (e.g. take the smallest length $\vec{\lambda}$ that works). Also, define the representative function $\rep:\mathcal{G}\to\R^{d}$ so that $\rep(X)$ is the midpoint of the hypercube $X$. Then, the deterministic rounding scheme of Hoza and Klivans is the function $f:\R^{d}\to\R^{d}$ defined by $f(\vec{x})=\rep(\member_{\mathcal{G}}(\vec{x}+s(\vec{x}))$ (conceptually, $\vec{x}$ is rounded by first shifting $\vec{x}$ by some amount $\vec{\lambda}$ so that it is $\epsilon$-far in the interior of some member of the partition $\mathcal{G}$, and then returning the center point of that member.

The partition induced by the Hoza-Klivans rounding scheme in $\R^{2}$ is shown in \Autoref{fig:preserving_randomness_hoza-klivans_partition}.

\begin{figure}[h]
\centering
\begin{subfigure}{0.45\textwidth}
    \includegraphics[width=\textwidth]{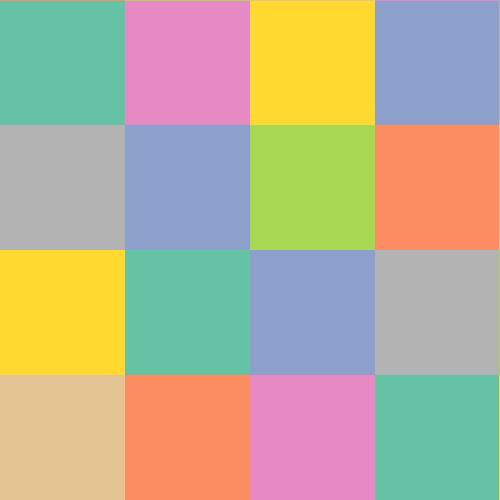}
    \caption{The grid partition $\mathcal{G}$ for $\R^{2}$.}
     \label{fig:preserving_randomness_grid_partition}
\end{subfigure}
\hfill
\begin{subfigure}{0.45\textwidth}
    \includegraphics[width=\textwidth]{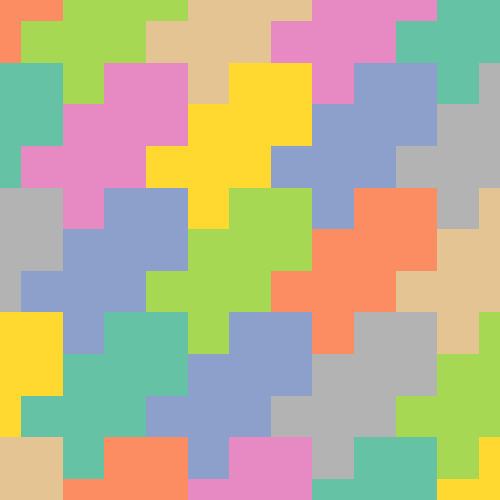}
    \caption{The Hoza-Klivans partition of $\R^{2}$.}
    \label{fig:preserving_randomness_hoza-klivans_partition}
\end{subfigure}
        
\caption{The partitions described in Section~\ref{sec:hoza-klivans}. $\mathcal{G}$ is given arbitrary colors for reference (there are duplicated colors because otherwise colors were hard to distinguish). The color of a member of the Hoza-Klivans partition indicates that all points in the member are rounded to the center of the cube in the grid partition $\mathcal{G}$ of the same color. Note that the Hoza-Klivans partition has the same grid structure as $\mathcal{G}$ despite the members no longer being cubes. Also note, that at most $3$ members of the Hoza-Klivans partition meet at a point.}
\label{fig:preserving_randomness}
\end{figure}

\begin{proposition}\label{:hoza-klivans-secluded}
The partition induced by the rounding scheme described above has the following properties:
\begin{itemize}
    \item Each member has diameter (in the $d_{max}$ metric) at most $6\epsilon(d+1)$
    \item The partition is $(d+1, \epsilon)$-secluded
\end{itemize}
\end{proposition}

If this partition is scaled by a factor of $6\epsilon(d+1)$, then it trivially becomes a partition in which all members have diameter at most $1$, and it is $(d+1,\frac{1}{6(d+1)})$-secluded. This result was stated in the paper and it restated here.
\RestatableHozaKlivansPartitionThm*

To prove this result, we will abstract this rounding scheme slightly and prove a corresponding version of the result so as to highlight the essential components of this deterministic rounding scheme if one wished to generalize it. In the statement of the following lemma, the only notation change is that $\P_{0}$ can be interpretted as indicting $\mathcal{G}$.

\begin{lemma}
 Let $d\in\N$ and $\epsilon,D\in(0,\infty)$.
 Let $\P_{0}$ be a partition of $\R^{d}$ such that all members have diameter at most $D$. Let $\rep:\P_{0}\to\R^{d}$ be a function such that $\rep(X)\in X$ (conceptually, this function defines a unique representative for each member of the partition). 
 Let $\Lambda$ be a finite set of vectors in $\R^{d}$ (conceptually a finite set of possible shifts). 
 Let $s:\R^{d}\to\Lambda$ be a function such that $\clos{B}_{\epsilon}(\vec{x}+s(\vec{x}))$ is a subset of some member
 \footnote{
   In particular, because $\vec{x}+s(\vec{x})\in\clos{B}_{\epsilon}(\vec{x}+s(\vec{x}))$ the member of $\P_{0}$ containing this ball must also contain $\vec{x}+s(\vec{x})$, and so this member must be $\member_{\P_{0}}(\vec{x}+s(\vec{x}))$.
 } 
 of $\P_{0}$. 
 Let $\ell=\max_{\vec{\lambda}\in\Lambda}\norm{\lambda}_\infty$ (the maximum length of a shift).
 Let $f:\R^{d}\to\R^{d}$ denote the function (i.e. deterministic rounding scheme) defined by $f(\vec{x})=\rep(\member_{\P_{0}}(\vec{x}+s(\vec{x}))))$. 
 
 Then the partition induced by the deterministic rounding scheme $f$ is $(\abs{\Lambda}, \epsilon)$-secluded and has members of diameter at most $2\ell+D$.
\end{lemma}

Once this is proven, \Autoref{:hoza-klivans-secluded} follows as an immediate corollary since in the initial partition $\P_{0}=\mathcal{G}$, all members have diameter $D=2\epsilon(d+1)$, and the longest vector $\vec{\lambda}\in\Lambda$ has length $\ell=2\epsilon(d+1)$ (in the $l^{\infty}$ norm), and $\abs{\Lambda}=d+1$.

\begin{proof}
Let $\P$ denote the partition induced by the deterministic rounding scheme $f$. We first show that $\P$ is $(\abs{\Lambda}, \epsilon)$-secluded (i.e. that for any $\vec{p}\in\R^{d}$ it holds that $\abs{\mathcal{N}_{\epsilon}(\vec{p})}\leq\abs{\Lambda}$)
\footnote{
  The neighborhood notation $\mathcal{N}_{\epsilon}(\vec{p})$ throughout this proof is always relative to the partition $\P$ and never the partition $\P_{0}$.
}.
Let $\vec{p}\in\R^{d}$ be arbitrary.
For any $\vec{x}\in\clos{B}_{\epsilon}(\vec{p})$, since the $d_{max}$ metric arises from a norm, it follows that $\vec{p}+s(\vec{x})\in\clos{B}_{\epsilon}(\vec{x}+s(\vec{x}))$
\footnote{
    \begin{align*}
      \vec{x}\in\clos{B}_{\epsilon}(\vec{p}) &\iff d_{max}(\vec{x},\vec{p})\leq\epsilon\\
                                             &\iff \norm{\vec{x}-\vec{p}}_{\infty}\leq\epsilon\\
                                             &\iff \norm{(\vec{x}+s(\vec{x}))-(\vec{p}+s(\vec{x}))}_{\infty}\leq\epsilon\\
                                             &\iff \vec{p}+s(\vec{x})\in\clos{B}_{\epsilon}(\vec{x}+s(\vec{x}))
    \end{align*}
}.
Since $\clos{B}_{\epsilon}(\vec{x}+s(\vec{x}))\subseteq\member(\vec{x}+s(\vec{x}))$ (by our requirements on $s$), it then follows that $\vec{p}+s(\vec{x})\in\member(\vec{x}+s(\vec{x}))$ and so $\member(\vec{p}+s(\vec{x}))=\member(\vec{x}+s(\vec{x}))$. This allows us to show as follows that $f$ takes on at most $\abs{\Lambda}$ values on the set $\clos{B}_{\epsilon}(\vec{p})$:
\begin{align*}
    \set{f(\vec{x}): \vec{x}\in\clos{B}_{\epsilon}(\vec{p})} &= \set{\rep(\member_{\P_{0}}(\vec{x}+s(\vec{x}))): \vec{x}\in\clos{B}_{\epsilon}(\vec{p})} \tag{Def'n of $f$} \\
    &= \set{\rep(\member_{\P_{0}}(\vec{p}+s(\vec{x}))): \vec{x}\in\clos{B}_{\epsilon}(\vec{p})} \tag{$\member_{\P_{0}}(\vec{p}+s(\vec{x}))=\member_{\P_{0}}(\vec{x}+s(\vec{x}))$} \\
    &\subseteq \set{\rep(\member_{\P_{0}}(\vec{p}+\vec{\lambda})): \vec{\lambda}\in\Lambda} \tag{$s(\vec{x})\in\Lambda$}
\end{align*}
The latter set clearly has cardinality at most $\abs{\Lambda}$ because $\rep(\member(\vec{p}+\vec{\lambda}))$ is a mapping of the elements of $\Lambda$. This is morally why the induced partition $\P$ has the property $\abs{\mathcal{N}_{\epsilon}(\vec{p})}\leq\abs{\Lambda}$; the following formalizes this, but the intuition of the above is somewhat lost in the notation.

\begin{align*}
    \abs{\mathcal{N}_{\epsilon}(\vec{p})} &= \abs{\set{\member_{\P}(\vec{x}):\vec{x}\in\clos{B}_{\epsilon}(\vec{p})}} \tag{\Autoref{:alternate-neighborhood-defn}} \\
    &= \abs{\set{f^{-1}(f(\vec{x})):\vec{x}\in\clos{B}_{\epsilon}(\vec{p})}} \tag{Def'n of the induced partition $\P$} \\
    &= \abs{\set{f^{-1}(\vec{w}):\vec{w}\in\set{f(\vec{x}):\vec{x}\in\clos{B}_{\epsilon}(\vec{p})}}} \tag{Reformat} \\
    &\leq \abs{\set{f^{-1}(\vec{w}):\vec{w}\in\set{\rep(\member_{\P_{0}}(\vec{p}+\vec{\lambda})): \vec{\lambda}\in\Lambda}}} \tag{Prior paragraph} \\
    &\leq \abs{\set{\rep(\member_{\P_{0}}(\vec{p}+\vec{\lambda})): \vec{\lambda}\in\Lambda}} \tag{$f^{-1}$ is a function}  \\
    &\leq \abs{\Lambda} \tag{Prior paragraph}
\end{align*}

We next show that every member of $\P$ has diameter at most $2\ell+D$.
Let $X\in\P$ be arbitrary and let $\vec{x},\vec{y}\in X$. This means that $f(\vec{x})=f(\vec{y})$ so by definition of $f$, this means $\rep(\member_{\P_{0}}(\vec{x}+s(\vec{x}))=\rep(\member_{\P_{0}}(\vec{y}+s(\vec{y}))$. By definition of $\rep$, the left hand side is contained in $\member_{\P_{0}}(\vec{x}+s(\vec{x}))$ and the right hand side is contained in $\member_{\P_{0}}(\vec{y}+s(\vec{y}))$, and since the left and right hand side are the same point, it must be that $\member_{\P_{0}}(\vec{x}+s(\vec{x}))=\member_{\P_{0}}(\vec{y}+s(\vec{y}))$, and because members of $\P_{0}$ have diameter at most $D$, it follows that $d_{max}(\vec{x}+s(\vec{x}), \vec{y}+s(\vec{y}))\leq D$. Now observe that
\begin{align*}
    d_{max}(\vec{x},\vec{y}) &\leq d_{max}(\vec{x}, \vec{x}+s(\vec{x})) + d_{max}(\vec{x}+s(\vec{x}), \vec{y}+s(\vec{y})) + d_{max}(\vec{x}, \vec{x}+s(\vec{x})) \\
    &\leq \norm{s(\vec{x})}_{\infty} + D + \norm{s(\vec{y})}_{\infty} \\
    &\leq \ell + D + \ell \tag{$s(\vec{x}),s(\vec{y})\in\Lambda$}
\end{align*}
so $\diam(X)\leq2\ell+D$.
\end{proof}





\section{Measure Theory}
\label{sec:measure-theory}

Throughout this section, by ``countable'' we mean finite or countably infinite.

\renewcommand{\A}{\mathcal{A}}
\newcommand{\F}{\mathcal{F}}
\renewcommand{\B}{\mathcal{B}}
\renewcommand{\P}{\mathcal{P}}
\begin{fact}\label{:disjoint-uncountable}
  If $m$ is a measure and $\A$ is a (possibly uncountable) family of pairwise disjoint measurable sets, then
  \[
    m(\bigsqcup_{A\in\A}A)\geq\sum_{A\in\A}m(A).
  \]
\end{fact}
\begin{proof}
  By definition of the arbitrary summation (c.f. \cite[p.~11]{folland_real_1999}) we have
  \[
    \sum_{A\in\A}m(A)\defeq\sup\set{\sum_{A\in\F}:\F\subseteq\A,\;\F\text{ finite}}
  \]
  and for any $\F\subseteq\A$ we have
  \[
    m(\bigsqcup_{A\in\A}A)\geq m(\bigsqcup_{A\in\F}A)=\sum_{A\in\F}m(A).
  \]
  Thus $m(\bigsqcup_{A\in\A}A)$ is an upper bound for the set $\set{\sum_{A\in\F}:\F\subseteq\A,\;\F\text{ finite}}$ and thus greater than or equal to the supremum.
\end{proof}

\begin{fact}
  If $m$ is a measure and $\A$ is a (possibly uncountable) family of pairwise disjoint measurable sets and $m(\bigsqcup_{A\in\A}A)<\infty$, then the set $\set{A\in\A:m(A)>0}$ is countable.
\end{fact}
\begin{proof}
  Let $\B=\set{A\in\A:m(A)>0}$ denote the set in question, and let $\B_{n}=\set{A\in\A:m(A)>\frac1n}$ so that $\B=\bigcup_{n=1}^{\infty}\B_{n}$. Clearly each $\B_{n}$ is finite since
  \[
    \infty > m(\bigsqcup_{A\in\A}A) \geq m(\bigsqcup_{A\in\B_{n}}A) \geq \sum_{A\in\B_{n}}m(A) \geq \sum_{A\in\B_{n}}\frac1n
  \]
  and $n$ is independent of the summation.

  Thus $\B$ is a countable union of finite families, so $\B$ is countable.
\end{proof}

\begin{fact}
  If $\P$ is a partition of $\R^{d}$, and $m$ is the Lebesgue measure on $R^{d}$, and for all $X\in\P$, $X$ is measurable and $m(x)>0$, then $\P$ is countable.
\end{fact}
\begin{proof}
  We first show that for any $n\in\N$, the set $\A_{n}=\set{X\cap B_{n}(\vec{0}):X\in\P,\;m(X\cap B_{n}(\vec{0}))>0}$ is countable. Observe that $\A_{n}$ is pairwise disjoint and $\bigsqcup_{A\in\A}A\subseteq B_{n}(\vec{0})$ so $\infty>m(B_{n}(\vec{0}))\geq m(\bigsqcup_{A\in\A}A)$, so by the previous result, $\A_{n}$ is countable. Observe that $A_{n}$ has the same cardinality as $\P_{n}=\set{X\in\P:m(S\cap B_{n}(\vec{0}))>0}$ (it is easy to inject $\P_{n}$ into $\A_{n}$ via intersection with $B_{n}(\vec{0})$, and it is easy to inject $\A_{n}$ into $\P_{n}$ by mapping $A$ to the unique member of $\P_{n}$ containing $A$). Thus $\P_{n}$ is countable.

  Clearly $\P\subseteq\bigcup_{n=1}^{\infty}\P_{n}$, and we also get the other inclusion because for any $X\in\P$ there is some $n\in\N$ such that $m(X\cap B_{n}(\vec{0}))>0$ (since $0<m(X)=m(\bigcup_{n=1}^{\infty}(X\cap B_{n}(\vec{0})))\leq\sum_{n=1}^{\infty}m(X\cap B_{n}(\vec{0}))$ so some term on the right must be positive). Thus $\P$ is a countable union of countable families, so $\P$ is countable.
\end{proof}

Note that the above proof can be easily generalized from $\R^{d}$ to any (non-empty) $\sigma$-finite measure space by replacing the $B_{n}(\vec{0})$ with a $\sigma$-decomposition of the space.

\section{Binary Relations}
\label{sec:binary-relations}

Let $R$ denote a binary relation on a set $X$, and let $R^{t}$ denote the transitive closure and let $R^{-1}$ denote the inverse relation $R^{-1}=\set{(a,b)\in X^{2}:(b,a)\in R}$. The following are easily verified:
\begin{itemize}
\item $R$ is symmetric if and only if $R=R^{-1}$.
\item $R$ is transitive if and only if $R^{-1}$ is transitive.
\item For another binary relation $S$, we have $R\subseteq S$ if and only if $R^{-1}\subseteq S^{-1}$.
\item If $\mathcal{S}$ is a collection of relations, then $\bigcap_{S\in\mathcal{S}}S^{-1}=\left(\bigcap_{S\in\mathcal{S}}S\right)^{-1}$.
\end{itemize}
From these it follows that $(R^{-1})^{t}=(R^{t})^{-1}$ as shown below.
\begin{align*}
  (R^{-1})^{t} &= \bigcap_{\substack{S\subset X^{2}\\S\text{ transitive}\\S\supseteq R^{-1}}}S\tag{Common alternate definition of transitive closure}\\
               &= \bigcap_{\substack{S\subset X^{2}\\S^{-1}\text{ transitive}\\S^{-1}\supseteq R}}S\tag{Inverse preserves transitivity and subsets}\\
               &= \bigcap_{\substack{T\subset X^{2}\\T\text{ transitive}\\T\supseteq R}}T^{-1}\\
               &= \Big(\bigcap_{\substack{T\subset X^{2}\\T\text{ transitive}\\T\supseteq R}}T\Big)^{-1}\tag{Inverse preserves intersections}\\
                 &= (R^{t})^{-1}\tag{Common alternate definition of transitive closure}
\end{align*}

\begin{fact}
  If $R$ is a reflexive and symmetric relation on $X$, then $R^{t}$ is an equivalence relation.
\end{fact}
\begin{proof}
  Since $R$ is reflexive, we have that for all $a\in X$, $(a,a)\in R\subseteq R^{t}$, so $R^{t}$ is reflexive. Since $R$ is symmetric, we have that $R=R^{-1}$, so $R^{t}=(R^{-1})^{t}=(R^{t})^{-1}$ which impies that $R^{t}$ is symmetric since it is equal to its inverse. That $R^{t}$ is transitive follows from the definition of transitive closure. Thus $R^{t}$ is an equivalence relation.
\end{proof}

\begin{fact}
  If $R,S$ are equivalence relations on $X$, and $R\subseteq S$, then each equivalence class of $R$ is a subset of some equivalence class of $S$.
\end{fact}
\begin{proof}
  Let $E_{R}$ denote an arbitrary equivalence class of $R$. Then $E_{R}$ contains some $x\in X$, and we denote $E_{R}$ using the standard notation $[x]_{R}$ which is the equivalence class containing $x$. We will show that $[x]_{R}$ is a subset of $[x]_{S}$. Let $y\in[x]_{R}$ be arbitrary. Then $(x,y)\in R\subseteq S$ which implies $(x,y)\in S$ and thus $y\in[x]_{S}$.
\end{proof}

Still letting $R$ denote a binary relation on $X$, let $R^{t_{0}}=R$, and inductively for all $n\in\N^{+}$, let $R^{t_{n}}=R^{t_{n-1}}\cup\set{(x,y)\in X^{2}:\text{ $\exists z\in X$ with $(x,z)\in R^{t_{n-1}}$ and $(z,y)\in R^{t_{n-1}}$}}$.

\begin{fact}
  If $R$ is a binary relation on $X$, then $R^{t}=\bigcup_{n=0}^{\infty}R^{t_{n}}$.
\end{fact}
\begin{proof}
  To show that $R^{t}\subseteq\bigcup_{n=0}^{\infty}R^{t_{n}}$ it suffices to show that $\bigcup_{n=0}^{\infty}R^{t_{n}}$ is transitive. First note that for any $n\in\N_{0}$, $R^{t_{n}}\subseteq R^{t_{n+1}}$. Let $(a,b),(b,c)\in\bigcup_{n=0}^{\infty}R^{t_{n}}$; then there is some $N$ such that $(a,b),(b,c)\in R^{t_{N}}$ which means that $(a,c)\in R^{t_{N+1}}$ and so $\bigcup_{n=0}^{\infty}R^{t_{n}}$ is transitive.

  For the other containment, for an inductive base case note that $R^{t_{0}}=R\subseteq R^{t}$. Then for the inductive case, if $R^{t_{n}}\subseteq R^{t}$ for some $n$, then because $R^{t}$ is transitive it follows that 
 
  \[\set{(x,y)\in X^{2}:\text{ $\exists z\in X$ with $(x,z)\in R^{t_{n}}$ and $(z,y)\in R^{t_{n}}$}}\subseteq R^{t}\] and thus $R^{t_{n+1}}\subseteq R^{t}$. Thus $\bigcup_{n=0}^{\infty}R^{t_{n}}\subseteq R^{t}$.
\end{proof}


\begin{fact}
  Let $(a,b)\in X^{2}$, then $(a,b)\in R^{t_{n}}$ if and only if there exists $0<k\leq 2^{n}$ and there exists a sequence $\langle{x_{i}}_{i=0}^{k}\rangle$ with $x_{0}=a$ and $x_{k}=b$ and for all $i\in[k]$, $(x_{i-1},x_{i})\in R$.
\end{fact}
\begin{proof}
  The case $n=0$ is trivial and serves as an inductive base case. For induction, assume the statement for $n-1$. For the forward direction, if $(a,b)\in R^{t_{n}}$ then either $(a,b)\in R^{t_{n-1}}$ and the required sequence exists by IH, or $(a,b)\in\set{(x,y)\in X^{2}:\text{ $\exists z\in X$ with $(x,z)\in R^{t_{n-1}}$ and $(z,y)\in R^{t_{n-1}}$}}$ and thus there exists $c\in X$ such that $(a,c),(c,b)\in R^{t_{n-1}}$ so by IH, there exists $0<k',k''\leq 2^{n-1}$ and sequences $\langle x_{i} \rangle_{i=0}^{k'}$ and $\langle y_{0} \rangle_{i=1}^{k''}$ with $x_{0}=a$, $x_{k'}=c=y_{0}$, and $y_{k''}=b$, and thus pasting the sequences together as $\langle z_{i} \rangle_{i=0}^{k'+k''}$ with $z_{i}=
  \begin{cases}
    x_{i} & 0\leq i\leq k'\\
    y_{i-k'} & k'\leq i\leq k'+k''
  \end{cases}
$
is a sequence with $0<k=k'+k''\leq 2^{n}$ and $z_{i}=a$ and $z_{k=k'+k''}=b$.

For the reverse direction, if a sequence $\langle x_{i} \rangle_{i=0}^{k}$ exists with $x_{0}=a$, $x_{k}=b$, and $0<k\leq 2^{n}$, then either $k=1$ and we are done (because then $(a,b)\in R$) or $k>1$ in which case we let $k'=\ceil{k/2}$ and $k''=\floor{k/2}$ so that $k'+k''=k$ and $0<k',k''\leq 2^{n-1}$, so by inductive hypothesis, the sequence $\langle x_{i} \rangle_{i=0}^{k'}$ demonstrates that $(x_{0},x_{k'})\in R^{t_{n-1}}$ and the sequence $\langle y_{i}\rangle_{i=k'+0}^{k'+k''=k}$ demonstrates that $(x_{k'},x_{k})\in R^{t_{n-1}}$ and thus $(a,b)=(x_{0},x_{k})\in R^{t_{n}}$.
\end{proof}

\begin{fact}
  For any $a,b\in X$, $(a,b)\in R^{t}$ if and only if there exists some $N\in\N$ and some sequence $\langle x_{i}\rangle_{i=0}^{N}$ with $x_{0}=a$, and $x_{N}=b$, and for all $i\in[N]$ $(x_{i-1},x_{i})\in R$.
\end{fact}
\begin{proof}
  If $(a,b)\in R^{t}$, then $(a,b)\in R^{t_{n}}$ for some $n\in\N\cup\set{0}$, so by the prior fact there exists some $0<N\leq2^{n}$ for which a sequence as described exists. Conversely, if such a sequence $\langle x_{i}\rangle_{i=0}^{N}$ exists then $(a,b)\in R^{t_{n}}\subseteq R^{t}$.
\end{proof}

\bibliographystyle{IEEEtran}
\bibliography{IEEEabrv,references,ref}
\end{document}